\documentclass[a4paper,12pt,onehalfspacing,twoside,openright,tools]{report}
\usepackage{breakcites}
\usepackage{longtable}
\usepackage[T1]{fontenc}
\usepackage[utf8]{inputenc}
\usepackage{babel}
\usepackage{graphicx}
\usepackage{epstopdf}
\usepackage{amsfonts}
\usepackage{amsmath}
\usepackage{amssymb}
\usepackage[ruled,vlined,titlenumbered]{algorithm2e}
\usepackage{fancyhdr}
\usepackage[usenames]{color}

\newtheorem{theorem}{Theorem}
\newtheorem{remark}{Remark}
\newtheorem{definition}{Definition}
\newtheorem{lemma}{Lemma}
\newtheorem{proposition}{Proposition}
\newtheorem{example}{Example}

\newenvironment{proof}[1][Proof]{\noindent\textbf{#1. }\it}{\mbox{}\\\\}

\renewcommand{\headrulewidth}{0.5pt}
\renewcommand{\footrulewidth}{0pt} 
\fancypagestyle{plain}{ 
	\fancyhead{} 
	\renewcommand{\headrulewidth}{0pt} 
}
\usepackage{url}
\usepackage{geometry}
\usepackage{pdflscape}
\usepackage{rotating}
\usepackage{float,lscape}
\usepackage{booktabs}

\newcommand{\twlrm}{\fontsize{8.5}{8.3pt}\normalfont\rmfamily}

\begin{document}
	\begin{center}
		\thispagestyle{empty}
		\vspace{-3cm}
		\center{\large{\textsc{Universit\'e de Tunis El-Manar}}}
		\vspace{-0.2cm}
		\center{\large{\textsc{Facult\'e des Sciences de Tunis}}}
		\begin{figure} [H]
			\begin{center}
				\includegraphics [width=2.5cm]{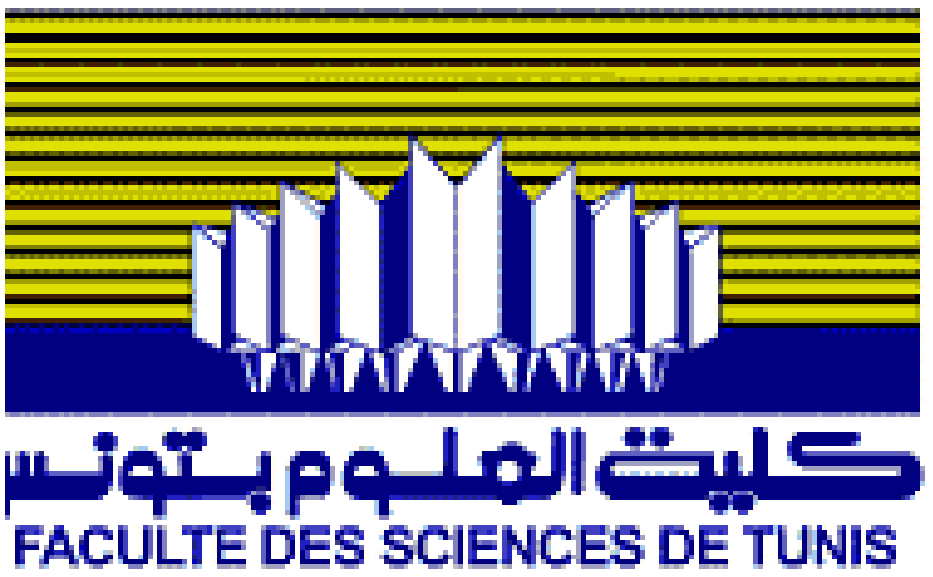}
				
			\end{center}
		\end{figure}
		\vspace{-0.5cm}
		
		{\huge{\textsc{\textbf{TH\`ESE}}}}
		
		\center{Pr\'esent\'ee en vue de l'Obtention du Dipl\^{o}me de}
		\vspace{-0.3cm}
		\center{\textsc{Doctorat en Informatique}} 
		\vspace{0.5cm}
		\center{Par:}
		\center{\large{\textsc{\textbf{Souad BOUASKER}}}}
		
		\vspace{1cm}
		\center{\textbf{\large{\textbf{Caract\'erisation et Extraction des repr\'esentations concises des Motifs Corr\'el\'es bas\'ee sur l'Analyse Formelle de Concepts}}}}  
		\vspace{1cm}
		\\ \par
		\hrule height 4pt
		\par
		\vspace{1cm}

		\vspace{-1.0cm}
		\begin{center}
			\textsc{\textbf{Comit\'e de Th\`ese}}
		\end{center}
		\begin{table} [H]
			\renewcommand{\footnoterule}{} 
			\renewcommand{\arraystretch}{1}
			\setlength\tabcolsep{5pt}
			\vspace{-1.0cm}
			\hspace{-1.0cm}
			\begin{tabular}{lll}
				\\ \textsc{Faouzi MOUSSA}  & \textsc{Professeur, Facult\'e des Sciences de Tunis} & \textsc{Pr\'esident}
				\\ \textsc{Nadia ESSOUSSI}  & \textsc{Maitre de Conf\'erences, F.S.E.G de Nabeul} & \textsc{Rapporteur}
				\\ \textsc{Philippe LENCA}  & \textsc{Professeur, Telecom Bretagne} & \textsc{Rapporteur}
				\\ \textsc{Amel TOUZI GRISSA}  & \textsc{Professeur, ESIG de Kairouan} & \textsc{Examinateur}
				\\ \textsc{Sadok BEN YAHIA}  & \textsc{Professeur, Facult\'e des Sciences de Tunis} & \textsc{Directeur}\\
			\end{tabular}
		\end{table}
		
		\vspace{1cm}
		\par
		\hrule height 4pt
		\par
		\vspace{1cm}
		\center{\footnotesize{2 Novembre 2016} \\}
		\vspace{0.5cm}
		\footnotesize{\textbf{L}aboratoire \textbf{I}nformatique de \textbf{P}rogrammation, \textbf{A}lgorithmique et \textbf{H}euristiques \textsc{LIPAH}}
	\end{center}

	\newpage
	\thispagestyle{empty}
	
	\begin{center}
		\Large{\textsc{\textbf{Acknowledgments}}}
	\end{center}
	
	I would like to thank the members of my thesis committee. I thank Professor \textsc{Faouzi MOUSSA} for agreeing to chair my thesis committee.

	\bigskip

	I would like to thank Associate-Professor \textsc{Nadia ESSOUSSI} and Professor \textsc{Philippe LENCA} for accepting to review my thesis report, and for providing me with detailed corrections and interesting comments. 
	
	\bigskip
	
	I would like to thank  Professor \textsc{Amel GRISSA TOUZI} for participating  to the thesis committee. 
	
	\bigskip
	
	I want to express my deep thanks to my thesis supervisor Professor \textsc{Sadok BEN YAHIA} for trusting me and for allowing me to grow as a research scientist. During the whole period of study, Professor \textsc{BEN YAHIA} contributes by giving me intellectual freedom in my work, engaging me in new creative ideas, supporting my participation to various conference, and requiring a high quality of work in all my efforts. 
	
	\bigskip
	
	I want to express my special thanks to Assistant-Professor \textsc{Tarek HAMROUNI} for collaborating in the realization of different phases of my thesis project. I am very grateful for all the offered efforts to ensure high-quality of this research. I greatly benefited from his scientific insight, his high-level of expertise in the field of Data-Mining and his ability to explore possible improvements in order to make a deeper development of our research. 
	
	\bigskip
	
	My sincere thanks go to all the members of the \textsc{LIPAH} Laboratory of the Faculty of Sciences of Tunis, for the friendship, for the encouraging ambiance and emotional atmosphere during the last years.
	
	\bigskip
	
	I cannot finish without thanking my family. A special dedicate goes to my precious treasure, my mother \textsc{Radhia BENFRADJ BOUASKER}, for supporting and encouraging me during my studies.   
	I also would like to thank my brothers, my sisters for providing assistance in numerous ways. I want to express my gratitude to my husband,  dear \textsc{Mohamed}, without his comprehension and encouragements, I could not have accomplished this project. A particular thought goes for my lovely baby girl \textsc{Meriam}, my angel baby, the greatest joy of my life.

	\newpage

	\begin{center}
		
		My Doctoral Graduation is dedicated to the memory of my beloved father, \textsc{Miled BOUASKER}. I am honored to have you as a father. Thank you for the high trust, for learning to me the strength and the patience, and for motivating me to always keep reaching for excellence.\\
		Thank you for the father you were.

	\end{center}


	\begin{abstract}
		
		Correlated pattern mining has increasingly become  an important task in data mining since these patterns allow conveying knowledge about meaningful and surprising relations among data.
		Frequent correlated patterns were thoroughly  studied in the literature. 
		
		In this thesis, we propose to benefit from both frequent correlated as well as rare correlated patterns according to the \textit{bond} correlation measure. Nevertheless, a main moan addressed to correlated pattern extraction approaches is their high number which handicap their extensive utilizations. In order to overcome this limit, we propose to 
		extract a subset without information loss of the sets of frequent correlated and of rare correlated patterns, this subset is called ``Condensed Representation``. In this regard, we are based on the notions derived from the Formal Concept Analysis FCA, specifically the equivalence classes associated to a closure operator $f_{bond}$ dedicated to the \textit{bond} measure, to introduce new concise representations of both frequent correlated and rare correlated patterns. We then design 
		the new mining approach, called \textsc{Gmjp}, allowing the extraction of the sets of frequent correlated patterns, of rare correlated patterns and their associated concise representations.  
		In addition, we present the \textsc{Regenerate} algorithm allowing the query of the $\mathcal{RCPR}$ condensed representation associated to the $\mathcal{RCP}$ set as well as the  \textsc{RcpRegeneration} algorithm dedicated to the regeneration of the whole set of rare correlated patterns from the $\mathcal{RCPR}$ representation.  
		
		The carried out experimental studies highlight the 
		very encouraging compactness rates offered by the proposed concise representations and prove the good performance of the \textsc{Gmjp} algorithm. To improve the obtained performance, we introduced and evaluated the optimized version of \textsc{Gmjp}. The latter shows much better performances than do the initial version of \textsc{Gmjp}. In order to prove the usefulness of the extracted condensed representation, we conduct a classification process based on correlated association rules derived from closed correlated patterns and their associated minimal generators. The obtained rules were applied to the context of intrusion detection and 
		achieve encouraging results.
		
		\bigskip
		
		\hspace{+0.5cm}\textbf{Key Words:} {Formal Concept Analysis, Constraint Data Mining, Monotonicity, Anti-monotonicity, \textit{bond} Correlation Measure, Itemset Extraction, Condensed Representation, Classification, Associative Rule.}
	\end{abstract}

	\mbox{}
	
	\newpage
	\thispagestyle{empty}
	
	\begin{center}
		\large{\textbf{Résumé}}
	\end{center} 
	
	La fouille des motifs corrélés est une piste de recherche de plus en plus attractive en fouille de données 
	grâce à la qualité et à l'utilité des connaissances offertes par ces motifs. Plus précisément, les motifs fréquents corrélés ont été largement étudiés auparavant dans la littérature.
	
	Notre objectif dans cette thèse est de bénéficier à la fois des connaissances offertes par les motifs corrélés fréquents ainsi que les motifs rares corrélés selon la mesure de corrélation \textit{bond}.
	Cependant, un principal problème est lié à la fouille des motifs corrélés concerne le nombre souvent très élevé des motifs corrélés extraits. Un tel nombre handicape une exploitation optimale et aisée des connaissances encapsulées dans ces motifs. Pour pallier ce problème, nous nous intéressons dans cette thèse à l'extraction d'un sous-ensemble, sans perte d'information, de l'ensemble de tous les motifs corrélés. Ce sous-ensemble, le noyau d'itemsets, appelé ``Représentations Concises'', à partir duquel tous les motifs redondants peuvent être régénérés sans perte d'informations. Le but d'une telle représentation est de minimiser le nombre de motifs extraits tout en préservant les connaissances cachées et pertinentes. 
	
	Afin de réaliser cet objectif, nous nous sommes basés sur les notions dérivées de l'analyse formelle de concepts AFC. Plus précisément, les représentations condensées, que nous proposons, sont issues des notions de classes d'équivalence induites par l'opérateur de fermeture $f_{bond}$ associé à la mesure de corrélation \textit{bond}. Après la caractérisation des représentations condensées proposées, nous introduisons l'algorithme \textsc{Gmjp} dédié à l'extraction des motifs corrélés fréquents, des motifs corrélés rares ainsi que leurs représentations condensées associées. Nous présentons également l'algorithme \textsc{Regenerate} d'interrogation de la représentation $\mathcal{RCPR}$ associée à l'ensemble $\mathcal{RCP}$ des motifs corrélés rares et nous proposons aussi l'algorithme \textsc{RCPRegeneration} dédié à la régénération de l'ensemble total des motifs corrélés rares à partir de la représentation concise $\mathcal{RCPR}$.
	
	L'évaluation expérimentale menée met en valeur les taux de compacités très intéressants offerts par les différentes représentations concises proposées et justifie également les performances encourageantes de l'approche \textsc{Gmjp}. Afin d'améliorer les performances de l'algorithme \textsc{Gmjp}, nous proposons une version optimisée de \textsc{Gmjp}. Cette version optimisée présente des temps d'exécution beaucoup plus réduits que la version initiale. De plus, nous avons conduit un processus de classification associative basé sur les règles associatives corrélées dérivées à partir des motifs corrélés fermés et de leurs générateurs minimaux. Les résultats de classification des données de détection d'intrusions, sont très encourageants et ont prouvé une grande utilité de la fouille des motifs corrélés.
	
	\hspace{+0.5cm} \textbf{Mots Cl\'es :} {Analyse Formelle de Concept, Fouille sous Contraintes, Monotonie, Anti-monotonie, Mesure \textit{bond}, Extraction de motifs, Représentation concise, Classification, R\`egles associatives.}
	\pagenumbering{roman}
	\thispagestyle{empty}
	\newpage
	\thispagestyle{empty}
	\cleardoublepage
	
	\tableofcontents
	\cleardoublepage
	\newpage
	\listoffigures
	\cleardoublepage
	\newpage
	\listoftables
	\newpage
	\listofalgorithms
	\newpage
	\thispagestyle{empty}
	\pagenumbering{arabic}
	\cleardoublepage

	\chapter{Introduction}\label{ch_introduction}
	\markboth{Introduction}{Introduction}
	\section{Introduction and Motivations}
	The development of new information and communication technologies and the globalization of markets make the competition more and more increased among companies. In this sense, the need for access to an accurate information for decision-making is increasingly urgent. The actual problem is linked to lack of access to relevant information in the presence of the large amount of data. The collected data in various fields are becoming larger. This motivates the need to analyze and interpret data in order to extract useful knowledge.
	
	In this context, the process of knowledge discovery from databases \textsc{(}KDD\textsc{)} is a complete process aiming to extract useful, hidden knowledge from huge amount of data \cite{Agra94}. Data Mining is one of the main steps of this process and is dedicated to offer the necessary tools needed for an optimal exploration of data. 
	Many state of the art approaches were focused on frequent itemset extraction and association rule generation. 
	Nevertheless, two main problems handicap the good use of the returned knowledge from the set of frequent itemsets.
	The first problem is related to the quality of the offered knowledge since the degree of correlation of the extracted itemsets may be not interesting for the end user. 
	The second problem is related often to the huge quantity of the extracted knowledge.
	
	To overcome these problems, many previous works propose to integrate the correlation measures within the mining process \cite{Brin97,comine_Lee,Omie03,ccmine_Kim,Xiong06hypercliquepattern}. 
	Correlated pattern mining is then shown to be more complex but more informative than traditional frequent patterns mining. In fact, correlated patterns offer a precise information about the degree of apparition of the items composing a given itemset \cite{borgelt}. This key information specifies the simultaneous apparition frequency among items, \textit{i.e.}, their co-occurrence, as well as their apparition frequency, \textit{i.e.}, their occurrence.
	
	Other state of the art approaches deal with the extraction of a subset, without information loss, of the whole set of correlated patterns. This subset, is named, ``Condensed Representation`` and from which we are able to derive all the redundant correlated patterns. 
	The condensed representations prove their high utility in different fields such as: bioinformatics \cite{pasquier2009} and data grids \cite{tarekJSS2015}. 
	
	The main objective behind defining such a condensed representation is to reduce the number of the extracted patterns while preserving the same amount of pertinent knowledge.
	In addition to this, all of the extracted associated rules,  derived from correlated patterns fulfilling a correlation measure such as \textit{all-confidence} or \textit{bond}, are valid with respect to  minimal support  and to minimal confidence thresholds \cite{Omie03}. 
	
	Frequent correlated itemset mining was then shown to be an interesting task in data mining. Since its inception, this key task grasped the interest of many researchers since it meets the needs of experts in several application fields \cite{tarekds2010}, such as market basket study. However, the application of correlated frequent patterns is not an attractive solution for some other applications, \textit{e.g.}, intrusion detection,  analysis of the genetic confusion from biological data,  detection of rare diseases from medical data, to cite but a few \cite{livreIGIGlobal2010,mahmood_less_frequent_patterns_vs_networks,romero2010,laszloIJSI2010,haglin08}.
	As an illustration of the rare correlated patterns applications in the field of medicine, the rare combination of symptoms can provide useful insights for doctors.

	To the best of our knowledge, there is no previous work that dealt with both frequent correlated as well as rare correlated patterns according to a specified correlation metric. Thus, motivated by this issue, we propose in this thesis to benefit from the knowledges returned from both frequent correlated as well as rare correlated patterns according to the \textit{bond} correlation measure. To solve this challenging problem, we propose an efficient algorithmic framework, called \textsc{GMJP}, allowing the extraction of both frequent correlated patterns, rare correlated patterns as well as their associated concise representations.
	
	\section{Contributions}
	
	Our first contribution consists in defining and studying the characteristics of the condensed representations associated to frequent correlated as well as the condensed representations associated to rare correlated ones.
	In this respect, we are based on the notions derived from the Formal Concept Analysis \textsc{(FCA)} \cite{ganter99}, specifically the equivalence classes associated to a closure operator $f_{bond}$ dedicated to the \textit{bond} measure to introduce our new concise representations of both frequent correlated and rare correlated patterns. The first concise representation $\mathcal{RCPR}$ associated to the $\mathcal{RCP}$ set of rare correlated patterns,  is composed by the maximal elements of the rare correlated equivalence classes, called ``Closed Rare Correlated Patterns $\mathcal{CRCP}$ set`` union of their associated minimal generators called ``Minimal Rare Correlated Patterns $\mathcal{MRCP}$ set``.  Two other optimizations of the  $\mathcal{RCPR}$ representation are also proposed. The first optimization is composed by the whole set $\mathcal{CRCP}$ of closed rare correlated patterns union of the minimal elements of 
	the $\mathcal{MRCP}$ set. The second optimization is composed by the maximal elements of the $\mathcal{CRCP}$ of closed rare correlated patterns union of the whole $\mathcal{MRCP}$ set. We prove that both of these representations are also concise and exact. Our third optimized representation is a condensed approximate representation. The latter is composed by 
	the maximal elements of the $\mathcal{CRCP}$ set union of the minimal elements of 
	the $\mathcal{MRCP}$ set. According to the $\mathcal{FCP}$ set of frequent correlated patterns, the condensed exact representation is composed by the Closed Correlated Frequent Patterns. We prove the theoretical  properties of accuracy and compactness of all the proposed representations.

	Our second contribution is the design and the implementation of a new mining approach, called \textsc{Gmjp}, allowing the extraction of the sets of frequent correlated patterns, of rare correlated patterns and their associated concise representations. \textsc{Gmjp} is a sophisticated mining approach that allows
	a simultaneous integration of two opposite paradigms of monotonic and anti-monotonic constraints. 
	In addition, we present the \textsc{Regenerate} algorithm allowing the query of the $\mathcal{RCPR}$ condensed representation associated to the $\mathcal{RCP}$ set as well as the  \textsc{RcpRegeneration} algorithm dedicated to the regeneration of the whole set of rare correlated patterns from the $\mathcal{RCPR}$ representation.

	Our third contribution consists in proposing an optimized version of \textsc{Gmjp}. The latter shows much better performance than the initial version of \textsc{Gmjp}. In order to prove the usefulness of the extracted condensed representation, we conduct a classification process based on correlated association rules derived from closed correlated patterns and their associated minimal generators. The obtained rules are applied to the context of intrusion detection and achieve promoting results.

	The evaluation protocol of our approaches consists in experimental studies carried out over dense and sparse benchmark datasets commonly used in evaluating data mining contributions. The evaluation of the classification process is based on the \textsc{KDD 99} database of intrusion detection data. We also conduct the process of applying the $\mathcal{RCPR}$ representation on the extraction of rare correlated association rules from Micro-array gene expression data related to Breast-Cancer. The diverse obtained association-rules reveals a variety of relationship between up and down regulated gene-expressions.

	\section{Thesis Organization}
	The remainder of this thesis is organized as follows:
	
	\bigskip
	
	\textbf{Chapter 2}  introduces the basic notions related to the itemset search space and to itemset extraction. We also define two distinct categories of constraints: monotonic and anti-monotonic. We equally introduce the environment of Formal Concept Analysis \textsc{(FCA)} which offers the basis for the proposition of our approaches, specifically the notions of Closure Operator, Minimal Generator, Closed Pattern, Equivalence class and Condensed representation of a set of patterns.
	
	\bigskip
	
	\textbf{Chapter 3}  offers an overview of the state of the art approaches dealing with correlated patterns mining. We start this chapter by defining the most used correlation measures. Then, we continue with the approaches related to frequent correlated patterns,  followed by the state of the art of rare correlated patterns then the overview of the algorithms focusing on condensed representations of correlated patterns.

	\bigskip
	
	\textbf{Chapter 4}  focuses on characterizing the $\mathcal{FCP}$ set of frequent correlated patterns as well as the $\mathcal{RCP}$ set of rare correlated patterns. It introduces the condensed exact and approximate representations associated to the $\mathcal{RCP}$ set as well as the concise exact representation associated to the $\mathcal{FCP}$ set. The main content of this chapter was published in \cite{rnti2012} and in \cite{ida2015}.

	\bigskip

	\textbf{Chapter 5}  introduces the \textsc{Gmjp} approach, allowing the extraction of the sets of frequent correlated patterns, of rare correlated patterns and their associated concise representations. 
	The optimized version of \textsc{Gmjp}, named \textsc{Opt-Gmjp}, was also presented. This chapter also presents the theoretical complexity approximation of \textsc{Gmjp}. In addition, this chapter describes the \textsc{Regenerate} algorithm allowing the query of the $\mathcal{RCPR}$ condensed representation associated to the $\mathcal{RCP}$ set as well as the  \textsc{RcpRegeneration} algorithm dedicated to the regeneration of the whole set of rare correlated patterns from the $\mathcal{RCPR}$ representation.  
	The main content of this chapter was published in \cite{egc2012} and in \cite{ida2015}.
	
	\bigskip

	\textbf{Chapter 6}  focuses on the experimental validation of the proposed approaches.
	The evaluation process is based on two main axes, the first is related to the compactness rates of the condensed representations  while the second axe concerns the running time. This chapter evaluates the optimized version of \textsc{Gmjp}, which presents much better performance than do \textsc{Gmjp} over different benchmark datasets. The content related to the optimizations and evaluations was published in \cite{sac2015}.

	\bigskip

	\textbf{Chapter 7} describes the classification process based on correlated patterns.
	This chapter starts by presenting the framework of association rule extraction, it clarifies the properties of the generic bases of association rules. Then, we continue with the detailed presentation of the application of both frequent correlated and rare correlated patterns within the classification of some UCI benchmark datasets. We equally present the application of rare correlated patterns in the classification of intrusion detection data from the \textsc{KDD 99} dataset. The obtained results showed the usefulness of our proposed classification method over four different intrusion classes.  This chapter is concluded with the 
	application of  the $\mathcal{RCPR}$ representation on the extraction of rare correlated association rules from Micro-array gene expression data. These extracted rules aims to identify relations among up and down regulated gene expressions.
	The main content of this chapter was published in \cite{pakdd2012} and in \cite{dexa2013}.
	
	\bigskip

	\textbf{Chapter 8} concludes the thesis and sketches out our perspectives for future work.
	
\part{Review of Correlated Patterns Mining} 
\chapter{Basic Notions}\label{ch2}

\section{Introduction} \label{se1}
The extraction of correlated patterns is shown to be more complex but more informative than traditional frequent patterns mining. In fact, these correlated patterns present a strong link among the items they compose and they prove their high utility in many real life applications fields.

This chapter is dedicated to the introduction of the basic notions needed for the presentation of our approaches.
The second section deals with the basic notions related to the search space as well as the itemsets's extraction. Then, we link in the third section with the presentation of the foundations of the formal concepts analysis \textsc{(FCA)} framework \cite{ganter99}. The last section concludes the chapter.
\section{Search Space}
We begin by presenting the key notions related to itemset extraction, that will be used thorough this thesis.  First, let us define an extraction context.
\subsection{Extraction Context}
\begin{definition} \label{definitionbasetransactions} \textbf{Extraction Context}\\
	An extraction context \textsc{(}also called Context or Dataset\textsc{)} is represented by a triplet $\mathcal{C}$ = \textsc{(}$\mathcal{T},\mathcal{I},\mathcal{R}$\textsc{)} with
	$\mathcal{T}$ and $\mathcal{I}$ are, respectively, a finite sets of transactions \textsc{(}or objects\textsc{)} and of items \textsc{(}or attributes\textsc{)}, and $\mathcal{R}$ $\subseteq$ $\mathcal{T} \times \mathcal{I}$ is a binary relation between the transactions and the items. A couple \textsc{(}$t$, $i$\textsc{)} $\in$ $\mathcal{R}$ if $t$ $\in$ $\mathcal{T}$ contains $i$ $\in$ $\mathcal{I}$.
\end{definition}
\begin{example}
	An example of an extraction context $\mathcal{C}$ $=$
	$\textsc{(}$$\mathcal{T},\mathcal{I},\mathcal{R}$$\textsc{)}$
	is given by Table \ref{Base_transactions}. In this context, the transaction set $\mathcal{T} = \{1, 2, 3, 4, 5\}$
	\textsc{(}\textit{resp.} the object set $\mathcal{O} = \{1, 2, 3, 4,
	5\}$\textsc{)} and  the items set $\mathcal{I}$ $=$
	$\{$\texttt{A}, \texttt{B}, \texttt{C}, \texttt{D}, \texttt{E},$\}$.
	The couple \textsc{(}2, B\textsc{)} $\in$ $\mathcal{R}$ since the transaction 2 $\in$ $\mathcal{T}$ contains the item B $\in$ $\mathcal{I}$.
\end{example}
\begin{table}[h]
	\begin{center}
		\footnotesize{
			\begin{tabular}{|c||c|c|c|c|c|c|}
				\hline & \texttt{A}  & \texttt{B}  & \texttt{C}  & \texttt{D} & \texttt{E}  \\
				\hline\hline
				1 & $\times$  &          &$\times$    & $\times$&         \\
				\hline
				2 &           & $\times$ &$\times$    &         & $\times$ \\
				\hline
				3 & $\times$  & $\times$ &$\times$    &         & $\times$ \\
				\hline
				4 &           & $\times$ &            &         & $\times$ \\
				\hline
				5 & $\times$  & $\times$ &$\times$    &         & $\times$  \\
				\hline
		\end{tabular}}
	\end{center}
	\caption{An example of an Extraction Context $\mathcal{C}$.}\label{Base_transactions}
\end{table}
\begin{remark}
	We note, by sake of accuracy, that the notations of transactions database and extraction context have the same
	meaning thorough this thesis.
	They are denoted as $\mathcal{D}$ $=$ $\textsc{(}\mathcal{T}, \mathcal{I},
	\mathcal{R}\textsc{)}$.
\end{remark}
\begin{definition} \label{motif} \textbf{Itemset or Pattern}\\
	A transaction $t$ $\in$ $\mathcal{T}$, having an identifier
	denoted by \textit{TID} $\textsc{(}$Tuple
	IDentifier$\textsc{)}$, contains a non-empty set of items belonging to $\mathcal{I}$.
	A subset $I$ of $\mathcal{I}$ where $k$ $=$
	$\vert I \vert $ is called  a \textit{$k$-pattern} or simply a
	\textit{pattern}, and $k$ represents the cardinality of $I$. The number of transactions $t$ of a context $\mathcal{C}$ containing a pattern $I$,
	$\vert$$\{$ $t$ $ \in $ $\mathcal{D}$ $\vert $ $I$ $\subseteq $
	$t$$\}$$\vert $, is called \textit{absolute support} of $I$ and is denoted
	$Supp\textsc{(}\wedge I\textsc{)}$.
	The \textit{relative support} of $I$ or the \textit{frequency} of $I$, denoted
	$freq\textsc{(}I\textsc{)}$, is the quotient of the absolute support
	by the total number of the transactions of $\mathcal{D}$,
	\textit{i.e.}, $freq\textsc{(}I\textsc{)}$ $=$
	$\displaystyle\frac{\displaystyle{\vert {\{}t \in \mathcal{D} | I \subseteq t
			{\}}\vert}}{\displaystyle{\vert \mathcal{T}\vert}}$.
\end{definition}
\begin{remark}
	We point that, thorough this thesis, we are mainly interested in itemsets \textit{i.e.} the set of items as a kind of patterns.
	Consequently, we use a form without separators to denote an itemset.
	For example, \texttt{BD} stands for the itemset composed by the items
	\texttt{B}  and \texttt{D}.
\end{remark}
\subsection{Supports of a Pattern}
To evaluate an itemset, many interesting measures can be used. The most common ones are presented by Definition \ref{definitionsupportmotif}.
\begin{definition} \label{definitionsupportmotif} \textbf{Supports of a Pattern}\\
	Let $\mathcal{D}$=\textsc{(}$\mathcal{T}, \mathcal{I}, \mathcal{R}$\textsc{)} an extraction context and a non empty itemset $I$ $\subseteq$ $\mathcal{I}$. We distinguish three kinds of supports for an itemset $I$ :
	
	- \textbf{\textit{The conjunctive support:}} \textit{Supp}\textsc{(}$\wedge$$I$\textsc{)} = $\mid$$\{$$t$ $\in$ $\mathcal{T}$ $\mid$ $\forall$ $i$ $\in$ $I$ : \textsc{(}$t$, $i$\textsc{)} $\in$ $\mathcal{R}$$\}$$\mid$
	
	- \textbf{\textit{The disjunctive support:}} \textit{Supp}\textsc{(}$\vee$$I$\textsc{)} = $\mid$$\{$$t$ $\in$ $\mathcal{T}$ $\mid$ $\exists$ $i$ $\in$ $I$ : \textsc{(}$t$, $i$\textsc{)} $\in$ $\mathcal{R}$$\}$$\mid$, and,
	
	- \textbf{\textit{The negative support:}} \textit{Supp}\textsc{(}$\neg$$I$\textsc{)} = $\mid$$\{$$t$ $\in$ $\mathcal{T}$ $\mid$ $\forall$ $i$ $\in$ $I$ : \textsc{(}$t$, $i$\textsc{)} $\notin$ $\mathcal{R}$$\}$$\mid$.
\end{definition}

More explicitly, for an itemset $I$, the supports are defined as follows:

$\bullet$  \textit{Supp}\textsc{(}$\wedge$$I$\textsc{)}: is equal to the number of transactions containing all the items of $I$.

$\bullet$  \textit{Supp}\textsc{(}$\vee$$I$\textsc{)}: is equal to the number of transactions containing  at least one item of $I$.

$\bullet$  \textit{Supp}\textsc{(}$\neg$$I$\textsc{)}: is equal to the number of transactions that do not contain any item of $I$.

It is important to note that the ``De Morgan'' law ensures the transition between the
disjunctive and the negative support of an itemset $I$ as follows :
\textit{Supp}\textsc{(}$\neg$\textit{I}\textsc{)} =
$\mid\mathcal{T}\mid$ - \textit{Supp}\textsc{(}$\vee$\textit{I}\textsc{)}.
\begin{example}
	Let us consider the extraction context given by Table \ref{Base_transactions} that will be used thorough the different examples. We have \textit{Supp}\textsc{(}$\wedge$\texttt{AD}\textsc{)} = $\mid$$\{$1$\}$$\mid$ = $1$, \textit{Supp}\textsc{(}$\vee$\texttt{AD}\textsc{)} = $\mid$$\{$ 1, 3, 5$\}$$\mid$ = $3$, and, \textit{Supp}\textsc{(}$\neg$\textsc{(}\texttt{AD}\textsc{\textsc{)}\textsc{)}} = $\mid$$\{$2, 4$\}$$\mid$ = $2$ $^{\textsc{(}}$.
\end{example}
In the following, if there is no risk of confusion, the conjunctive support will be simply denoted by \textit{support}.
Note that \textit{Supp}\textsc{(}$\wedge$$\emptyset$\textsc{)} = $|\mathcal{T}|$ since the empty set is included in all transactions, while \textit{Supp}\textsc{(}$\vee$$\emptyset$\textsc{)} = $0$ since the empty set does not contain any item. Moreover, $\forall$ $i$ $\in$ $\mathcal{I}$, \textit{Supp}\textsc{(}$\wedge$$i$\textsc{)} = \textit{Supp}\textsc{(}$\vee$$i$\textsc{)}, while in the general
case, for $I$ $\subseteq$ $\mathcal{I}$ and $I$ $\neq$
$\emptyset$, \textit{Supp}\textsc{(}$\wedge$$I$\textsc{)} $\leq$
\textit{Supp}\textsc{(}$\vee$$I$\textsc{)}. A pattern $I$ is said to
be \textit{frequent} if \textit{Supp}\textsc{(}$\wedge$$I$\textsc{)}
is greater than or equal to a user-specified minimum support
threshold, denoted \textit{minsupp} \cite{Agra94}. The following
lemma shows the links that exist between the different supports of a
non-empty pattern $I$. These links are based on the
\textit{inclusion-exclusion identities} \cite{galambos}.
\begin{lemma}\label{lemmaidentitésinclusionexclusion} - \textbf{Inclusion-exclusion
		identities} - The inclusion-exclusion identities ensure the links
	between the conjunctive, disjunctive and negative supports of a non-empty pattern $I$.
	\begin{center}
		\vspace{-0.5cm}
		\begin{tabular}{lr}
			$ \textit{Supp}\textsc{(}\wedge I\textsc{)}\mbox{ }=\mbox{ }\displaystyle\sum_{\emptyset \subset I_1 \subseteq I} {\textsc{(}-1\textsc{)}^{\mbox{$\mid I_1\mid$\mbox{ - 1}}}\mbox{
				}\textit{Supp}\textsc{(}\vee I_1\textsc{)}}$
			& \textsc{(}1\textsc{)} \\
			$ \textit{Supp}\textsc{(}\vee I\textsc{)}\mbox{ }=\mbox{ }\displaystyle\sum_{\emptyset \subset I_1 \subseteq I} {\textsc{(}-1\textsc{)}^{ \mbox{$\mid I_1\mid $\mbox{ - 1}}}\mbox{
				}\textit{Supp}\textsc{(}\wedge I_1\textsc{)}}$
			& \textsc{(}2\textsc{)} \\
			$ \textit{Supp}\textsc{(}\neg I\textsc{)}\mbox{ } =\mbox{ } \mid \mathcal{T} \mid \mbox{ }-\mbox{
			} \textit{Supp}\textsc{(}\vee I\textsc{)} \mbox{ }\textsc{(}\mbox{The De Morgan's law}\textsc{)}$ & \textsc{(}3\textsc{)}
		\end{tabular}
	\end{center}
\end{lemma}

\subsection{Frequent Itemset  - Rare Itemset - Correlated Itemset}
Given a minimal threshold of support \cite{Agra94}, we distinguish between two kinds of patterns, frequent patterns and infrequent patterns \textsc{(}also called Rare patterns\textsc{)}.
\begin{definition} \label{motiffréq} \textbf{Frequent Itemset  - Rare Itemset} \\
	Let an extraction context $\mathcal{C}$ = $\textsc{(}\mathcal{T}, \mathcal{I},\mathcal{R}\textsc{)}$, a minimal threshold of the conjunctive support \textit{minsupp}, an itemset $I$ $\subseteq$ $\mathcal{I}$ is said \textit{frequent} if \textit{Supp}\textsc{(}$\wedge$$I$\textsc{)} $\geq$ \textit{minsupp}. Otherwise, $I$ is said \textit{infrequent} or \textit{rare}.
\end{definition}
\begin{example} Let  \textit{minsupp} = 2. \textit{Supp}\textsc{(}$\wedge$\texttt{BCE}\textsc{)}
	= 3, the pattern \texttt{BCE} is a frequent pattern. However, the pattern \texttt{CD} is a rare pattern since \textit{Supp}\textsc{(}$\wedge$\texttt{CD}\textsc{)} = 1 $<$ 2.
\end{example}
In the following, we need to define the smallest rare patterns according to the relation of inclusion set. They correspond to rare patterns having all subsets frequent, and are defined as follows:
\begin{definition}\label{mrp} \textbf{Minimal rare patterns}\\
	The $\mathcal{M}in \mathcal{RP}$ set of minimal rare patterns  is composed of rare patterns having no rare proper subsets. This set
	is defined as: $\mathcal{M}in$$\mathcal{RP}$ = $\{I$ $\in$ $\mathcal{I} |$ $\forall$ $I_1 \subset I$: \textit{Supp}\textsc{(}$\wedge I_1$\textsc{)} $\geq$ \textit{minsupp}$\}$.
\end{definition}
\begin{example}\label{example_MRP}
	Let us consider the extraction context sketched by Table \ref{Base_transactions}. For \textit{minsupp} = 4,
	we have $\mathcal{M}in \mathcal{RP}$ = $\{$$A$, $D$, $BC$, $CE$$\}$. $A$ and $D$ are minimal rare items, $BC$ is a minimal rare itemset since it is composed by two frequent items: $B$ with \textit{Supp}\textsc{(}$\wedge$\texttt{B}\textsc{)} = 4 and $C$
	with \textit{Supp}\textsc{(}$\wedge$\texttt{C}\textsc{)} = 4.
\end{example}

In fact, in order to reduce the high number of frequent itemsets and to improve the quality of the extracted frequent itemets,  other interesting measures apart from the conjunctive support are introduced  within the mining process. These latter are called  ``Correlation Measures''.
The itemsets fulfilling a given correlation measure are called ``Correlated Itemsets''. This latter type of itemsets is defined in a generic way in what follows:
\begin{definition} \label{motifCorr} \textbf{Correlated Itemset} \\
	Let a correlation measure M, a minimal correlation threshold \textit{minCorr}, an itemset $I$ $\subseteq$ $\mathcal{I}$ is said \textit{correlated} according to the measure M, if \textit{M}\textsc{(}$I$\textsc{)} $\geq$ \textit{minCorr}. $I$ is said \textit{non correlated} otherwise.
\end{definition}
\subsection{Categories of Constraints}
Besides the minimal frequency constraint expressed by the \textit{minsupp} threshold, other constraints can be
integrated within the itemset's extraction process. These constraints have two distinct types, ``The monotonic constraints'' and ``The anti-monotonic constraints'' \cite{luccheKIS05_MAJ_06}.
\begin{definition}\label{anti-monotone} \textbf{Anti-monotonic Constraint}\\
	A constraint $Q$ is \textit{anti-monotone} if $\forall$ $I$ $\subseteq$ $\mathcal{I}$, $\forall$ $I_1$ $\subseteq$ $I$ : $I$ fulfills $Q$ $\Rightarrow$ $I_1$ fulfills $Q$.
\end{definition}
\begin{definition}\label{monotone} \textbf{Monotone Constraint}\\
	A constraint $Q$ is \textit{monotone} if $\forall$ $I$ $\subseteq$ $\mathcal{I}$, $\forall$ $I_1$ $\supseteq$ $I$ : $I$ fulfills $Q$ $\Rightarrow$ $I_1$ fulfills $Q$.
\end{definition}
\begin{example} The \textit{frequency constraint}, i.e. having a support greater than or equal to \textit{minsupp}, is an anti-monotonic constraint. In fact, $\forall$ $I$, $I_1$ $\subseteq$ $\mathcal{I}$, if $I_1$ $\subseteq$ $I$ and \textit{Supp}\textsc{(}$\wedge$$I$\textsc{)} $\geq$ \textit{minsupp}, then \textit{Supp}\textsc{(}$\wedge$$I_1$\textsc{)} $\geq$ \textit{minsupp} since \textit{Supp}\textsc{(}$\wedge$$I_1$\textsc{)} $\geq$ \textit{Supp}\textsc{(}$\wedge$$I$\textsc{)}.
	
	Dually, the \textit{constraint of rarity}, i.e. having a support strictly lower than \textit{minsupp}, is a monotonic constraint. In fact, $\forall$ $I$, $I_1$ $\subseteq$ $\mathcal{I}$, if $I_1$ $\supseteq$ $I$ and \textit{Supp}\textsc{(}$\wedge$$I$\textsc{)} $<$ \textit{minsupp}, then \textit{Supp}\textsc{(}$\wedge$$I_1$\textsc{)} $<$ \textit{minsupp} since \textit{Supp}\textsc{(}$\wedge$$I_1$\textsc{)} $\leq$ \textit{Supp}\textsc{(}$\wedge$$I$\textsc{)}.
\end{example}
A set of itemset may fulfill different constraints simultaneously. Proposition \ref{conjCs}, whose proof is in
\cite{lee2006}, clarifies the conjunction of two constraints of the same nature.
\begin{proposition} \label{conjCs}
	The conjunction of anti-monotonic constraints
	$\textsc{(}$\textit{resp.} monotonic$\textsc{)}$ is an anti-monotonic $\textsc{(}$\textit{resp.} monotonic$\textsc{)}$ constraint.
\end{proposition}
Let us define now the dual notions of order-ideal and order-filter \cite{ganter99} defined on $\mathcal{P}\textsc{(}\mathcal{I}\textsc{)}$ and associated to the two kinds of constraints given by definitions
\ref{anti-monotone} et \ref{monotone}.
\begin{definition} \textbf{Order Ideal}\\
	A subset $\mathcal{S}$ of $\mathcal{P}\textsc{(}\mathcal{I}\textsc{)}$ is an order ideal if it fulfills the following properties:
	\begin{itemize}
		\item If $I$ $\in$ $\mathcal{S}$, then $\forall$ $I_1$ $\subseteq$ $I$ : $I_1$ $\in$ $\mathcal{S}$.
		\item If $I$ $\notin$ $\mathcal{S}$, then $\forall$ $I$ $\subseteq$ $I_1$ : $I_1$ $\notin$ $\mathcal{S}$.
	\end{itemize}
\end{definition}
\begin{definition} \textbf{Order Filter} \\
	A subset $\mathcal{S}$ of $\mathcal{P}\textsc{(}\mathcal{I}\textsc{)}$ is an order filter if it fulfills the following properties:
	\begin{itemize}
		\item If $I$ $\in$ $\mathcal{S}$, then  $\forall$ $I_1$ $\supseteq$ $I$ : $I_1$ $\in$ $\mathcal{S}$.
		\item If $I$ $\notin$ $\mathcal{S}$, then $\forall$ $I$ $\supseteq$ $I_1$ : $I_1$ $\notin$ $\mathcal{S}$.
	\end{itemize}
\end{definition}
An anti-monotone constraint such as the frequency constraint induces an order ideal on the itemset lattice. Dually, a monotonic constraint as the rarity constraint induces an order filter on the itemset lattice.
The set of itemsets fulfilling a given constraint is called \textit{a Theory} \cite{mannila97}. This theory
is delimited by two borders, the positive and the negative one, that are defined as follows:
\begin{definition}\label{bd} \textbf{Negative/Positive Border} \cite{luccheKIS05_MAJ_06}\\
	When considering an anti-monotonic constraint $C_{am}$,
	the border corresponds to the set of itemsets whose all subsets fulfill this constraint and whose all super-sets
	do not fulfill.
	Let a set of itemsets $\mathcal{S}$$_{am}$ fulfilling an anti-monotonic constraint $C_{am}$, the border is formally defined as:
	\begin{center}
		$\mathcal{B}d$\textsc{(}$\mathcal{S}$$_{am}$\textsc{)} =
		$\{$$X$ $|$ $\forall$ $Y$  $\subset$ $X$ : $Y$ $\in$ $\mathcal{S}$$_{am}$ and
		$\forall$ $Z$  $\supset$ $X$ : $Z$ $\notin$ $\mathcal{S}$$_{am}$$\}$
	\end{center}
	In the case of monotonic constraint  $C_{m}$,
	the border corresponds to the set of patterns whose all supersets
	fulfills this constraint and whose all subsets do not fulfill.
	
	Let a set of patterns $\mathcal{S}$$_{m}$ fulfilling a  monotonic constraint $C_{m}$, the border is formally defined as follows:
	\begin{center}
		$\mathcal{B}d$\textsc{(}$\mathcal{S}$$_{m}$\textsc{)} =
		$\{$$X$ $|$ $\forall$ $Y$  $\supset$ $X$ : $Y$ $\in$ $\mathcal{S}$$_{m}$ and
		$\forall$ $Z$ $\subset$ $X$ : $Z$ $\notin$ $\mathcal{S}$$_{m}$$\}$
	\end{center}
	However, we have to distinguish for a given constraint $C$ between positive and negative borders. Let a set of patterns  $\mathcal{S}$ fulfilling a constraint $C$. The positive border is denoted by $\mathcal{B}{d}^{+}\textsc{(}\mathcal{S}$\textsc{)} and corresponds to the patterns belonging to the border $\mathcal{B}{d}\textsc{(}\mathcal{S}$\textsc{)}
	and fulfilling the constraint $C$.
	The negative border is denoted by $\mathcal{B}{d}^{-}\textsc{(}\mathcal{S}$\textsc{)} and corresponds to the set of patterns belonging to the border $\mathcal{B}{d}\textsc{(}\mathcal{S}$\textsc{)}
	and not fulfilling the constraint $C$.
	These two borders are formally expressed as follows:
	\begin{center}
		$\mathcal{B}{d}^{+}\textsc{(}\mathcal{S}\textsc{)}$ =
		$\mathcal{B}{d}\textsc{(}\mathcal{S}$\textsc{)} $\cap$ $\mathcal{S}$, \\
		$\mathcal{B}{d}^{-}\textsc{(}\mathcal{S}\textsc{)}$ =
		$\mathcal{B}{d}\textsc{(}\mathcal{S}$\textsc{)} $\setminus$ $\mathcal{S}$.
	\end{center}
\end{definition}

In the next sub-section, we focus on the definition and the presentation of the notions related to condensed representations associated to a set of patterns.
\subsection{Condensed Representations of a set of Patterns}
The extraction of interesting patterns may be a costly operation in execution time and in memory consumption. This is due to the high number of the generated candidates.
In this regard, an interesting issue consists in extracting sets of patterns with more reduced sizes. From which it is possible to regenerate the whole sets of patterns. These reduced sets are called ```Condensed Representations''.
In the case where the regeneration is performed in an exact way without information loss then the condensed representation is said \emph{exact}. Otherwise, the condensed representation is said \emph{approximative}.
These representations are formally defined in what follows.
\begin{definition}\label{repConcise}\textbf{Condensed Representations} \cite{mannila97}\\
	A concise representation of a set of interesting itemsets is a representative set allowing the characterization of the initial set in an exact or an approximative way.
\end{definition}
\begin{example}
	Let $\mathcal{R}$ be a concise representation of a set of frequent patterns $\mathcal{E}$. $\mathcal{R}$ is said \emph{concise exact representation}, if starting from  $\mathcal{R}$, we are able to determine for a given pattern whether it is a frequent pattern or not and to determine its conjunctive support also.
	For example, the closed frequent patterns \cite{pasquier99_2005} constitute a concise exact representation of the set of frequent itemsets.
	
	Otherwise, $\mathcal{R}$ is a \emph{concise approximative representation} of a set of patterns $\mathcal{S}$ if
	it is not able to exactly determine the support values of all the itemsets belonging to the $\mathcal{S}$ set.
	The representation $\mathcal{R}$ returns approximate values of these supports. For example,
	the maximal frequent itemsets \cite{maxminer} constitute an approximative concise representation of the frequent patterns set. In fact, thanks to maximal frequent itemsets we are able to determine whether a given itemset is frequent or rare but it is not possible to exactly derive  its conjunctive support value.
\end{example}
In general, a representation $\mathcal{R}$ constitutes ``a perfect cover'' if it fulfills the conditions established by the following definition:
\begin{definition}\label{Cover} \textbf{Perfect Cover}\\
	A set $\mathcal{E}1$ is said a perfect cover of a set $\mathcal{E}$ if and only if
	$\mathcal{E}1$ allows to cover $\mathcal{E}$ without information loss and the size of $\mathcal{E}1$ never exceeds that of the set $\mathcal{E}$.
\end{definition}

Various proposals aiming to reduce the size of a set of patterns  $\mathcal{E}$ are based on the foundations of formal concepts analysis \cite{ganter99}. The next section is dedicated to the presentation of the formal concept analysis's framework.
\section {Formal Concepts Analysis}
\subsection{Introduction}
The formal concept analysis initially introduced by Wille in 1982 \cite{wille82} treats formal concepts.
A formal concept is a set of objects, \textit{The Extension}, to which we applied a set of attributes, \textit{The Intention}. The formal concept analysis provides a classification and an analysis tool whose principal element is the itemsets's lattice defined as follows:
\begin{definition} \label{deftreillis} \textbf{Itemsets's Lattice}\\
	An itemsets's lattice is a conceptual and hierarchical schema of patterns. It is also said
	lattice of set inclusion. In fact, the power set of $\mathcal{I}$ is ordered by set inclusion in the itemsets' lattice.
\end{definition}
This lattice shows the frequent itemsets, the rare ones as well as the $\mathcal{M}$in$\mathcal{RP}$ set
of minimal rare patterns composing the positive border of the whole set of rare patterns.
\subsection{Galois Connection}
\bigskip
\textbf{2.3.2.1. Closure Operator}
\bigskip
In what follows, we present the fundamental basis of a closure operator.
\begin{definition}\label{definitionEnsembles
		ordonnés}  \textbf{Ordred Set} \\
	Let $E$ a set. A \textit{Partial Order} over the set $E$ is a binary relation $\leq$ over the elements of $E$, such as for  $x$, $y$, $z\in E$, the following properties holds \cite{davey02} :
	\\ 1. \textit{Reflexivity} : $x\leq x$
	\\ 2. \textit{Anti-symmetry} : $x\leq y$ and $y\leq x \Rightarrow x=y$
	\\ 3. \textit{Transitivity} : $x\leq y$ and $y\leq z \Rightarrow x\leq z$
	
	\bigskip
	
	A set $E$ with a partial order  $\leq$, denoted by
	$\textsc{(}$$E$,$\leq$\textsc{)}, is a \textit{partially ordered set} \cite{davey02}.
\end{definition}
	Through the following definition, we introduce the notion of closure operator.
	\begin{definition}\label{ClosOp}\textbf{Closure Operator} \cite{ganter99} \\
		Let a partially ordered set \textsc{(}$E$, $\leq$\textsc{)}. An application $f$ from
		\textsc{(}$E$, $\leq$\textsc{)} to \textsc{(}$E$, $\leq$\textsc{)} is a  \textit{closure operator}, if and only if $f$ fulfills the following properties. For all sub-sets $S, S'\subseteq E$ :
		
		1. \textit{Isotonic} : $S\leq S' \Rightarrow f\textsc{(}S\textsc{)}
		\leq f\textsc{(}S'\textsc{)}$
		
		2. \textit{Extensive} : $S\leq f\textsc{(}S\textsc{)}$
		
		3. \textit{Idempotency} :
		$f\textsc{(}f\textsc{(}S\textsc{)}\textsc{)}=f\textsc{(}S\textsc{)}$
	\end{definition}
	We now define, the closure operator related to the conjunctive search space where the conjunctive support
	characterizes the associated patterns.\\
	\bigskip
	\textbf{2.3.2.2. The Galois Connection}
	\bigskip
	\begin{definition}\label{Connexion de Galois} \textbf{Galois Connection} \cite{ganter99} \\
		Let an extraction context $\mathcal{C}$ $=$
		$\textsc{(}\mathcal{T}$, $\mathcal{I}$, $\mathcal{R}\textsc{)}$.
		Let $g_{c}$ the application from the power-set of $\mathcal{T}$ $^{\textsc{(}}$\footnote{The power-set of a set  $\mathcal{T}$, is constituted by the sub-sets
			of $\mathcal{T}$, is denoted by
			$\mathcal{P}$\textsc{(}$\mathcal{T}$\textsc{)}.}$^{\textsc{)}}$ to the power-set of items
		$\mathcal{I}$, and associate to the set of objects $T$ $\subseteq \mathcal{T}$ the set of items $i$ $\in$
		$\mathcal{I}$ that are common to all the objects $t$ $\in$ $T$ :

		\hspace{3.5cm}$g_{c} : \mathcal{P}\textsc{(}\mathcal{T}\textsc{)}{}
		\rightarrow \mathcal{P}\textsc{(}\mathcal{I}\textsc{)}{}$
		
		\hspace{5cm}$T\mapsto g_{c}\textsc{(}T\textsc{)}=\{i \in \mathcal{I}
		| \forall\mbox{ } t\in T, \textsc{(}t,i\textsc{)}\in\mathcal{R}$
		$\}$
		
		Let $h_{c}$ the application, from the power-set
		of $\mathcal{I}$ to the
		power-set of $\mathcal{T}$, which associate to each set of items
		$\textsc{(}$commonly called pattern$\textsc{)}$ $I$ $\subseteq \mathcal{I}$
		the set of objects $t$ $\subseteq \mathcal{T}$ containing all the items $i$ $\in$ $I$ :
		
		\hspace{3.5cm}$h_{c} : \mathcal{P}\textsc{(}\mathcal{I}\textsc{)}{}
		\rightarrow \mathcal{P}\textsc{(}\mathcal{T}\textsc{)}{}$
		
		\hspace{5cm}$I\mapsto h_{c}\textsc{(}I\textsc{)}=\{t\in \mathcal{T}
		| \forall\mbox{ } i\in I, \textsc{(}t,i\textsc{)}\in\mathcal{R}$
		$\}$
		
		The couple of applications $\textsc{(}$$g_{c}$,$h_{c}$$\textsc{)}$ is a \textit{Galois connection} between the power-set  of $\mathcal{T}$ and the power-set of $\mathcal{I}$.
	\end{definition}
	\begin{example}
		The images of $\{1\}$ and of $\{1, 2\}$ by $g_{c}$ as well as those of $\{\textsl{A}, \textsl{E}\}$ and of $\{\textsl{C}, \textsl{D}\}$ by the application
		$h_{c}$ are :
		
		$g_{c}$\textsc{(}$\{1\}$\textsc{)} $=$
		$\{\textsl{A},\textsl{C}, \textsl{D}\}$ ;
		$g_{c}$\textsc{(}$\{2\}\textsc{)} =
		\{\textsl{B}, \textsl{C}, \textsl{E}\}$ ;
		$g_{c}$\textsc{(}$\{1,
		2\}\textsc{)} = \{\textsl{C}\}$.
		
		$h_{c}\textsc{(}\{\textsl{A}, \textsl{E}\}\textsc{)} =
		\{3, 5\}$ ; $h_{c}\textsc{(}\{\textsl{C},\textsl{D}\textsc{)} = \{1\}$.
	\end{example}
	\begin{proposition} \cite{ganter99}\\
		Given a Galois connection, the following properties are fulfilled: $\forall$ $I$, $I_{_{1}}$, $I_{_{2}}$ $\subseteq$
		$\mathcal{I}$ and $T$, $T_{_{1}}$, $T_{_{2}}$ $\subseteq$
		$\mathcal{T}$ :
		
		1. $I_{_{1}}$ $\subseteq $ $I_{_{2}}$ $\Rightarrow $
		$h_{c}$\textsc{(}$I_{_{2}}$\textsc{)} $\subseteq $
		$h_{c}$\textsc{(}$I_{_{1}}$\textsc{)};
		
		2. $T_{_{1}}$ $\subseteq $ $T_{_{2}}$ $\Rightarrow $
		$g_{c}$\textsc{(}$T_{_{2}}$\textsc{)} $ \subseteq $
		$g_{c}$\textsc{(}$T_{_{1}}$\textsc{)};
		
		3. $T$ $ \subseteq $ $h_{c}$\textsc{(}$I$\textsc{)}
		$\Leftrightarrow$ $I$ $ \subseteq $ $g_{c}$\textsc{(}$T$\textsc{)}.
	\end{proposition}
	\medskip
	Thanks to Definition \ref{Clos of Galois}, we introduce the closure operators associated to the Galois connection.
	\begin{definition}\label{Clos of Galois}\textbf{Closure Operators of the Galois Connection} \cite{ganter99} \\
		Lets consider the power-sets
		$\mathcal{P}\textsc{(}\mathcal{I}\textsc{)}$ and
		$\mathcal{P}\textsc{(}\mathcal{T}\textsc{)}$ provided with the inclusion set link $\subseteq$, \textit{i.e}, the partially ordered sets \textsc{(}$\mathcal{P}\textsc{(}\mathcal{I}\textsc{)},$ $\subseteq$\textsc{)}
		and
		\textsc{(}$\mathcal{P}\textsc{(}\mathcal{T}\textsc{)},$ $\subseteq$\textsc{)}.
		The operators $f_{c}$ $^{\textsc{(}}$\footnote{We use the index \textbf{c} since the closure operator gathers itemsets sharing the same common \textbf{c}onjunctive support.}$^{\textsc{)}}$
		and $O_{c}$ such as $f_{c} = g_{c} \circ h_{c}$ of
		\textsc{(}$\mathcal{P}\textsc{(}\mathcal{I}\textsc{)},$ $\subseteq$\textsc{)}
		in
		\textsc{(}$\mathcal{P}\textsc{(}\mathcal{I}\textsc{)},$ $\subseteq$\textsc{)}
		and $O_{c} = h_{c} \circ g_{c}$ of
		\textsc{(}$\mathcal{P}\textsc{(}\mathcal{T}\textsc{)},$ $\subseteq$\textsc{)}
		in
		\textsc{(}$\mathcal{P}\textsc{(}\mathcal{T}\textsc{)},$ $\subseteq$\textsc{)}
		are the \textit{closure operators of the Galois connection}.
	\end{definition}
	\begin{example} \label{examplefoggof}
		Let the extraction context illustrated by Table \ref{Base_transactions}, we then have :\\
		$h_{c}\circ g_{c}$\textsc{(}$\{2\}$\textsc{)} $=$ $\{2, 3, 5\}$ ;
		$h_{c}\circ g_{c}$\textsc{(}$\{3\}$\textsc{)} $=$ $\{5\}$ ;
		$h_{c}\circ g_{c}$\textsc{(}$\{2, 3\}$\textsc{)} $=$ $\{2, 3 ,5\}$.\\
		$g_{c}\circ h_{c}$\textsc{(}$\{\textsl{B}\}$\textsc{)} $=$
		$\{\textsl{B}, \textsl{E}\}$ ;
		$g_{c}\circ h_{c}$\textsc{(}$\{\textsl{D}\}$\textsc{)} $=$ $\{\textsl{D}\}$ ;
		$g_{c}\circ h_{c}$\textsc{(}$\{\textsl{A}, \textsl{D}\}$\textsc{)} $=$
		$\{\textsl{D}\}$.
	\end{example}
	\subsection{Equivalence Classes, Closed Patterns and Minimal Generators}
	The application of the closure operator $\gamma$ induces an equivalence relation in the power-set
	$\mathcal{P}\textsc{(}\mathcal{I}$\textsc{)}, partitioning it on equivalence classes \cite{AyouniLYP10,PASCAL00},
	denoted by $\gamma$-equivalence-class, defined as follows.
	\begin{definition}\label{Cls-equiv} \textbf{$\gamma$-Equivalence-Class}\\
		A $\gamma$-Equivalence-Class contains all the itemsets belonging exactly to the same transactions and sharing the
		same closure according to the $\gamma$ closure operator.
	\end{definition}
	Within a $\gamma$-Equivalence-Class, the maximal element, according to the set inclusion, is said, ``Closed Pattern'' where as the minimal elements which are incomparable according to the set inclusion, are called ``Minimal Generators''.
	They are defined in what follows.
	\begin{definition} \label{DefMferme} \textbf{Closed Pattern} \cite{PASCAL00}\\
		An  itemset $I$ $\subseteq$ $\mathcal{I}$ is a closed itemset iff, $\gamma$\textsc{(}$I$\textsc{)}=$I$.
	\end{definition}
	\begin{definition} \label{DefGM} \textbf{Minimal Generator} \cite{PASCAL00}\\
		An itemset $I1$ $\subseteq$ $\mathcal{I}$ is a minimal generator of a closed pattern $I$ if
		$\gamma$\textsc{(}$I1$\textsc{)}=$I$ and $\forall$ $I2$ $\subseteq$ $\mathcal{I}$, if
		$I2$ $\subseteq$ $I1$ and $\gamma$\textsc{(}$I2$\textsc{)}=$I$ then $I2$ = $I1$.
	\end{definition}
	The following proposition introduces an interesting property of the minimal generators set.
	\begin{proposition} \cite{titanic02}
		Let $\mathcal{GM}$ be the set of minimal generators extracted from a context $\mathcal{C}$,  the $\mathcal{GM}$ set fulfills an order ideal property on the itemset lattice.
	\end{proposition}
	\begin{example}
		A \textit{conjunctive equivalence class} is a set containing all the patterns having the same conjunctive closure.
		Thus, these patterns owns the same value of conjunctive support. The minimal generators are the
		smallest elements, according to the set inclusion property, in their equivalence classes.
		Whereas, the largest element in this class corresponds to the closed pattern.
		An example of a conjunctive equivalence class is given by Figure \ref{ClSeQuivConj}.
		In this class,  \textsl{ABCE} is the closed pattern whereas \textsl{AB} and \textsl{AE} are the associated
		minimal generators. All the elements belonging to this class share exactly the same conjunctive support, equal to
		2.
	\end{example}
	\begin{figure}[htbp]
		\hspace{+0.5cm}
		\includegraphics [scale = 0.40]{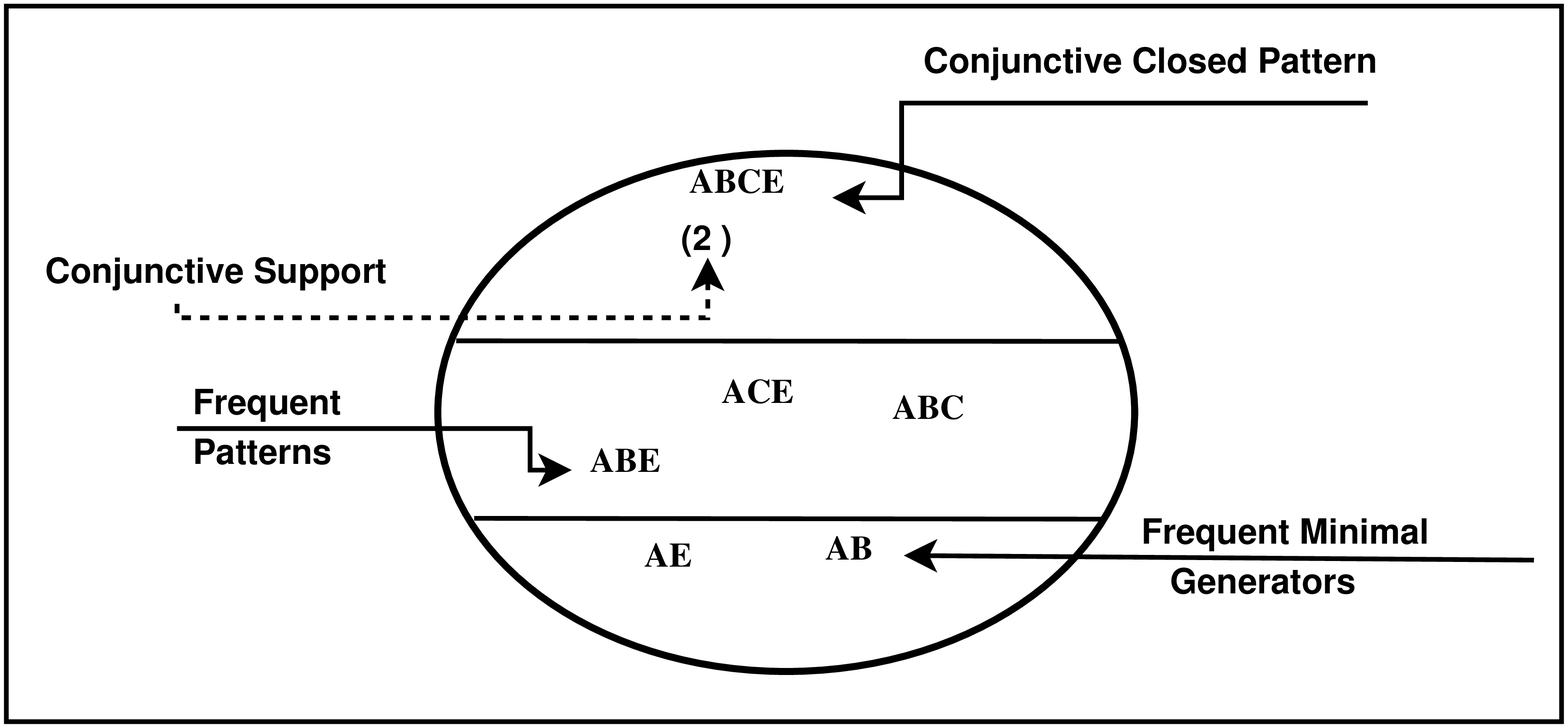}
		\caption{Characterization of a Conjunctive Equivalence Class.}\label{ClSeQuivConj}
	\end{figure}
	
	At this level, we have presented the basic notions related to itemset's extraction and to condensed representations. 

	\section{Conclusion}\label{se4}
	Different approaches, derived from Formal Concept Analysis \textsc{(}FCA\textsc{)}, were proposed in order to reduce the size of the set of frequent itemsets. In addition, correlated pattern mining constitutes an interesting alternative to get more informative patterns with a manageable size and a high quality returned knowledge.  \\
	The next chapter will be dedicated to the presentation, going from the general to the more specific, of the state of the art approaches related to correlated patterns mining. A Comparative study of these approaches will be also conducted.

\chapter{Correlated Patterns Mining: Review of the Literature}\label{ch3}
\section{Introduction} \label{introChap3}
In this chapter, we focus on presenting an overview of the literature approaches,  which are related to our topic of mining correlated patterns. Our study goes from general to more specific. In this respect, we present in Section \ref{CsDM} the approaches related to constraint-based data mining, we deal with the two kinds of constraints. Then, in Section \ref{CP-DM}, we specially concentrate on correlated pattern mining. We start by
introducing the most common correlation measures, then we join with the state of the art of rare correlated patterns mining followed by frequent correlated patterns mining approaches. A synthetic summary of the studied approaches is proposed in Section \ref{seDisc}. The chapter is concluded in Section \ref{ConcChap3}.
\section{Constraint-based Itemset Mining} \label{CsDM}
Within a process of pattern extraction, it is more difficult to localize the set of patterns fulfilling a set of constraints of different natures
than to extract theories associated to a conjunction of constraints of the same nature \cite{luccheKIS05_MAJ_06}.
Indeed, the opposite nature of the constraints makes that the reduction strategies
are applicable to only a part of the constraints and not to all the constraints.
Therefore, the extraction process will be more complicated and more expensive in terms of execution costs and memory greediness.

Many approaches have paid attention to the
extraction of interesting patterns under constraints
\cite{boulicautsurvey_contraintes}.
One of the first algorithms belonging to this context is \textsc{DualMiner} \cite{Bucila03}.
The latter allows the reduction of the search space while considering both of the monotonic and the anti-monotonic constraints.
However, as highlighted by \cite{boley2009}, \textsc{DualMiner} suffers from a main drawback related to the high cost of constraints evaluation.

In \cite{lee_constraint_06}, the authors have proposed an approach of pattern extraction under constraints.
The \textsc{ExAMiner} algorithm \cite{examiner_kais05} was also proposed in order to mine frequent patterns under monotonic constraints.
It is important to mention that the effective reduction strategy adopted by \textsc{ExAMiner} could not be of use in the case of the monotonic constraint
of rarity that we treat in this work, since this latter is sensitive to the changes in the transactions of the extraction context.

Many other works have also emerged. We cite for example, the \textsc{VST} algorithm \cite{vst} which allows the extraction of all the strings satisfying the set of monotonic and
anti-monotonic constraints. Later, the \textsc{FAVST} algorithm \cite{favst} was introduced in order to improve the performance of the \textsc{VST} algorithm  by reducing the number of scans of the database.
Other approaches, belonging to this framework, have also been proposed such as the \textsc{DPC-COFI} algorithm and the \textsc{BifoldLeap} algorithm \cite{bifold}. The strategy of these approaches consists in extracting the maximal frequent itemsets which fulfill all of the constraints and from which the set of all the frequent valid itemsets will be derived.

In \cite{miningzinc-2013}, the authors proposed the \textsc{MiningZinc} framework dedicated to 
constraint programming for itemset mining. The constraints are defined, within the \textsc{MiningZinc} system, in a declarative way close to mathematical notations. The solved tasks within the proposed system concerns closed frequent itemset mining, cost-based itemset mining, high utility itemset mining and discriminative patterns mining.  In a more generic way, in \cite{Guns2016}, the author presented a generic overview of methods devoted to bridge the gap between the two fields of constraint-based itemset mining and constraint programming. 

\section{Correlated Pattern Mining} \label{CP-DM}
This section is dedicated to the study of the correlated pattern mining. First, we start by introducing the commonly used correlation measures, presenting their properties and comparing them.
\subsection{Correlation Measures} \label{MesCorr}
The integration of the correlation measures within the mining process allows to reduce the number of the extracted patterns while improving the quality of the retrieved knowledge. The quality is expressed by the
degree of correlation between the items composing the result itemsets. To achieve this goal,
different correlation measures were proposed in the literature, we start with the \textit{bond} measure.
\smallskip
\subsubsection{3.3.1.1  The \textit{bond} measure}
\smallskip
The \textit{bond} measure \cite{Omie03} is mathematically equivalent to \textit{Coherence} \cite{comine_Lee}, \textit{Tanimoto-coefficient} \cite{Tanimoto1958}, and \textit{Jaccard}.
In \cite{tarekds2010}, the authors propose a new expression of
\textit{bond} in Definition \ref{La mesure bond}.
\begin{definition}\label{La mesure bond} \textbf{The \textit{bond} measure} \\
	The \textit{bond} measure of a non-empty pattern $I$ $\subseteq$ $\mathcal{I}$ is defined as follows:
	\begin{center}
		$\textit{bond}\textsc{(}\textit{I}\textsc{)} = \frac{\displaystyle
			\textit{Supp}\textsc{(}\wedge\textit{I}\textsc{)}}{\displaystyle
			\textit{Supp}\textsc{(}\vee\textit{I}\textsc{)}}$
	\end{center}
\end{definition}

This measure conveys the information about the correlation
of a pattern $I$ by computing the ratio between the number of
co-occurrences of its items and the cardinality of its universe,
which is equal to the transaction set containing a non-empty subset
of $I$. It is worth mentioning that, in the previous works dedicated
to this measure, the disjunctive support has never been used to
express it.

The use of the disjunctive support allows to reformulate the expression of the \textit{bond} measure in order to bring out some pruning conditions for the
extraction of the patterns fulfilling this measure. Indeed, as shown later, the \textit{bond} measure fulfills several properties that offer interesting pruning strategies allowing to reduce the number of generated pattern during the extraction process. Note that the value of the \textit{bond} measure of the empty set is undefined since its disjunctive support is equal to $0$. However, this value is positive since $\lim_{ I\mapsto\emptyset}$ \textit{bond} \textsc{(}\textit{I}\textsc {\textsc{)}} = $\frac{\displaystyle |\mathcal{T}|}{\displaystyle 0}$ = $+\infty$. As a result, the empty set will be considered as a correlated pattern for any minimal threshold of the \textit{bond} correlation measure.


It has been proved, in \cite{tarekds2010}, that the \textit{bond} measure fulfills other interesting properties. In fact, \textit{bond} is: \textsc{(}$i$\textsc{)} \textit{Symmetric} since we have $\forall$ $I$, $J$ $\subseteq$ $\mathcal{I}$, \textit{bond}\textsc{(}$IJ$\textsc{)} = \textit{bond}\textsc{(}$JI$\textsc{)}; \textsc{(}$ii$\textsc{)} \textit{descriptive} \textit{i.e.} is not influenced by the variation of the number of the transactions of the extraction context.

In addition, it has been shown in \cite{HanDMKD2010} that it is desirable to select a descriptive measure which is not influenced by the number of transactions that contain none of pattern items. The symmetric property fulfilled by the \textit{bond} measure makes it possible not to treat all the combinations induced by the precedence order of items within a given pattern. Noteworthily, the anti-monotony property, fulfilled by the \textit{bond} measure as proven in \cite{Omie03}, is of interest. Indeed, all the subsets of a correlated pattern are also necessarily correlated. Then, we can deduce that any pattern having at least one uncorrelated proper subset is necessarily uncorrelated. It will thus be pruned without computing the value of its \textit{bond} measure. In the next definition, we introduce the relationship between the \textit{bond} measure and the cross-support property.

\begin{definition} \textbf{Cross-support property of the \textit{bond} measure} \cite{Xiong06hypercliquepattern}\\
	Thanks to the cross-support property, having a minimal threshold \textit{minbond} and an itemset $I$ $\subseteq$ $\mathcal{I}$, if $\exists$ $x$ and $y$ $\in$ $I$ such as
	$\frac{\displaystyle\textit{Supp}\textsc{(}\wedge
		x\textsc{)}}{\displaystyle\textit{Supp}\textsc{(}\wedge y\textsc{)}}
	< \textit{minbond}$ then
	$I$ is not correlated since \textit{bond}\textsc{(}$I$\textsc{)} < \textit{minbond};
\end{definition}

We continue, in what follows, with the presentation of the \textit{all-confidence} measure.
\smallskip
\subsubsection{ 3.3.1.2  The \textit{all-confidence} measure}
\smallskip
The \textit{all-confidence} measure \cite{Omie03} is defined as follows:
\begin{definition} \textbf{The \textit{all-confidence} measure} \\
	The \textit{all-confidence} measure \cite{Omie03} is defined for any non-empty set $I$ $\subseteq$ $\mathcal{I}$ as follows:
	\begin{center}
		\textit{all-conf}\textsc{(}$I$\textsc{)} =
		$\displaystyle\frac{\displaystyle\textit{Supp}\textsc{(}\wedge
			I\textsc{)}}{\displaystyle \textit{max} \{
			\textit{Supp}\textsc{(}\wedge i\textsc{)} | i\in I \} }$
	\end{center}
\end{definition}
\textit{All-confidence} conserves the anti-monotonic property \cite{Omie03} as well as the cross-support property \cite{Xiong06hypercliquepattern}.
\begin{example}
	Let us consider the extraction context given by Table
	\ref{Base_transactions} \textsc{(}\textit{cf.} page \pageref{Base_transactions}\textsc{)} . For a minimal threshold of \textit{all-confidence} equal to  \textit{0.4}.
	We have \textit{all-confidence}\textsc{(}\texttt{ABCE}\textsc{)} =
	$$\displaystyle\frac{\displaystyle
		\textit{Supp}\textsc{(}\wedge \texttt{ABCE}\textsc{)}}{\displaystyle
		\textit{max} \{ \textit{Supp}\textsc{(}\wedge \texttt{A}\textsc{)},
		\textit{Supp}\textsc{(}\wedge \texttt{B}\textsc{)},
		\textit{Supp}\textsc{(}\wedge \texttt{C}\textsc{)},
		\textit{Supp}\textsc{(}\wedge \texttt{E}\textsc{)}\}}$$
	= $$\displaystyle\frac{\displaystyle 2}
	{\displaystyle \textit{max}\{3, 4\}}$$ = \textit{0.50}.
	The \textsl{ABCE} itemset is correlated according to the \textit{all-confidence} measure.
	All the direct subsets of \textsl{ABCE} are also correlated. We have
	\textit{all-confidence}\textsc{(}\texttt{ABE}\textsc{)} =
	\textit{all-confidence}\textsc{(}\texttt{ACE}\textsc{)} =
	$\displaystyle\frac{\displaystyle 2}
	{\displaystyle 4}$ = \textit{0.50},
	\textit{all-confidence}\textsc{(}\texttt{BCE}\textsc{)} = $\displaystyle\frac{\displaystyle 3}
	{\displaystyle 4}$ = \textit{0.75}.
	
	For the itemset \texttt{AD}, we have
	$\displaystyle\frac{\displaystyle
		\textit{Supp}\textsc{(}\wedge\textsl{D}\textsc{)}}{\displaystyle
		\textit{Supp}\textsc{(}\wedge\textsl{A}\textsc{)}}$ $=$
	$\displaystyle\frac{\displaystyle 1}
	{\displaystyle 3}$
	= \textit{0.33} $<$ \textit{0.4} and we have  \textit{all-confidence}\textsc{(}\texttt{AD}\textsc{)} =
	$\displaystyle\frac{\displaystyle 1}
	{\displaystyle 3}$ = \textit{0.33}.
	The \texttt{AD} itemset does not fulfill the cross-support property, thus it is a non-correlated itemset. This example illustrates the conservation of the anti-monotonicity and the cross-support
	properties of the \textit{all-confidence} measure.
\end{example}
We continue in what follows with the {hyper-confidence} measure.
\smallskip
\subsubsection{3.3.1.3  The \textit{hyper-confidence} measure}
\smallskip
The \textit{hyper-confidence} measure denoted by \textit{h-conf}
of an itemset $I$ $\subseteq$ $\mathcal{I}$ is defined as follows.
\begin{definition} \textbf{The  \textit{hyper-confidence} measure}\\
	The \textit{hyper-confidence} measure of an itemset $I$ $=$ $\{$\textit{$i_{1}$}, \textit{$i_{2}$}, $\ldots$, \textit{$i_{m}$}$\}$
	is equal to:
	\begin{center}
		\textit{h-conf}\textsc{(}$X$\textsc{)}=min$\{$\textit{Conf}\textsc{(} $i_{1}$ $\Rightarrow$ $i_{2}$, $i_{3}$, $\ldots$, $i_{m}$ \textsc{)}, $\ldots$, \textit{Conf}\textsc{(}$i_{m}$ $\Rightarrow$ $i_{1}$, $i_{2}$, $\ldots$, $i_{m-1}$ \textsc{)}$\}$,
	\end{center}
	where \textit{Conf} stands for the \textit{Confidence} measure associated to association rules.
\end{definition}
The \textit{hyper-confidence} measure is equivalent to the \textit{all-confidence} measure, it thus fulfills
the anti-monotonicity and the cross-support properties.

We continue in what follows with the \textit{any-confidence} measure.
\smallskip
\subsubsection{3.3.1.4 The \textit{any-confidence} measure}
\smallskip
This measure is defined, for any non empty set $I$ $\subseteq$ $\mathcal{I}$ as follows:
\begin{definition} \textbf{The \textit{any-confidence} measure}\\
	\begin{center}
		\textit{any-conf}\textsc{(}$I$\textsc{)} = $\displaystyle
		\frac{\displaystyle
			\textit{Supp}\textsc{(}\wedge I\textsc{)}}{\displaystyle
			\textit{min} \{ \textit{Supp}\textsc{(}\wedge i\textsc{)} | i\in I \}}$
	\end{center}
\end{definition}
The \textit{any-confidence} measure \cite{Omie03} does not preserve nor the anti-monotonicity neither the cross-support properties.
\begin{example}
	Let us consider the extraction context given by Table \ref{Base_transactions}. For a minimal correlation threshold equal to \textit{0.80}. The \textit{any-confidence} value of \textsl{AB} is equal to,
	\textit{any-confidence}\textsc{(}\texttt{AB}\textsc{)} =
	$\displaystyle
	\frac{\displaystyle
		\textit{Supp}\textsc{(}\wedge\texttt{AB}\textsc{)}}{\displaystyle
		\textit{min}\{\textit{Supp}\textsc{(} \wedge \texttt{A}\textsc{)},
		\textit{Supp}\textsc{(}\wedge \texttt{B}\textsc{)} \}}$
	= $\displaystyle
	\frac{\displaystyle 2}
	{\displaystyle \textit{min}\{3, 4\}}$
	= \textit{0.66}. \textsl{AB} do not fulfill the minimal threshold of correlation,
	thus it is a non-correlated itemset according to the \textit{any-confidence} measure. Whereas, the
	\textsl{AD} itemset is correlated and its correlation value is equal to 1.
	We also have,
	$\displaystyle
	\frac{\displaystyle
		\textit{Supp}\textsc{(}\wedge\textsl{A}\textsc{)}}{\displaystyle
		\textit{Supp}\textsc{(}\wedge\textsl{C}\textsc{)}}$ $=$
	$\displaystyle
	\frac{\displaystyle 3}
	{\displaystyle 4}$ = \textit{0.75} $<$ \textit{0.80}, however,
	\textit{any-confidence}\textsc{(}\textsl{AD}\textsc{)} $=$
	\textit{1} $>$ \textit{0.80}.
	This example illustrates the non preservation of the anti-monotonicity as well as the cross-support properties.
\end{example}

We present in what follows the \textit{$\chi^2$} Coefficient.
\smallskip
\subsubsection{3.3.1.5 The \textit{$\chi^2$} Coefficient}
\smallskip
The \textit{$\chi^2$} coefficient is defined as follows :
\begin{definition} \textbf{The \textit{$\chi^2$} Coefficient} \cite{Brin97}\\
	The \textit{$\chi^2$} coefficient of an itemset $Z$ $=$ $xy$, with $x$ and $y$ $\in$ $\mathcal{I}$, is defined as follows:
	\begin{center}
		\textit{$\chi^2$}\textsc{(}$Z$\textsc{)} $=$ $|\mathcal{T}|$
		$\times$ $\frac{\displaystyle
			\textsc{(}\textit{Supp}\textsc{(}\wedge xy\textsc{)} -
			\textit{Supp}\textsc{(}\wedge
			x\textsc{)}\times\textit{Supp}\textsc{(}\wedge
			y\textsc{)}\textsc{)}^{2}}{\displaystyle\textit{Supp}\textsc{(}\wedge
			x\textsc{)}\times\textit{Supp}\textsc{(}\wedge
			y\textsc{)}\times\textsc{(}1 - \textit{Supp}\textsc{(}\wedge
			x\textsc{)}\textsc{)}\times\textsc{(}1 -
			\textit{Supp}\textsc{(}\wedge y\textsc{)}\textsc{)}}$
	\end{center}
\end{definition}
Some relevant properties of the $\chi^2$ coefficient are given by the following proposition.
\begin{proposition} The $\chi^2$ coefficient is a statistic and symmetric measure \cite{Brin97}.
\end{proposition}

Other correlation measures are also of use in the literature, we mention for example
the \textit{cosine} measure, the \textit{lift} measure \cite{Brin97}, the $\phi$ coefficient also named the Pearson coefficient \cite{XiongKDD2004}.
\smallskip
\subsubsection{3.3.1.6 Synthesis}
\smallskip
We recapitulate the different properties of the presented measures in Table \ref{tab_recapitulatif_prop}.
The ``$\checkmark$'' symbol indicates that the measure fulfills the property.

In our previous study, we specifically focused on correlation measures which are most used in correlated patterns mining.
Withal, the \textit{cosine} and the \textit{kulczynski} measures were not studied since these two measures
are rarely used on correlated patterns mining due to the non conservation of the anti-monotonicity property \cite{HanDMKD2010}. The  \textit{lift} measure is used within the association rule evaluation.
\begin{table}\begin{center}\small{
			\begin{tabular}{|l||c|c|c|c|}
				\hline
				Measure  & Independence of $|\mathcal{T}|$ & Symmetry & Anti-monotonicity & Cross-support
				\\ \hline \hline
				\textit{bond}           & $\checkmark$ & $\checkmark$ & $\checkmark$ & $\checkmark$ \\ \hline
				\textit{any-confidence} & $ \checkmark$ & $\checkmark$ &         &   \\ \hline
				\textit{all-confidence} & $\checkmark$ & $\checkmark$ & $\checkmark$ & $\checkmark$ \\ \hline
				\textit{hyper-confidence} &$\checkmark$& $\checkmark$ & $\checkmark$ & $\checkmark$ \\ \hline
				$\chi^{2}$             &          & $\checkmark$ &       &  \\
				\hline
		\end{tabular}}
	\end{center}
	\caption{Summary of the properties of the studied correlation measures and coefficients.}\label{tab_recapitulatif_prop}
\end{table}
We conclude, according to this overview, that the most interesting  measures are
\textit{bond} and \textit{all-confidence}. This is justified by the fact that these two measures fulfilled
the pertinent properties of anti-monotonicity and cross-support.

We present, in what follows, the state of the art approaches dealing with correlated patterns mining.
We precisely start with rare correlated pattern mining.
\subsection{Rare Correlated Patterns Mining} \label{EdeAMCR}
Various approaches devoted to the  extraction of correlated patterns under constraints have been proposed. However, the recuperation of all the patterns that are
both highly correlated and infrequent is based on the naive idea to extract the set of all frequent patterns for a very low threshold \textit{minsupp} and then to filter out these patterns by a measure of correlation.

Another idea is to extract the whole set of the correlated patterns without any integration of the rarity constraint. The obtained set contains obviously all the frequent correlated as well as the rare correlated patterns. It is relevant to note that the application of these two ideas is very expensive in execution time and in memory consumption due to the explosion of the number of candidates to be evaluated.

The approach proposed in \cite{Cohen_mcr_2000}
is based on the previous principle. This approach allows to extract the items's pairs correlated according to the \textit{Similarity} measure but without computing their support. In fact, the \textit{Similarity} measure allows to evaluate the similarity between two items and corresponds to the quotient of the number of the simultaneous appearance divided by the
number of the complementary appearance. Consequently, the \textit{Similarity} measure is semantically equivalent to the  \textit{bond} measure. However, any analysis of this measure have been conducted. \\
In fact, this approach  proposes to assign to each item a signature composed by the identifier list of the transactions to which the item belongs. Then, the \textit{Similarity} is computed and it corresponds to the
number of the intersections of their signatures divided by the union of their signatures. We conclude that the frequency constraint was not integrated in order to recuperate the highly correlated itemsets with a
weak support. From these patterns, the association rules with a high confidence and a weak support are generated.

In this same context, we mention the \textsc{DiscoverMPatterns} algorithm \cite{ma-icdm2001}.
In fact, this latter is devoted to the extraction of the correlated patterns based on the \textit{all-confidence} measure. Nevertheless, a first version of the approach was
dedicated to the extraction of all the correlated patterns without any restriction of the support value in order to specifically get the rare correlated itemsets. Then, within the second version of the approach, the minimum support threshold constraint was integrated. Consequently, this constraint integration allows to extract the frequent correlated patterns.

Another principle of the resolution of the rare correlated patterns extraction consists in
extracting all the frequent patterns for a very weak minimal support threshold. Evidently, the obtained set contains a subset of the infrequent correlated patterns.
Xiong et al. relied on this idea to introduce the \textsc{Hyper-CliqueMiner} algorithm
\cite{Xiong06hypercliquepattern}. The output of this algorithm is the set of frequent correlated patterns for a very low \textit{minsupp} value. It is to note, that the good performances of this algorithm are
justified by the use of the anti-monotonic property of the correlation measure as well as the \textit{cross-support} property which allows to reduce significantly the evaluated candidates and thus to reduce the time needed.

The approach proposed in \cite{thomo_2010} stands also within this principle. This approach allows to extract the frequent and frequently correlated 2-itemsets. It is judged as a naive approach that is based on the extraction of all the  solution set for a very low \textit{minsupp} values. Then, a post processing is performed in order to maintain only the high correlated itemsets. The \textsc{FT-Miner} algorithm \cite{ptminer} outputs the correlated infrequent itemsets according to the \textit{N-Confidence} semantically equivalent to \textit{all-Confidence}. The \textit{all-Confidence} measure was also treated in the
\textsc{Partition} algorithm \cite{Omie03}, which allows to extract the correlated patterns
according to both \textit{all-Confidence} and \textit{bond} measures. The choice of the measure to be considered depends on the user's input preferences.

The approach proposed in \cite{MCROkubo} also belongs to the same trend of approaches dealing with correlated infrequent itemsets. Indeed, it is based on the principle that the patterns which are weakly correlated according to the \textit{bond} correlation measure are generally rare in the extraction context.
The expressed constraint corresponds to a restriction of the maximum correlation value. This is a monotonic constraint since it corresponds to the opposite of
the anti-monotonic constraint of minimal correlation. In order to get rid from rare patterns that represent exceptions, and they are not informative, a minimal frequency constraint was also integrated.
The idea consists then in extracting the top$-N$ rare patterns which are the most informative ones.

The problem of integrating constraints during the process of correlated pattern mining was also studied in the works, respectively, proposed in \cite{Brin97} and in \cite{grahne_correlated_2000}.
These approaches deal with constrained correlated pattern mining, they rely on the $\chi^2$ correlation coefficient.
They exploit the various pruning opportunities offered by these constraints and benefit from the selective power of each type of constraints. However, the coefficient $\chi^2$ does not fulfill the anti-monotonic constraint as does the \textit{bond} measure.
Besides, these approaches are limited to the extraction of a small subset which is composed only by minimal valid patterns \textit{i.e.} the minimal patterns which fulfill all of the imposed constraints. Furthermore, the authors do not propose any concise representation of the extracted correlated patterns.

Also, in \cite{surana2010}, a study of different properties of interesting measures was conducted in order to suggest a set of the most adequate properties to consider while mining rare associations rules.

It is deduced that for all these approaches, the monotonic constraint of rarity was never included within the mining process in order to retrieve all the rare highly correlated patterns.

\subsection{Frequent Correlated Patterns Mining} \label{EdeAMCF}

In \cite{comine_Lee}, the authors proposed the \textsc{CoMine} approach which is dedicated to the extraction of frequent correlated patterns according to the \textit{all-confidence} and to the \textit{bond} measures.
We distinguish  two different versions of the \textsc{CoMine} approach. The first version treats the
\textit{bond} measure while the second treats the \textit{all-Confidence} measure. \textsc{CoMine} also constitute the core of the \textsc{I-IsCoMine-AP} and \textsc{I-IsCoMine-CT} algorithms \cite{IsComine2011}.

Also, the \textit{bond} measure was studied in \cite{LeBras2011}, the authors proposed an apriori-like algorithm for mining classification rules. Moreover, the authors in \cite{borgelt} proposed a generic approach for frequent correlated pattern mining. Indeed, the \textit{bond} correlation measure and eleven other correlation measures were used. All of them fulfill the anti-monotonicity property. Correlated patterns mining was then shown to be more complex and more informative than frequent pattern mining \cite{borgelt}.

Many other works have also emerged. In \cite{HanDMKD2010}, the authors provide a unified definition of existing null-invariant correlation measures and propose the \textsc{GAMiner} approach allowing the extraction of frequent high correlated patterns according to the \textit{Cosine} and to the \textit{Kulczynsky} measures.
In this same context, the \textsc{NICOMiner} algorithm was also proposed in \cite{Kimpkdd2011} and it allows the extraction of correlated patterns according to the \textit{Cosine} measure. We highlight that the \textit{Cosine} measure has the specificity of being not monotonic neither anti-monotonic.

In this same context, we also cite the \textsc{Atheris} approach \cite{skypattern2011} which allows the extraction of condensed representation of correlated patterns according to user's preferences.
In \cite{FlipPattern2012}, the authors introduced the concept of flipping correlation patterns according to the \textit{Kulczynsky} measure. However, the \textit{Kulczynsky} measure does not fulfill the interesting anti-monotonic property as the \textit{bond} measure.

The \textit{all-confidence} measure was handled within the work proposed in \cite{Karim-etal-2012}. The approach outputs the correlated patterns \textsc{(}also called the associated patterns\textsc{)}, the non correlated patterns \textsc{(}also called the independent patterns\textsc{)}. Also, in \cite{pakdd2013} the authors propose a method to extract \textit{all-confidence} frequent correlated patterns and they also discuss the impact of fixing the \textit{minsupp} threshold value over the quality of the obtained itemsets and propose to fix a minimal correlation threshold for each item.

In the next subsection, we study the approaches of extracting the condensed representations of frequent correlated patterns.
\subsection{Condensed Representations of Correlated Patterns Mining}
The problem of mining concise representations of correlated patterns was not widely studied in the literature.
We mention  the \textsc{Ccmine} \cite{ccmine_Kim} approach of mining closed correlated patterns according to the \textit{all-confidence} measure which constitute a condensed representation of frequent correlated patterns.
We also precise that the authors in \cite{tarekds2010} proposed the \textsc{CCPR-Miner} algorithm allowing the extraction of closed frequent correlated patterns according to the \textit{bond} measure.

In this context, we also cite the $\textsc{Jim}$ approach \cite{borgelt}. In fact, $\textsc{Jim}$ allows to extract the closed correlated frequent patterns which constitute a perfect cover of the whole set of frequent correlated patterns. The choice of the considered correlation measure is fixed by the user's parameters within the $\textsc{Jim}$ approach.

In fact, the $\textsc{Jim}$ approach is, on the one hand the most efficient state of the art approach extracting condensed representation of frequent correlated patterns according to the \textit{bond} measure. On the other hand, $\textsc{Jim}$ is the unique approach which 
dealt with the same kind of patterns as we treat in our mining approach, that we present in the following chapters.  In this sense, in our experimental study, we will focus on comparing our mining approach by the $\textsc{Jim}$ approach.

\section{Discussion} \label{seDisc}
Based on the previous review of the literature, we conclude that most of the approaches dealt with the
\textit{bond} and the \textit{all-confidence} measures. These latter fulfill the interesting anti-monotonic property, that allows to reduce the search space by early pruning irrelevant candidates. Therefore, the frequent correlated set of patterns results from the conjunction of both constraints of the same type: the correlation and the frequency.

In fact, the recuperation of all the patterns that are both highly correlated and infrequent is based on the naive idea to extract the set of all frequent patterns for a very low threshold \textit{minsupp} and then to filter out these patterns by a measure of correlation.  Another resolution strategy consists in extracting the whole set of the correlated patterns without any integration of the rarity constraint.
Then, a post-processing is performed in order to uniquely retrieve the rare correlated itemsets.

In other words, the monotonic constraint of rarity was never integrated within the mining process and thus the exploration of the search space of candidates that does not fulfill the rarity constraint is obviously barren.
In addition, another problem is related to the high consuming of the memory and the CPU resources due to the combinatorial explosion of the number of candidates depending on the size of the mined dataset. We highlight, that \textsc{Jim} \cite{borgelt} is the unique approach that dealt with different anti-monotonic correlation measures. However, \textsc{Jim} is limited to frequent correlated patterns and do not consider the rare correlated ones.

Table \ref{TabEdeA} recapitulates the characteristics of the different visited approaches. This table summarizes the following properties:
\begin{enumerate}
	\item \textbf{The correlation measure:} This property describes the considered correlation measure.
	\item \textbf{The kind of the extracted patterns:} This property describes the kind of patterns outputted by the mining algorithm
	\item \textbf{The nature of constraints:} This property describes the nature of the constraints included within the algorithm: anti-monotonic or monotonic.
\end{enumerate}

To the best of our knowledge, no previous work was dedicated to the extraction of concise
representations of patterns under the conjunction of constraints of distinct types. This problem is then a challenging task in data mining, which strengthens our motivation for the treatment of this problematic. Therefore, the work proposed in this thesis is the first one that puts the focus on mining concise representations of both frequent and rare correlated patterns according to the anti-monotonic \textit{bond} measure.
\begin{table}[!t]\small{
		\begin{center}
			\footnotesize
			\begin{tabular}{|l||c|c|c|}
				\hline
				\multicolumn{1}{|c||}{\textbf{Extraction}}   &  \textbf{Correlation}  & \textbf{Kind of the extracted} &\textbf{Nature of}\\
				\multicolumn{1}{|c||}{\textbf{Algorithm}}    &\textbf{Measure}           & \textbf{Patterns}
				& \textbf{Constraints} \\
				\hline\hline
				The approach of & \textit{bond} & correlated & anti-monotonic  \\
				\cite{Cohen_mcr_2000} &       & 2-itemsets &    \\
				\hline	
				\textsc{DiscoverMPattern}& \textit{all-confidence} & all the correlated & anti-monotonic \\
				\cite{ma-icdm2001} &                                & itemsets &    \\
				\hline
				\textsc{Partition}  &\textit{all-confidence} & all the correlated & anti-monotonic\\
				\cite{Omie03}&\textit{bond}& itemsets                      &            \\
				\hline
				\textsc{CoMine}\textsc{(}$\alpha$\textsc{)}&\textit{all-confidence} & correlated frequent & anti-monotonic \\
				\cite{comine_Lee}&         &                      &            \\
				\hline
				\textsc{CoMine}\textsc{(}$\gamma$\textsc{)}&\textit{bond}  & correlated frequent& anti-monotonic \\
				\cite{comine_Lee}&         &                      &            \\
				\hline
				\textsc{CCMine}&\textit{all-confidence} & closed & anti-monotonic \\
				\cite{ccmine_Kim}&  & frequent correlated&      \\
				\hline
				\textsc{HypercliqueMiner} & \textit{h-confidence} & correlated frequent & anti-monotonic\\
				\cite{Xiong06hypercliquepattern} &          & and a subset of rare &            \\
				&          & correlated itemsets &                       \\	
				\hline
				The approach of & \textit{all-confidence} & correlated frequent & \\
				\cite{thomo_2010} &                     & and a subset of rare & anti-monotonic\\
				&                      & correlated itemsets &\\	
				\hline
				The approach of  & \textit{bond} & weakly correlated & monotonic \\
				\cite{MCROkubo}&               &&              \\
				\hline
				\textsc{CCPR\_Miner} &\textit{bond}&closed & anti-monotonic \\
				\cite{tarekds2010}   &             & frequent correlated    & \\
				\hline
				\textsc{Jim} &\textit{Eleven different}&       closed & anti-monotonic \\
				\cite{borgelt}   &  \textit{anti-monotonic }          & frequent correlated    & \\
				& \textit{measures}& and frequent correlated & \\
				\hline
			\end{tabular}
			\caption{Comparison between the correlated patterns mining approaches.}\label{TabEdeA}
		\end{center}
	}
\end{table}
\bigskip
\section{Conclusion} \label{ConcChap3}
In this chapter, we proposed an overview of the state of the art approaches dealing with
correlated patterns mining preceded by a presentation of the different correlation measures.
We deduced that, there is no previous work that dealt with both frequent correlated as well as rare correlated patterns according to a specified correlation metric. Thus, motivated by this issue, we propose in this thesis to benefit from the knowledge returned from both frequent correlated as well as rare correlated patterns according to the \textit{bond} measure. To tackle this challenging task, we propose
in the next chapter the characterization of both  frequent correlated patterns, rare correlated patterns and their associated concise representations.

\part{Condensed Representations of Correlated Patterns}
\chapter{Condensed Representations of Correlated Patterns}\label{ch_4}

\section{Introduction} \label{introChap4}
The main moan that can be related to frequent pattern mining approaches stands in the fact that the latter do not offer the information concerning the correlation degree among the items in the extraction context.  This stands behind our motivation to provide to the user the key information about the correlation between items as well as the frequency of their occurrence.
This aim is reachable thanks to the integration of the correlation measures within the mining process.

The correlation measure, that we treat throughout this thesis, is \textit{bond}.
Our motivations behind the choice of this measure is explicitly described in Section \ref{S2chap4}.
In Section \ref{sebd}, we focus on the correlated patterns associated to the \textit{bond} measure, we characterize this set of patterns.
Section \ref{sec_FBond} is devoted to the presentation of the closure operator associated to \textit{bond}. We introduce the associated exact condensed representations in Section \ref{section_RC} and in
Section \ref{section_RCfcp}. Section \ref{ConcChap4} concludes the chapter.
\section{Motivations behind our choice of the \textit{bond} measure} \label{S2chap4}
Based on the study of the state of the art approaches proposed in the previous chapter, we find that the almost of the existing approaches are dealing with the \textit{bond} and the \textit{all-confidence} measures.
The \textit{bond} measure fulfills the anti-monotony property which is an interesting property. Indeed, the latter
reduce the search space when pruning the non potential candidates, therefore optimizing the extraction time as well as the memory consumption.

It has been proved in the literature that the \textit{bond} measure presents many interesting properties. In fact, the \textit{bond} measure is:
\begin{enumerate}
	\item \textit{symmetric} since $\forall$ $I$, $J$ $\subseteq$ $\mathcal{I}$, \textit{bond}\textsc{(}$IJ$\textsc{)} = \textit{bond}\textsc{(}$JI$\textsc{)};
	\item \textit{descriptive} since it is not influenced by the number of transactions that contain none of the items composing the pattern;
	\item fulfills the \textit{cross-support} property \cite{Xiong06hypercliquepattern}. Thanks to this property, given a minimal threshold \textit{minbond} and an itemset $I$ $\subseteq$ $\mathcal{I}$, if $\exists$ $x$ and $y$ $\in$ $I$ such as
	$\frac{\displaystyle\textit{Supp}\textsc{(}\wedge
		x\textsc{)}}{\displaystyle\textit{Supp}\textsc{(}\wedge y\textsc{)}}
	< \textit{minbond}$ then
	$I$ is not correlated since \textit{bond}\textsc{(}$I$\textsc{)} $<$ \textit{minbond};
	\item induces an \textit{anti-monotonic} constraint for a fixed minimal threshold \textit{minbond}.
	In fact, $\forall$ $I$, $I_1$ $\subseteq$ $\mathcal{I}$, if $I_1$ $\subseteq$ $I$, then \textit{bond}\textsc{(}$I_1$\textsc{)} $\geq$ \textit{bond}\textsc{(}$I$\textsc{)}.
	Therefore, the set $\mathcal{CP}$ of correlated patterns forms an order ideal. Indeed, all the subsets of a correlated pattern are necessarily correlated ones.
\end{enumerate}
We present in the following an interesting relation between the value of the \textit{bond} measure and the conjunctive and disjunctive supports values for each couple of two patterns $I$ and $I_1$ such as $I$ $\subseteq$ $I_1$ \cite{tarekds2010}.
\begin{proposition} \label{PropBond} Let $I$, $I_1$ $\subseteq$ $\mathcal{I}$ and $I$ $\subseteq$ $I_1$. If  \textit{bond}\textsc{(}$I$\textsc{)} =
	\textit{bond}\textsc{(}$I_1$\textsc{)}, then
	\textit{Supp}\textsc{(}$\wedge I$\textsc{)} =
	\textit{Supp}\textsc{(}$\wedge I_1$\textsc{)} and
	\textit{Supp}\textsc{(}$\vee I$\textsc{)} =
	\textit{Supp}\textsc{(}$\vee I_1$\textsc{)}.
\end{proposition}
According to the previous proposal, if  \textit{bond}\textsc{(}$I$\textsc{)} =
\textit{bond}\textsc{(}$I_1$\textsc{)}, then \textit{Supp}\textsc{(}$\neg I$\textsc{)} =
\textit{Supp}\textsc{(}$\neg I_1$\textsc{)}.
In fact, both $I$ and $I_1$
have the same conjunctive support and, according to the Morgan law, we build the following relation between the disjunctive and the negative supports of a pattern: \textit{Supp}\textsc{(}$\neg I$\textsc{)} = $|\mathcal{T}|$ - \textit{Supp}\textsc{(}$\vee I$\textsc{)}.
On the other hand, if \textit{bond}\textsc{(}$I$\textsc{)} $\neq$
\textit{bond}\textsc{(}$I_1$\textsc{)}, then
\textit{Supp}\textsc{(}$\wedge I$\textsc{)} $\neq$
\textit{Supp}\textsc{(}$\wedge I_1$\textsc{)} or
\textit{Supp}\textsc{(}$\vee I$\textsc{)} $\neq$
\textit{Supp}\textsc{(}$\vee I_1$\textsc{)} \textsc{(}\textit{i.e.} one of the two supports is different or both\textsc{)}.\\

In this context, we propose to study the \textit{bond} correlation measure in an integrated mining process aiming to extract both frequent and rare correlated patterns as well as their associated condensed representations. In this regard, we present in the next section the specification of the frequent correlated patterns as well as the rare correlated patterns according to the \textit{bond} measure.
\section{Characterization of the Correlated patterns according to the \textit{bond} measure}\label{sebd}
\subsection{Definitions and Properties}
The \textit{bond} measure \cite{Omie03}  is mathematically equivalent to
\textit{Coherence} \cite{comine_Lee}, \textit{Tanimoto
	coefficient} \cite{Tanimoto1958}, and \textit{Jaccard}
\cite{jaccard_1901}. It was redefined in \cite{tarekds2010} as follows:
\begin{definition}\label{Defbond} \textbf{The \textit{bond} measure}\\
	The \textit{bond} measure of a non-empty pattern $I \subseteq \mathcal{I}$ is defined as follows:
	\begin{center}
		$\textit{bond}\textsc{(}\textit{I}\textsc{)} = \frac{\displaystyle
			\textit{Supp}\textsc{(}\wedge\textit{I}\textsc{)}}{\displaystyle
			\textit{Supp}\textsc{(}\vee\textit{I}\textsc{)}}$
	\end{center}
\end{definition}
The \textit{bond} measure takes its values within the interval $[0,1]$. While considering the universe of a pattern
$\mathcal{I}$ \cite{comine_Lee}, \textit{i.e.}, the set of transactions containing a non empty subset of $I$,
the \textit{bond} measure represents the simultaneous occurrence rate of the items
of the pattern $I$ in its universe.
Thus, the higher the items of the pattern are dispersed in its universe, \textsc{(}\textit{i.e.} weakly correlated\textsc{)},
the lower the value of the \textit{bond} measure is, as \textit{Supp}\textsc{(}$\wedge I$\textsc{)} is  smaller than \textit{Supp}\textsc{(}$\vee I$\textsc{)}.
Inversely, the more the items of $I$ are dependent from each other, \textsc{(}\textit{i.e.} strongly correlated\textsc{)}, the higher the value of the
\textit{bond} measure is, since  \textit{Supp}\textsc{(}$\wedge I$\textsc{)} would be closer to \textit{Supp}\textsc{(}$\vee I$\textsc{)}.

The set of correlated patterns associated to the \textit{bond} measure is defined as follows.
\begin{definition}\label{defCP}  \textbf{Correlated patterns}\\
	Considering a correlation threshold \textit{minbond}, the set of correlated patterns, denoted $\mathcal{CP}$, is equal to:
	$\mathcal{CP}$ = $\{I$ $\subseteq \mathcal{I}|$
	\textit{bond}\textsc{(}$I$\textsc{)} $\geq$ \textit{minbond}$\}$.
\end{definition}
\begin{example}
	Let us consider the dataset given by Table \ref{Base_transactions}. For \textit{minbond} = 0.5, we have \textit{bond}\textsc{(}$AB$\textsc{)} = $\displaystyle\frac{2}{5}$ = 0.4 $<$ 0.5. The itemset $AB$ is then not a correlated one.
	Whereas, since \textit{bond}\textsc{(}$BCE$\textsc{)} = $\displaystyle\frac{3}{5}$ = 0.6 $\geq$ 0.5, the itemset $BCE$ is a correlated one.
\end{example}

In the following, we define the set composed by the  maximal correlated patterns as follows:
\begin{definition} \label{bdpos} \textbf{Maximal correlated patterns}\\
	The set of maximal correlated patterns constitutes the positive border of correlated patterns and is composed by correlated patterns having no correlated proper superset.
	This set is defined as:  $\mathcal{M}ax$$\mathcal{CP}$ = $\{$$I$ $\in$ $\mathcal{CP}$$\mid$ $\forall$ $I_1$ $\supset$ $I$: $I_1$ $\notin$ $\mathcal{CP}$$\}$, or equivalently:
	$\mathcal{M}ax$$\mathcal{CP}$ = $\{$$I$ $\in$ $\mathcal{CP}$$\mid$ $\forall$ $I_1$ $\supset$ $I$:
	\textit{bond}\textsc{(}$I_1$\textsc{)} $<$ \textit{minbond}$\}$.
\end{definition}
\begin{example}\label{example_MCMax}
	Consider the extraction context sketched by Table \ref{Base_transactions} \textsc{(Page \pageref{Base_transactions})}. For \textit{minbond} = 0.2, we have $\mathcal{M}ax$$\mathcal{CP}$  = $\{$$ACD$, $ABCE$$\}$.
\end{example}

As far as we integrate the frequency constraint with the correlation constraint, we can distinguish between two sets of correlated patterns, which are the "Frequent correlated patterns" set and the "Rare correlated patterns" set. These two distinct sets will be characterized separately in the remainder.
\subsection{Frequent Correlated Patterns}
\begin{definition}\label{defCFP} \textbf{The set of frequent correlated patterns}\\
	Considering the support threshold \textit{minsupp} and the
	correlation threshold \textit{minbond}, the set of frequent
	correlated patterns, denoted $\mathcal{FCP}$, is equal to:
	$\mathcal{FCP}$ = $\{$$I$ $\subseteq$ $\mathcal{I}$ $|$
	\textit{Supp}\textsc{(}$\wedge I$\textsc{)} $\geq$
	\textit{minsupp} and \textit{bond}\textsc{(}$I$\textsc{)} $\geq$ \textit{minbond}$\}$.
\end{definition}
In fact, the $\mathcal{FCP}$ set is composed by the patterns fulfilling at the same time the correlation and the frequency constraints. A pattern is said to be ``Frequent Correlated'' if its support exceeds the minimal frequency threshold \textit{minsupp} and its correlation value also exceeds the
minimal correlation threshold \textit{minbond}. The $\mathcal{FCP}$ set corresponds to the conjunction of two anti-monotonic constraints of correlation and of frequency. Thus, it induces an order ideal on the itmeset lattice.
\begin{example}\label{example_set_FCP}
	Consider the extraction context sketched by Table \ref{Base_transactions} \textsc{(Page \pageref{Base_transactions})}. For \textit{minsupp} = 4 and \textit{minbond} = 0.2,
	the $\mathcal{FCP}$ set consists of the following patterns where each triplet represents the pattern, its conjunctive support value and its \textit{bond} value:
	$\mathcal{FCP}$ = $\{$\textsc{(}$B$, 4, $\displaystyle\frac{4}{4}$\textsc{)},
	\textsc{(}$C$, 4, $\displaystyle\frac{4}{4}$\textsc{)},
	\textsc{(}$E$, 4, $\displaystyle\frac{4}{4}$\textsc{)},
	\textsc{(}$BE$, 4, $\displaystyle\frac{4}{4}$\textsc{)}$\}$.
\end{example}
\subsection{Rare Correlated Patterns}
The set of rare correlated patterns associated to the \textit{bond} measure is defined as follows.
\begin{definition}\label{defCRP} \textbf{The set of rare correlated patterns}\\
	Considering the support threshold \textit{minsupp} and the
	correlation threshold \textit{minbond}, the set of rare correlated patterns, denoted $\mathcal{RCP}$, is equal to:
	$\mathcal{RCP}$ = $\{I$ $\subseteq$ $\mathcal{I}$ $|$
	\textit{Supp}\textsc{(}$\wedge I$\textsc{)} $<$
	\textit{minsupp} and \textit{bond}\textsc{(}$I$\textsc{)} $\geq$ \textit{minbond}$\}$.
\end{definition}
\begin{example}\label{example_set_CRP}
	Consider the extraction context sketched by Table \ref{Base_transactions} \textsc{(Page \pageref{Base_transactions})}. For \textit{minsupp} = 4 and \textit{minbond} = 0.2,
	the set $\mathcal{RCP}$ consists of the following patterns where each triplet represents the pattern, its conjunctive support value and its \textit{bond} value:
	$\mathcal{RCP}$ = $\{$\textsc{(}$A$, 3, $\displaystyle\frac{3}{3}$\textsc{)},
	\textsc{(}$D$, 1, $\displaystyle\frac{1}{1}$\textsc{)},
	\textsc{(}$AB$, 2, $\displaystyle\frac{2}{5}$\textsc{)},
	\textsc{(}$AC$, 3, $\displaystyle\frac{3}{4}$\textsc{)},
	\textsc{(}$AD$, 1, $\displaystyle\frac{1}{3}$\textsc{)},
	\textsc{(}$AE$, 2, $\displaystyle\frac{2}{5}$\textsc{)},
	\textsc{(}$BC$, 3, $\displaystyle\frac{3}{5}$\textsc{)},
	\textsc{(}$CD$, 1, $\displaystyle\frac{1}{4}$\textsc{)},
	\textsc{(}$CE$, 3, $\displaystyle\frac{3}{5}$\textsc{)},
	\textsc{(}$ABC$, 2, $\displaystyle\frac{2}{5}$\textsc{)},
	\textsc{(}$ABE$, 2, $\displaystyle\frac{2}{5}$\textsc{)},
	\textsc{(}$ACD$, 1, $\displaystyle\frac{1}{4}$\textsc{)},
	\textsc{(}$ACE$, 2, $\displaystyle\frac{2}{5}$\textsc{)},
	\textsc{(}$BCE$, 3, $\displaystyle\frac{3}{5}$\textsc{)},
	\textsc{(}$ABCE$, 2, $\displaystyle\frac{2}{5}$\textsc{)}$\}$.
	This associated $\mathcal{RCP}$ set as well as the $\mathcal{FCP}$ set of the previous example are depicted by Figure \ref{figure_CRP}.
	The support shown at the top left of each frame represents the conjunctive one.
	As shown in Figure \ref{figure_CRP}, the rare correlated patterns are localized below the
	border induced by the anti-monotonic constraint of correlation and over the border
	induced by the monotonic constraint of rarity.
\end{example}
We deduce from Definition \ref{defCRP} that the $\mathcal{RCP}$ set corresponds to the intersection between the set $\mathcal{CP}$
of correlated patterns and the set $\mathcal{RP}$
of rare patterns, \textit{i.e.},  $\mathcal{RCP}$ = $\mathcal{CP}$ $\cap$ $\mathcal{RP}$.
The following proposition derives from this result.
\begin{proposition}\label{prop_set_CRP}
	Let $I$ $\in$ $\mathcal{RCP}$. We have:
	\begin{itemize}
		\item  Based on the order ideal of the set $\mathcal{CP}$ of correlated patterns, we have $\forall$ $I_1$ $\subseteq$ $I$: $I_1$ $\in$ $\mathcal{CP}$
		\item  Based on the order filter of the set $\mathcal{RP}$ of rare patterns, we have $\forall$ $I_1$ $\supseteq$ $I$: $I_1$ $\in$ $\mathcal{RP}$.
	\end{itemize}
\end{proposition}
\begin{proof}
	The proof follows from the properties induced by the constraints of rarity and correlation. The set $\mathcal{RCP}$, whose elements fulfill
	the constraint ``being a rare correlated pattern'', results from the conjunction between two theories corresponding to both constraints of distinct types.
	So, the set $\mathcal{RCP}$ is neither an order ideal nor an order filter. The search space of this set is delimited
	by: \textsc{(}i\textsc{)} The maximal correlated elements which are also rare,  \textit{i.e.} the rare patterns among \textit{the set $\mathcal{M}$$ax$$\mathcal{CP}$} of maximal correlated patterns
	\textsc{(}\textit{cf.} Definition \ref{bdpos}\textsc{)} and;
	\textsc{(}ii\textsc{)} The minimal rare elements which are correlated, \textit{i.e.} the correlated patterns among \textit{the set $\mathcal{M}$$in$$\mathcal{RP}$} of minimal rare patterns
	\textsc{(}\textit{cf.} Definition \ref{mrp}\textsc{)}.
	Therefore, each rare correlated pattern is necessarily included between an element from each set of the two aforementioned sets.
\end{proof}
\begin{figure}[htbp]
	\begin{center}
		\includegraphics[scale = 0.4]{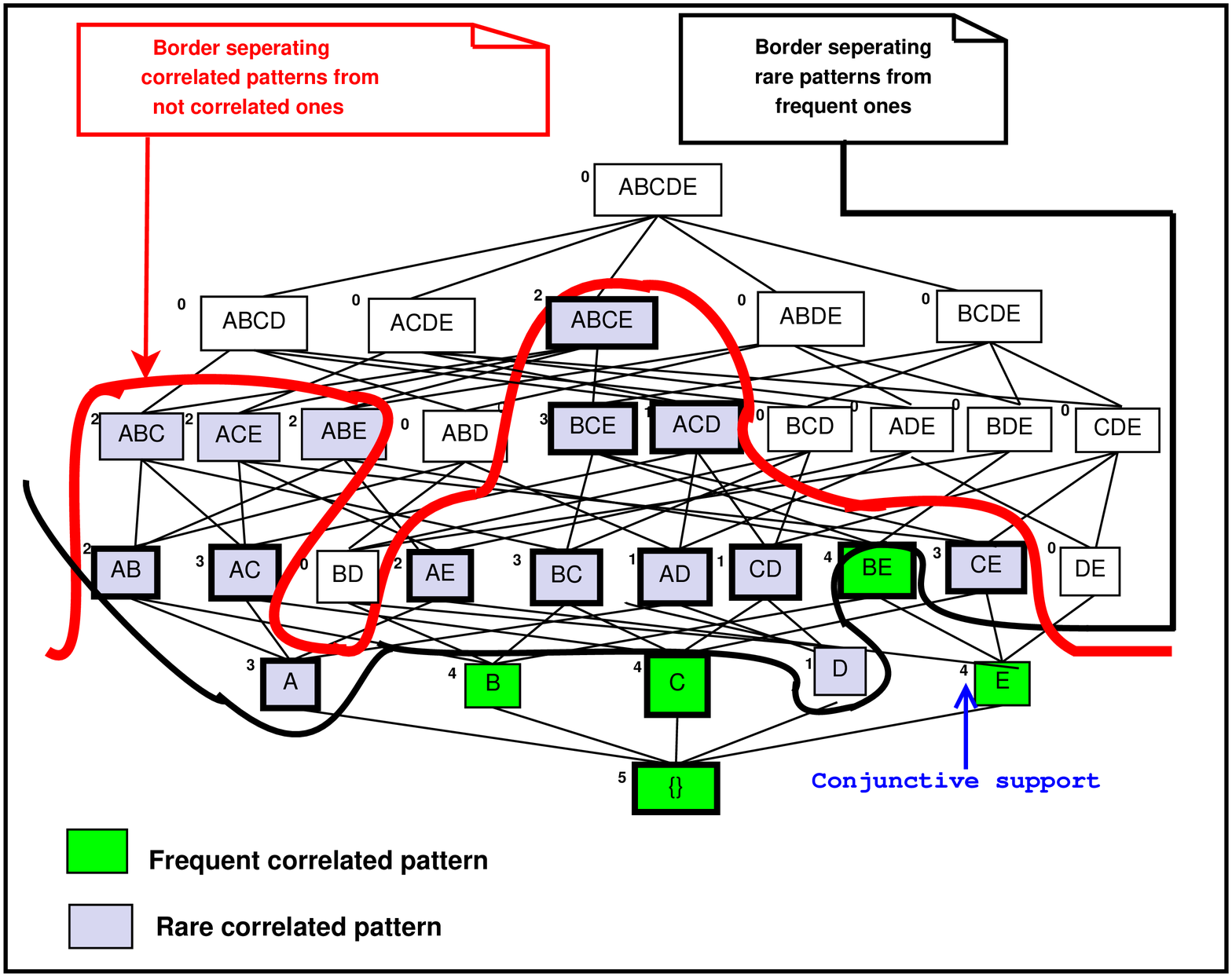}
		\caption{Localization of the correlated patterns for \textit{minsupp} = 4 and \textit{minbond} = 0.2 according to the extraction context given in Table \ref{Base_transactions}.}
		\label{figure_CRP}
	\end{center}
\end{figure}
Therefore, the localization of these elements is more difficult than the localization of theories corresponding to constraints of the same nature.
Indeed, the conjunction of anti-monotonic constraints \textsc{(}\textit{resp.} monotonic\textsc{)} is an anti-monotonic one \textsc{(}\textit{resp.} monotonic\textsc{)} \cite{luccheKIS05_MAJ_06}.
For example, the constraint "being a correlated frequent pattern" is anti-monotonic, since it results from the conjunction of two anti-monotonic constraints
namely,  "being a correlated pattern" and "being a frequent pattern". This constraint induces, then, an order ideal on the itemset lattice \cite{tarekds2010}.
However, the constraint ``being a rare and a not correlated pattern'' is monotonic, since it results from the conjunction of two monotonic constraints
namely,  "being a not correlated pattern'' and "being a rare pattern''. This constraint induces, then, an order filter on the itemset lattice.

In order to assess the size of the $\mathcal{RCP}$ set, and given the nature of the two constraints induced by the minimal thresholds of rarity and correlation respectively \textit{minsupp} and \textit{minbond},
the size of the $\mathcal{RCP}$ set of rare correlated patterns varies as shown in the following proposition.
\begin{proposition}
	\textbf{a\textsc{)}} Let \textit{minsupp}$_1$ and \textit{minsupp}$_2$ be two minimal thresholds of conjunctive support and $\mathcal{RCP}_{s1}$ and $\mathcal{RCP}_{s2}$ be the two sets of patterns associated to each threshold for the same value of \textit{minbond}.
	We have: if \textit{minsupp}$_1$ $\leq$ \textit{minsupp}$_2$, then $\mathcal{RCP}_{s1}$ $\subseteq$ $\mathcal{RCP}_{s2}$ and consequently $|\mathcal{RCP}_{s1}|$ $\leq$ $|\mathcal{RCP}_{s2}|$.\\
	\textbf{b\textsc{)}} Let \textit{minbond}$_1$ and \textit{minbond}$_2$
	be two minimal thresholds of \textit{bond} measure and let
	$\mathcal{RCP}_{b1}$ and $\mathcal{RCP}_{b2}$ be the two sets of patterns associated to each threshold for the same value of \textit{minsupp}.
	We have: if \textit{minbond}$_1$ $\leq$ \textit{minbond}$_2$, then $\mathcal{RCP}_{b2}$ $\subseteq$ $\mathcal{RCP}_{b1}$, consequently
	$|\mathcal{RCP}_{b2}|$ $\leq$ $|\mathcal{RCP}_{b1}|$.
\end{proposition}
\begin{proof}
	- The proof of \textbf{a\textsc{)}} derives from the fact that for $I$ $\subseteq$ $\mathcal{I}$, if \textit{Supp}\textsc{(}$\wedge I$\textsc{)} $<$ \textit{minsupp}$_1$, then \textit{Supp}\textsc{(}$\wedge I$\textsc{)} $<$ \textit{minsupp}$_2$. Therefore, $\forall$ $I$ $\in$ $\mathcal{RCP}_{s1}$, $I$ $\in$ $\mathcal{RCP}_{s2}$. As a result,  $\mathcal{RCP}_{s1}$ $\subseteq$ $\mathcal{RCP}_{s2}$.\\
	- The proof of \textbf{b\textsc{)}} derives from the fact that for $I$ $\subseteq$ $\mathcal{I}$, if \textit{bond}\textsc{(}$I$\textsc{)} $\geq$ \textit{minbond}$_2$, then \textit{bond}\textsc{(}$I$\textsc{)} $\geq$ \textit{minbond}$_1$. Therefore, $\forall$ $I$ $\in$ $\mathcal{RCP}_{b2}$, $I$ $\in$ $\mathcal{RCP}_{b1}$.
	As a result, $\mathcal{RCP}_{b2}$ $\subseteq$ $\mathcal{RCP}_{b1}$.\\\mbox{}
\end{proof}
We can then deduce that the size of the set  $\mathcal{RCP}$ is proportional to \textit{minsupp} and inversely proportional to \textit{minbond}. However, in the general case, we cannot decide about the size of the set  $\mathcal{RCP}$ when both thresholds vary simultaneously.
The next section is dedicated to the presentation of the closure operator associated to the \textit{bond} measure. This operator characterizes the correlated patterns through the induced equivalence classes.

\section{The $f_{bond}$ closure operator}\label{sec_FBond}
The $f_{bond}$ closure operator associated to the \textit{bond} measure is defined as follows \cite{tarekds2010}:
\begin{definition} \label{closure_fbond1} \textbf{The operator $f_{bond}$}\\
	{\setlength\arraycolsep{5pt}
		\begin{eqnarray*}
			\large f_{bond}:
			\mathcal{P}\textsc{(}\mathcal{I}\textsc{)}&
			\rightarrow &
			\mathcal{P}\textsc{(}\mathcal{I}\textsc{)}\\
			I & \mapsto &
			f_{bond}\textsc{(}I\textsc{)} =
			I\cup\{i\in\mathcal{I}\setminus I  | \  \textit{bond}\textsc{(}I\textsc{)} = \textit{bond}\textsc{(} I\cup\{i\}\textsc{)}\}
	\end{eqnarray*}}
\end{definition}
The operator $f_{bond}$ has been shown to be a closure operator \cite{tarekds2010}. Indeed, it fulfills the extensitivity,
the isotony and the idempotency properties \cite{ganter99}.
The closure of a pattern $I$ by $f_{bond}$, \textit{i.e.} $f_{bond}\textsc{(}I$\textsc{)}, corresponds to the maximal set of items containing $I$ and sharing the same
\textit{bond} value with $I$.
\begin{example}
	Consider the extraction context sketched by Table \ref{Base_transactions} \textsc{(Page \pageref{Base_transactions})}. For
	\textit{minbond} = 0.2, we have \textit{bond}\textsc{(}$AB$\textsc{)} = $\displaystyle\frac{2}{5}$, \textit{bond}\textsc{(}$ABC$\textsc{)} = $\displaystyle\frac{2}{5}$ and \textit{bond}\textsc{(}$ABE$\textsc{)} = $\displaystyle\frac{2}{5}$. Thus, $C$ $\in$ $f_{bond}$\textsc{(}$AB$\textsc{)}, and $E$ $\in$ $f_{bond}$\textsc{(}$AB$\textsc{)}. Contrariwise, \textit{bond}\textsc{(}$ABD$\textsc{)} = $\displaystyle\frac{0}{5}$ = 0. Thus, $D$ $\notin$ $f_{bond}$\textsc{(}$AB$\textsc{)}. Consequently, we have $f_{bond}$\textsc{(}$AB$\textsc{)} = $ABCE$.
	
	Let us illustrate the different properties of the $f_{bond}$ closure operator:
	\begin{enumerate}
		\item For the Extensitivity property:  we have, for example, $f_{bond}$\textsc{(}\texttt{CD}\textsc{)} = \texttt{ACD},
		\texttt{CD} $\subseteq$ $f_{bond}$\textsc{(}\texttt{CD}\textsc{)}.
		\item For the Isotony property: we have, for example, \texttt{AB} $\supset$ \texttt{B}, $f_{bond}$\textsc{(}\texttt{AB}\textsc{)} = \texttt{ABCE} and $f_{bond}$\textsc{(}\texttt{B}\textsc{)} = \texttt{BE}.
		\item For the Idempotency property: we have, the example of the closed itemset \texttt{ABCE},
		$f_{bond}$\textsc{(}$f_{bond}$\textsc{(}\texttt{ABCE}\textsc{)}\textsc{)} = \texttt{ABCE}.
	\end{enumerate}
	
\end{example}
The closure operator $f_{bond}$ induces an equivalence relation on
the power-set of the set of items $\mathcal{I}$, splitting it into disjoint
\textit{$f_{bond}$} equivalence classes which are formally defined as follows:
\begin{definition}\label{EquivClsbond} \textbf{Equivalence class associated to the closure operator $f_{bond}$}\\
	An equivalence class associated to the $f_{bond}$ closure operator is composed by all the patterns having the same closure by the operator
	$f_{bond}$.
\end{definition}
In each class, all the elements have the same $f_{bond}$ closure and the same value of \textit{bond}. The minimal patterns of a \textit{bond} equivalence class are the smallest incomparable members, \textit{w.r.t.} set inclusion, while the $f_{bond}$ closed pattern
is the largest one. These sets are formally defined in the following:
\begin{definition}\label{ClosedCorr} \textbf{Closed correlated patterns}\\
	The set $\mathcal{CCP}$ of closed correlated patterns by $f_{bond}$ is equal to: $\mathcal{CCP}$ = $\{$$I$ $\in$ $\mathcal{CP}$$\mid$ $\nexists$ $I_{1}$ $\supset$ $I$: \textit{bond}\textsc{(}$I$\textsc{)} = \textit{bond}\textsc{(}$I_{1}\textsc{)}\}$, or equivalently: $\mathcal{CCP}$ = $\{$$I$ $\in$ $\mathcal{CP}$$\mid$ $\nexists$ $I_{1}$ $\supset$ $I$: $f_{bond}\textsc{(}I$\textsc{)} = $f_{bond}\textsc{(}I_{1}\textsc{)}\}$.
\end{definition}
\begin{definition} \label{MinCorr} \textbf{Minimal correlated patterns}\\
	The set $\mathcal{MCP}$ of minimal correlated patterns is equal to: $\mathcal{MCP}$ = $\{$$I$ $\in$ $\mathcal{CP}$$\mid$ $\nexists$ $I_{1}$ $\subset$ $I$: \textit{bond}\textsc{(}$I$\textsc{)} = \textit{bond}\textsc{(}$I_{1}\textsc{)}\}$, or equivalently: $\mathcal{MCP}$ = $\{$$I$ $\in$ $\mathcal{CP}$$\mid$ $\nexists$ $I_{1}$ $\subset$ $I$: $f_{bond}\textsc{(}I$\textsc{)} = $f_{bond}\textsc{(}I_{1}\textsc{)}\}$.
\end{definition}
The set $\mathcal{MCP}$ of minimal correlated patterns forms an order ideal.
In fact, this set is composed by the patterns which fulfill the anti-monotonic constraint ``Being minimal in the equivalence class and being correlated''.
Indeed, this constraint corresponds to the conjunction between the two following anti-monotonic constraints,  ``being minimal'' and ``being correlated''.

The following proposal presents the common properties of patterns belonging to the same $f_{bond}$ equivalence class.
\begin{proposition}\label{proprietes_CEq_f_bond}
	Let $\mathcal{C}$ be an equivalence class associated to the closure operator $f_{bond}$ and $I$, $I_1$ $\in$ $\mathcal{C}$. We have: \textbf{a\textsc{)}} $f_{bond}\textsc{(}I$\textsc{)} = $f_{bond}\textsc{(}I_1$\textsc{)}, \textbf{b\textsc{)}} \textit{bond}\textsc{(}$I$\textsc{)} = \textit{bond}\textsc{(}$I_1$\textsc{)}, \textbf{c\textsc{)}} \textit{Supp}\textsc{(}$\wedge I$\textsc{)} = \textit{Supp}\textsc{(}$\wedge I_1$\textsc{)}, \textbf{d\textsc{)}} \textit{Supp}\textsc{(}$\vee I$\textsc{)} = \textit{Supp}\textsc{(}$\vee I_1$\textsc{)}, and, \textbf{e\textsc{)}} \textit{Supp}\textsc{(}$\neg I$\textsc{)} = \textit{Supp}\textsc{(}$\neg I_1$\textsc{)}.
\end{proposition}
\begin{proof}
	\begin{description}
		\item[a\textsc{)}] Thanks to Definition \ref{EquivClsbond}, $I$ and $I_1$ share the same closure by $f_{bond}$. Let $F$ be this closure.
		\item[b\textsc{)}] Since the closure operator preserves the value of the \textit{bond} measure \textsc{(}\textit{cf.} Definition \ref{closure_fbond1}\textsc{)}, and since $I$ and $I_1$ have the same closure $F$, we have so  \textit{bond}\textsc{(}$I$\textsc{)} = \textit{bond}\textsc{(}$F$\textsc{)}, and \textit{bond}\textsc{(}$I_1$\textsc{)} = \textit{bond}\textsc{(}$F$\textsc{)}. Therefore, \textit{bond}\textsc{(}$I$\textsc{)} = \textit{bond}\textsc{(}$I_1$\textsc{)}.
		\item[c\textsc{)}, d\textsc{)}, and e\textsc{)}] As $I$ $\subseteq$ $F$ and \textit{bond}\textsc{(}$I$\textsc{)} = \textit{bond}\textsc{(}$F$\textsc{)}, according to Proposition \ref{PropBond}, both of $I$ and $F$ share the same conjunctive, disjunctive and negative supports. It is the same case for $I_1$ and $F$. Therefore, both $I$ and $I_1$ have the same conjunctive, disjunctive and negative supports.
\end{description}\end{proof}
Therefore, all the patterns belonging to the same equivalence class induced by $f_{bond}$, appear exactly in the same transactions
\textsc{(}thanks to the equality of the conjunctive support\textsc{)}.
Besides, the items associated to the patterns of the same class characterize the same transactions. In fact, each class necessarily contains
a non empty subset of every pattern of the class \textsc{(}thanks to the equality of the disjunctive support\textsc{)}.
This closure operator links the conjunctive
search space to the disjunctive one.
In this respect, we begin the next section by the study of the characteristics of the \textit{rare correlated equivalence classes}, induced by the $f_{bond}$ closure operator.
\section{Condensed representations of rare correlated patterns} \label{section_RC}
Before introducing our condensed representations, we highlight that the condensed representations prove their high utility in various fields such as: bioinformatics \cite{pasquier2009} and data grids \cite{tarekJSS2015}. 
\subsection{Characterization of the rare correlated equivalence classes} \label{sub_sec_EquivCls}
The  equivalence  classes permit to retain only the non-redundant patterns.
Indeed, among all the patterns of a given equivalence class, only the patterns which are necessary for the regeneration of the whole set of rare correlated patterns, are maintained. Doing so, it considerably reduces the redundancy among the extracted knowledge.
The notion of equivalence classes also facilitates the exploration of the search space.
Indeed, the application of the $f_{bond}$ closure operator allows switching from the minimal elements of a class to its maximal element without having to sweep through the intermediate levels.

Each equivalence class, induced by the $f_{bond}$ closure operator, contains the patterns sharing the same $f_{bond}$ closure,
and thus they are characterized by the same conjunctive, disjunctive supports as well as the same \textit{bond} value \textsc{(}\textit{cf.} Proposition \ref{proprietes_CEq_f_bond}\textsc{)}.
Therefore, the elements of the same equivalence class have the same behavior towards the correlation and the rarity constraints.
In fact, for a correlated equivalence class, \textit{i.e.} a class which contains correlated patterns, all of them are rare or frequent.
It is also the same for a non-correlated equivalence class, \textit{i.e.} which contains non-correlated patterns.
Therefore we can deduce that, for an equivalence class induced by $f_{bond}$,
it is sufficient to evaluate the correlation and the rarity constraints for just one pattern in order to get information about all the other elements of this class.
In this respect, we distinguish four different types of equivalence classes: \textsc{(}\textbf{i}\textsc{)} correlated frequent classes; \textsc{(}\textbf{ii}\textsc{)} non-correlated frequent classes; \textsc{(}\textbf{iii}\textsc{)} rare correlated classes; and \textsc{(}\textbf{iv}\textsc{)} rare non-correlated classes.
The main characteristic of equivalence classes induced by the $f_{bond}$ operator is very interesting. Indeed, this is not the case for all the closure operators.
For example, the application of the conjunctive closure operator associated to the conjunctive support induces equivalence classes where the behavior of a given pattern
towards the correlation constraint is not representative of the behavior of all the patterns of this class.
For each class, each pattern must be independently tested  from the other patterns in the same class to check whether it fulfills the correlation constraint or not.
It results from the above that, the application of the $f_{bond}$ provides a more selective process to extract rare correlated patterns.
\begin{example}
	Consider the extraction context sketched by Table \ref{Base_transactions} \textsc{(}Page \pageref{Base_transactions}\textsc{)}. For \textit{minsupp} = 4 and \textit{minbond} = 0.2, Figure \ref{ExpEquivCls} shows
	the obtained rare correlated equivalence classes. We enumerate for example the class which contains the patterns $AB$, $AE$, $ABC$, $ABE$, $ACE$, and $ABCE$.
	Their respective conjunctive supports are equal to 2 and their \textit{bond} value is equal to $\displaystyle\frac{2}{5}$.
	The pattern $ABCE$ is the closed correlated one of this class.
\end{example}

The $\mathcal{RCP}$ set of rare correlated patterns is then split into disjoint equivalence classes, the rare correlated equivalence classes. In each class, the closed pattern is the largest one with respect to the inclusion set relation. On the other hand, the smallest incomparable patterns are the minimal rare correlated patterns
\textit{w.r.t.} the inclusion set relation. The set of minimal and set of closed patterns are formally defined as follows:
\begin{definition}\label{CRCP} \textbf{Closed rare correlated patterns}\\
	The $\mathcal{CRCP}$
	$^{\textsc{(}}$\footnote{$\mathcal{CRCP}$ stands for \textbf{C}losed \textbf{R}are \textbf{C}orrelated \textbf{P}atterns.}$^{\textsc{)}}$
	set of closed rare correlated patterns is equal to:
	$\mathcal{CRCP}$ = $\{$$I$ $\in$ $\mathcal{RCP}$$|$ $\forall$ $I_{1}$ $\supset$ $I$: \textit{bond}\textsc{(}$I$\textsc{)} $>$ \textit{bond}\textsc{(}$I_{1}\textsc{)}\}$.
\end{definition}
The $\mathcal{CRCP}$ set corresponds to the intersection between the rare correlated patterns set and the set of closed correlated patterns.
We have so, $\mathcal{CRCP}$  = $\mathcal{RCP}$ $\cap$ $\mathcal{CCP}$.
\begin{definition}\label{MRCP} \textbf{Minimal rare correlated patterns}\\
	The $\mathcal{MRCP}$
	$^{\textsc{(}}$\footnote{$\mathcal{MRCP}$ stands for \textbf{M}inimal  \textbf{R}are \textbf{C}orrelated \textbf{P}atterns.}$^{\textsc{)}}$
	set of minimal rare correlated patterns is equal to:
	$\mathcal{MRCP}$ = $\{$$I$ $\in$ $\mathcal{RCP}$$|$ $\forall$ $I_{1}$ $\subset$ $I$: \textit{bond}\textsc{(}$I$\textsc{)} $<$ \textit{bond}\textsc{(}$I_{1}\textsc{)}\}$.
\end{definition}
The $\mathcal{MRCP}$ set corresponds to the intersection between the set of rare correlated patterns and the set of minimal correlated patterns.
Thus, we have, $\mathcal{MRCP}$ = $\mathcal{RCP}$ $\cap$ $\mathcal{MCP}$.

\begin{figure}[htbp]
	\begin{center}
		\includegraphics[scale = 0.4]{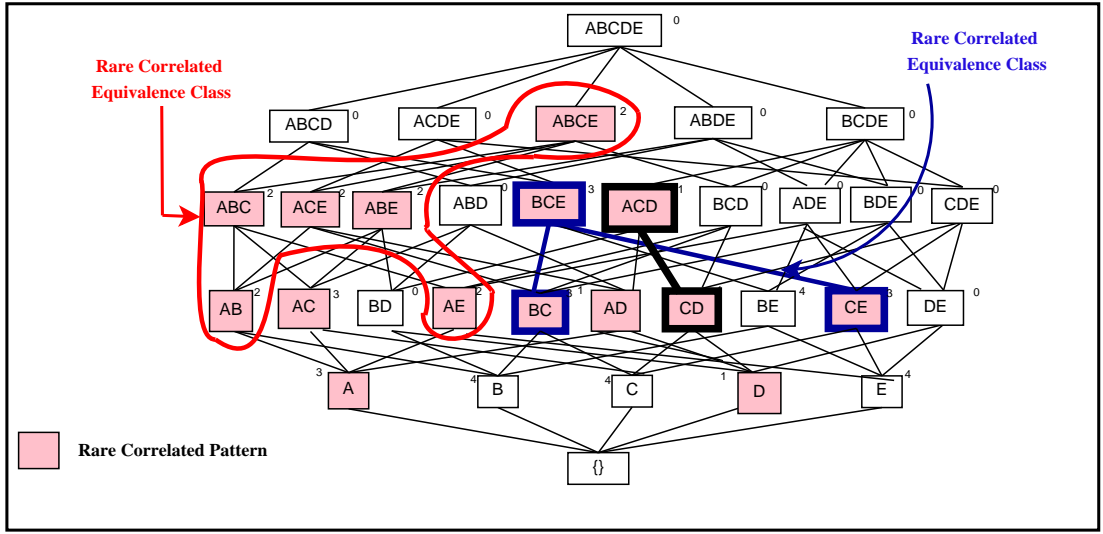}
		\caption{An example of rare correlated equivalence classes for \textit{minsupp} = 4 and \textit{minbond} = 0.2.}
		\label{ExpEquivCls}
	\end{center}
\end{figure}

\begin{example}\label{example_MCRP_and_CCRP}
	Consider the extraction context sketched by Table \ref{Base_transactions} \textsc{(Page \pageref{Base_transactions})}.
	For \textit{minsupp} = 4 and \textit{minbond} = 0.2,
	we have $\mathcal{CRCP}$ = $\{$$A$, $D$, $AC$, $AD$, $ACD$, $BCE$, $ABCE$$\}$
	and $\mathcal{MRCP}$ = $\{$$A$, $D$, $AB$, $AC$, $AD$, $AE$, $BC$, $CD$, $CE$$\}$.
\end{example}

An accurate representation of rare correlated patterns should determine, for an arbitrary pattern, whether it is rare correlated or not.
If it is a rare correlated one, then this representation must allow drifting without  information loss the values of
its support and its \textit{bond} measure. In this respect, the proposed representations in this work will be later shown to be perfect: their respective sizes never exceed that of the whole set of rare correlated patterns.
In addition, since they are information lossless, these representations allow, whenever of need, the derivation of the whole set of rare correlated patterns efficiently.

To define our concise exact representations of rare correlated patterns, we are based on the notion of equivalence classes.

The first intuitive idea when defining a concise exact representation of the rare correlated patterns
is to study whether the minimal elements or maximal ones of the equivalence classes would constitute an exact concise representation of the $\mathcal{RCP}$ set.

In this respect, it is important to remind that the $\mathcal{RCP}$ set results from the intersection of the
order ideal of correlated patterns and the order filter of rare patterns. Thus, the $\mathcal{RCP}$ set does not induce neither an order ideal nor an order filter.
In this situation, we take independently each set, to check whether the $\mathcal{CRCP}$ set or the $\mathcal{MRCP}$ set can provide a concise exact representation of the $\mathcal{RCP}$ set.

Let us analyze, in the following, each of the two sets separately:

- Let us begin with the $\mathcal{MRCP}$ set composed by the minimal elements of the rare correlated equivalence classes: In fact, due to the nature of its elements - minimal of their equivalence classes - this set allows for a given pattern  $I$ to evaluate it towards the constraint of rarity. Indeed, it is enough to find an element $J$ $\in$ $\mathcal{MRCP}$  \textit{s.t.} $J$ $\subseteq$ $I$  to decide whether $I$ is a rare pattern or not. If it is not the case, then $I$ is not a rare pattern. However, the set $\mathcal{MRCP}$ cannot determine, in the general case, whether $I$ is correlated or not
\textsc{(}this is possible only if $I$ $\in$ $\mathcal{MRCP}$\textsc{)}.
Even if it exists $J$ $\in$ $\mathcal{MRCP}$  \textit{s.t.} $J$ $\subset$ $I$, and even knowing that $J$ is correlated,
we cannot confirm the correlation of nature  of $I$ since this constraint is an anti-monotonic one \textsc{(}the fact that $J$ is correlated does not imply that $I$ is also correlated\textsc{)}.
Thus, the $\mathcal{MRCP}$ set, taken alone, cannot be an exact representation of the $\mathcal{RCP}$ set.

- Let us now treat the case of the $\mathcal{CRCP}$ of maximal elements of the rare correlated equivalence classes: Dually to the previous analysis of $\mathcal{MRCP}$, the $\mathcal{CRCP}$ set allows determining the nature of correlation for a given pattern $I$.
If it is included in just one  pattern $J$ $\in$ $\mathcal{CRCP}$, then $I$ is correlated. Otherwise, it is not a correlated one. However, due to their nature, the patterns composing the $\mathcal{CRCP}$ cannot in the general case derive the information about the status of rarity of a given pattern $I$ \textsc{(}this is possible only if $I$ $\in$ $\mathcal{CRCP}$\textsc{)}. Even if it exists $J$ $\in$ $\mathcal{CRCP}$,  \textit{s.t.} $I$ $\subset$ $J$ and even if we already know that $J$ is rare, we cannot decide whether $I$ is rare or not since the constraint of rarity is monotone \textsc{(}the fact that $J$ is rare does not imply that $I$ is also rare\textsc{)}. Thus, the $\mathcal{CRCP}$ set, taken alone, cannot be an exact representation of the $\mathcal{RCP}$ set.
Nevertheless, it is proved from the previous analysis that the union of the $\mathcal{MRCP}$ set and the $\mathcal{CRCP}$ set would constitute an accurate concise representation of the $\mathcal{RCP}$ set of rare correlated patterns.
The first alternative will be studied in the next sub-section, and will be then followed by two other optimizations in order to retain only the key elements for the lossless regeneration of the $\mathcal{RCP}$ set.

Based on this study, we introduce in the following subsections our new concise exact and approximate representations.
\subsection{The \textbf{$\mathcal{RCPR}$} concise exact representation}
The first representation, that we introduce, is defined as follows:
\begin{definition}\label{rmcr} \textbf{The $\mathcal{RCPR}$ representation}\\
	Let $\mathcal{RCPR}$ be the concise exact representation of the $\mathcal{RCP}$ set
	based on the $\mathcal{CRCP}$ set and on the $\mathcal{MRCP}$ set of the minimal rare correlated patterns.
	The $\mathcal{RCPR}$ representation is equal to: $\mathcal{RCPR}$ = $\mathcal{CRCP}$ $\cup$ $\mathcal{MRCP}$.
	The support, \textit{Supp}\textsc{(}$\wedge I$\textsc{)},
	and the \textit{bond} value, \textit{bond}\textsc{(}$I$\textsc{)} of each pattern $I$ of $\mathcal{RCPR}$ are exactly determined.
\end{definition}
\begin{example} \label{exprep1}
	Consider the extraction context sketched by Table \ref{Base_transactions} \textsc{(}Page \pageref{Base_transactions}\textsc{)}.
	For \textit{minsupp} = 4 and \textit{minbond} = 0.2,
	while considering the $\mathcal{CRCP}$ set and $\mathcal{MRCP}$ set
	\textsc{(}\textit{cf.} Example \ref{example_MCRP_and_CCRP}\textsc{)}, the $\mathcal{RCPR}$ set is equal to:
	$\{$\textsc{(}$A$, $3$, $\displaystyle\frac{3}{3}$\textsc{)},
	\textsc{(}$D$, $1$, $\displaystyle\frac{1}{1}$\textsc{)},
	\textsc{(}$AB$, $2$, $\displaystyle\frac{2}{5}$\textsc{)},
	\textsc{(}$AC$, $3$, $\displaystyle\frac{3}{4}$\textsc{)},
	\textsc{(}$AD$, $1$, $\displaystyle\frac{1}{3}$\textsc{)},
	\textsc{(}$AE$, $2$, $\displaystyle\frac{2}{5}$\textsc{)},
	\textsc{(}$BC$, $3$, $\displaystyle\frac{3}{5}$\textsc{)},
	\textsc{(}$CD$, $1$, $\displaystyle\frac{1}{4}$\textsc{)},
	\textsc{(}$CE$, $3$, $\displaystyle\frac{3}{5}$\textsc{)},
	\textsc{(}$ACD$, $1$, $\displaystyle\frac{1}{4}$\textsc{)},
	\textsc{(}$BCE$, $3$, $\displaystyle\frac{3}{5}$\textsc{)}
	and \textsc{(}$ABCE$, $2$, $\displaystyle\frac{2}{5}$\textsc{)}$\}$.
\end{example}
The following theorem proves that the $\mathcal{RCPR}$
representation is a lossless concise representation of the $\mathcal{RCP}$ set.
\begin{theorem}\label{representation1_exacte}
	The $\mathcal{RCPR}$ representation is a concise exact representation of the $\mathcal{RCP}$ set of rare correlated patterns.
\end{theorem}
\begin{proof}
	Let $I$ $\subseteq$ $\mathcal{I}$. We distinct between three different cases:
	
	\textbf{a\textsc{)}} If $I$ $\in$ $\mathcal{RCPR}$, then $I$ is a rare correlated pattern and we have
	its support and its \textit{bond} values.
	
	\textbf{b\textsc{)}} If $\nexists$ $J$ $\in$ $\mathcal{RCPR}$ as $J$ $\subseteq$ $I$ and $\nexists$ $Z$ $\in$ $\mathcal{RCPR}$ as $I$ $\subseteq$ $Z$, then $I$ $\notin$ $\mathcal{RCP}$
	since $I$ does not belong to any rare correlated equivalence class.
	
	\textbf{c\textsc{)}} If $I$ $\in$ $\mathcal{RCP}$.
	In fact, according to Proposition \ref{prop_set_CRP}, $I$
	is correlated since it is included in a correlated pattern, namely $Z$.
	It is also rare, since it contains a rare pattern, namely $J$. In this case, it is sufficient to localize
	the $f_{bond}$ closure of $I$ namely $F$. Then, the closed pattern $F$ belongs then to
	$\mathcal{RCPR}$ since $I$ is rare correlated and
	$\mathcal{RCPR}$ includes the $\mathcal{CRCP}$ set of closed rare correlated patterns.
	Therefore, $F$ = $min_{\subseteq}$\{$I_1$ $\in$ $\mathcal{RCPR} |$ $I$ $\subseteq$ $I_1$\}.
	Since the $f_{bond}$ closure operator
	conserves the \textit{bond} value and thus the conjunctive support \textsc{(}\textit{cf.} Proposition \ref{proprietes_CEq_f_bond}\textsc{)},
	we have: \textit{bond}\textsc{(}$I$\textsc{)} = \textit{bond}\textsc{(}$F$\textsc{)} and
	\textit{Supp}\textsc{(}$\wedge I$\textsc{)} = \textit{Supp}\textsc{(}$\wedge F$\textsc{)}.
\end{proof}
\begin{example}\label{example_3_cas_RMCR}
	Consider the $\mathcal{RCPR}$ representation illustrated by the previous example. Let us consider each case separately. The pattern $AD$ $\in$ $\mathcal{RCPR}$.
	Thus, we have its support equal to 1 and its \textit{bond} value equal to $\displaystyle\frac{1}{3}$.
	Even though, the pattern $BE$ is included in two patterns from the
	$\mathcal{RCPR}$ representation, namely $BCE$ and $ABCE$, $BE$ $\notin$ $\mathcal{RCP}$ since no element of $\mathcal{RCPR}$ is included in $BE$.
	Consider now the pattern $ABC$.
	It exists two patterns of $\mathcal{RCPR}$ proving that the pattern
	$ABC$ is a rare correlated one, namely $AB$ and $ABCE$, since $AB$ $\subseteq$ $ABC$ $\subseteq$ $ABCE$.
	The smallest pattern in $\mathcal{RCPR}$ covering $ABC$, \textit{i.e.} its closure, is
	$ABCE$. Then, \textit{bond}\textsc{(}$ABC$\textsc{)} = \textit{bond}\textsc{(}$ABCE$\textsc{)} = $\displaystyle\frac{2}{5}$, and \textit{Supp}\textsc{(}$\wedge ABC$\textsc{)} = \textit{Supp}\textsc{(}$
	\wedge ABCE $\textsc{)} = 2.
\end{example}
We show through the following proposition that the $\mathcal{RCPR}$ representation is a \textit{perfect cover} of the
$\mathcal{RCP}$ set of rare correlated patterns.
\begin{proposition}
	The $\mathcal{RCPR}$ representation is a \textit{perfect cover} of the $\mathcal{RCP}$ set.
\end{proposition}
\begin{proof}
	In fact, the size of the $\mathcal{RCPR}$ representation \textit{does never exceed} that of the $\mathcal{RCP}$ set whatever the extraction context, the \textit{minsupp} and the \textit{minbond} values.
	Indeed, it is always true that \textsc{(}$\mathcal{CRCP}$ $\cup$ $\mathcal{MRCP}$\textsc{)} $\subseteq$ $\mathcal{RCP}$.
	Furthermore, knowing the conjunctive support of a given pattern and its \textit{bond} value, we can compute the disjunctive support and thus the negative support.
	The interrogation of the representation $\mathcal{RCPR}$ can be based on the proof of Theorem \ref{representation1_exacte}.
	Thus, for a given pattern, thanks to the $\mathcal{RCPR}$ representation, we can determine whether it is rare correlated or not.
	If it is correlated rare, then its support as well as its \textit{bond} value will be derived using the mechanism described by the previous theorem.
	The regeneration process of the whole set of rare correlated patterns can also be based on this theorem.
	This process starts by the smallest rare correlated patterns namely the
	minimal rare correlated patterns
	\textsc{(}constituting the $\mathcal{MRCP}$ set\textsc{)}.
	These patterns  belong to  $\mathcal{RCPR}$ and we have therefore all their required information.
	It is then sufficient to localize for each minimal $M$ its closure $F$ which belongs to $\mathcal{RCPR}$
	\textsc{(}$F$ $\in$ $\mathcal{CRCP}$ and this set is included in $\mathcal{RCPR}$\textsc{)}.
	All the patterns which are included between $M$ and $F$ share the same support and \textit{bond} value as $M$ and $F$ since they belong to the same rare correlated equivalence class.
\end{proof}
\begin{remark}
	It is important to mention that it is necessary to maintain for each pattern $I$ of the representation, at the same time
	\textit{Supp}\textsc{(}$\wedge I$\textsc{)} as well as \textit{bond}\textsc{(}$I$\textsc{)}. On the one hand, \textit{bond}\textsc{(}$I$\textsc{)} is equal to the ratio between the conjunctive and the disjunctive support of $I$ and cannot determine the conjunctive support of $I$. On the other hand, knowing only the conjunctive support of $I$ is not sufficient to compute the \textit{bond} value. The disjunctive support can be derived by the inclusion-exclusion identities only if we know all the conjunctive supports values of all the subsets of $I$ \cite{galambos}. Nevertheless, if $I$ is a rare correlated pattern, then its subsets are not necessarily rare correlated. Therefore, we don't have the values of their conjunctive supports. Thus, we must keep track of the conjunctive support and the \textit{bond} value for each element of the $\mathcal{RCPR}$ representation.
	
	In this case, the closed and minimal patterns of the equivalence classes constitute, as shown previously, an interesting solution in order to represent with a concise and exact manner the $\mathcal{RCP}$ set. In fact, the localization of these patterns requires a limited neighborhood, \textit{i.e.}, just the strict supersets and subsets, and
	not all their subsets.
	In addition to this, the derivation of the support of the whole set of patterns from the closed and minimal non derivable can be done directly.
\end{remark}
\begin{remark}\label{remak_manage_MCRP_vs_CCRP}
	It is also interesting to mention that the fact to consider in  the $\mathcal{RCPR}$ set, the union of the two sets $\mathcal{MRCP}$ and
	$\mathcal{CRCP}$ allows avoiding redundancy, because of the duplication of some patterns,
	which may appear in the representation when we consider the $\mathcal{MRCP}$ set and the $\mathcal{CRCP}$ set separately.
	For example, if we consider the example \ref{exprep1}, we note that the elements \textsc{(}$A$, 3, $\displaystyle\frac{3}{3}$\textsc{)} ,
	\textsc{(}$D$, 1, $\displaystyle\frac{1}{1}$\textsc{)}, \textsc{(}$AC$, 3, $\displaystyle\frac{3}{4}$\textsc{)} and
	\textsc{(}$AD$, 1, $\displaystyle\frac{1}{3}$\textsc{)} belong to both sets $\mathcal{MRCP}$ and $\mathcal{CRCP}$.
	However, one advantage of saving each set separately allows the reduction of some tests of inclusion
	in the extraction of the representation.
	In fact, to make the choice between tolerating some duplication or benefiting from a potential reduction of the regeneration cost
	depends on the nature of the application where we can privilege either the optimization of the memory space or the derivation cost.
\end{remark}
We propose in the following two refined versions of the $\mathcal{RCPR}$ representation, in order to further reduce
the size of this representation.

\subsection{The $\mathcal{MM}$$ax$$\mathcal{CR}$ concise exact representation}

The first refinement is based on the fact that the $\mathcal{MRCP}$ set of minimal rare correlated patterns increased only by the maximal patterns according to the inclusion set, among the
$\mathcal{CRCP}$ set of closed rare correlated patterns is sufficient to faithfully represent the $\mathcal{RCP}$ set.
In this respect, we define the $\mathcal{M}$$ax$$\mathcal{CRCP}$ set of maximal closed rare correlated patterns as follows:
\begin{definition}\label{defMF} \textbf{The $\mathcal{M}$$ax$$\mathcal{CRCP}$ set of maximal closed rare correlated patterns}\\
	The  $\mathcal{M}$$ax$$\mathcal{CRCP}$ set is composed by the patterns which are closed correlated rare ones
	\textsc{(}\textit{cf.} Definition \ref{CRCP}, page \pageref{CRCP}\textsc{)} and at the same time they are maximal correlated \textsc{(}\textit{cf.} Definition \ref{bdpos}, page \pageref{bdpos}\textsc{)}.
	Then, we have $\mathcal{M}$$ax$$\mathcal{CRCP}$ = $\mathcal{CRCP}$ $\cap$ $\mathcal{M}ax$$\mathcal{CP}$.
\end{definition}
The $\mathcal{M}$$ax$$\mathcal{CRCP}$ set is then limited to the elements of the $\mathcal{M}ax$$\mathcal{CP}$
set which are also rare, in addition of being the largest correlated patterns.
\begin{example}\label{example_MaxCCRP}
	Consider the extraction context sketched by Table \ref{Base_transactions}\textsc{(}Page \pageref{Base_transactions}\textsc{)}. For \textit{minsupp} = 4 and \textit{minbond} = 0.2, we have
	$\mathcal{M}$$ax$$\mathcal{CRCP}$ = $\{$$ACD$, $ABCE$$\}$.
\end{example}

We define in the following the representation based on the $\mathcal{M}$$ax$$\mathcal{CCRP}$ set.
\begin{definition} \label{rep2concise exacte def} \textbf{The $\mathcal{MM}$$ax$$\mathcal{CR}$ representation}\\
	Let $\mathcal{MM}$$ax$$\mathcal{CR}$ be the representation based on the $\mathcal{MRCP}$ set
	and the $\mathcal{M}$$ax$$\mathcal{CRCP}$.
	We have,  $\mathcal{MM}$$ax$$\mathcal{CR}$ = $\mathcal{MRCP}$ $\cup$ $\mathcal{M}$$ax$$\mathcal{CRCP}$.
	For each element $I$ of the $\mathcal{MM}ax \mathcal{CR}$, the support, \textit{Supp}\textsc{(}$\wedge I$\textsc{)}, and the \textit{bond} value,
	\textit{bond}\textsc{(}$I$\textsc{)} are computed.
\end{definition}
\begin{example} \label{exprep2}
	Consider the extraction context sketched by Table \ref{Base_transactions} \textsc{(}Page \pageref{Base_transactions}\textsc{)}.
	For \textit{minsupp} = 4 and
	\textit{minbond} = 0.2, the representation $\mathcal{MM}ax \mathcal{CR}$ is equal to: $\{$\textsc{(}$A$, $3$, $\displaystyle\frac{3}{3}$\textsc{)},
	\textsc{(}$D$, $1$, $\displaystyle\frac{1}{1}$\textsc{)},
	\textsc{(}$AB$, $2$, $\displaystyle\frac{2}{5}$\textsc{)},
	\textsc{(}$AC$, $3$, $\displaystyle\frac{3}{4}$\textsc{)},
	\textsc{(}$AD$, $1$, $\displaystyle\frac{1}{3}$\textsc{)},
	\textsc{(}$AE$, $2$, $\displaystyle\frac{2}{5}$\textsc{)},
	\textsc{(}$BC$, $3$, $\displaystyle\frac{3}{5}$\textsc{)},
	\textsc{(}$CD$, $1$, $\displaystyle\frac{1}{4}$\textsc{)},
	\textsc{(}$CE$, $3$, $\displaystyle\frac{3}{5}$\textsc{)},
	\textsc{(}$ACD$, $1$, $\displaystyle\frac{1}{4}$\textsc{)},
	and \textsc{(}$ABCE$, $2$, $\displaystyle\frac{2}{5}$\textsc{)}$\}$.
	We remark that, for this example, the unique element of the $\mathcal{RCPR}$ representation
	that does not belong to the $\mathcal{MM}$$ax$$\mathcal{CR}$ representation is, the pattern $BCE$.
	In fact, the closed patterns that were removed, \textit{i.e.}, $A$, $D$, $AC$ and $AD$, are also minimal. Indeed, the  $\mathcal{MM}$$ax$$\mathcal{CR}$ representation would be more optimized than do the $\mathcal{RCPR}$ representation, if
	the sets $\mathcal{MRCP}$ and $\mathcal{CRCP}$ were saved separately \textsc{(}\textit{cf.} Remark \ref{remak_manage_MCRP_vs_CCRP}\textsc{)}. In fact, the duplicate storage of the patterns $A$, $D$, $AC$ and $AD$ is avoided.
\end{example}
Theorem \ref{representation2_exacte} proves that the $\mathcal{MM}$ax$\mathcal{CR}$ set is a
lossless representation of the $\mathcal{RCP}$ set.
\begin{theorem}\label{representation2_exacte}
	The $\mathcal{MM}$ax$\mathcal{CR}$ set
	is an exact concise representation of the $\mathcal{RCP}$ set of rare correlated patterns.
\end{theorem}
\begin{proof}
	Let a pattern $I$ $\subseteq$ $\mathcal{I}$. We distinguish between three cases:
	
	\textbf{a\textsc{)}} If $I$ $\in$ $\mathcal{MM}$$ax$$\mathcal{CR}$, then $I$ is a rare correlated pattern and we have its support and \textit{bond} values.
	
	\textbf{b\textsc{)}} If $\nexists$ $J$ $\in$ $\mathcal{MM}$$ax$$\mathcal{CR}$ as $J$ $\subseteq$ $I$ or $\nexists$ $Z$ $\in$ $\mathcal{MM}$$ax$$\mathcal{CR}$ as $I$ $\subseteq$ $Z$, then $I$ $\notin$ $\mathcal{RCP}$ since $I$ do not belong to any rare correlated equivalence class.
	
	\textbf{c\textsc{)}} Else, $I$ $\in$ $\mathcal{RCP}$. In fact, according to Proposition \ref{prop_set_CRP}, $I$ is a correlated pattern since it is included in a correlated pattern, say $Z$.
	It is also rare since it contains a rare pattern, say $J$.
	Since $I$ is a rare correlated pattern and the representation
	$\mathcal{MM}$$ax$$\mathcal{CR}$ includes the $\mathcal{MRCP}$ set
	containing the minimal elements of the different rare correlated equivalence classes,
	this representation has at least one element of the equivalence class of $I$, particularly all the minimal patterns of its class.\\
	Since both the conjunctive support and the \textit{bond} measure decrease as far the size of the patterns is lowered, the support and bond values of $I$ are equal to the minimal values of the measures associated to its subsets belonging to the
	$\mathcal{MM}$$ax$$\mathcal{CR}$ representation.
	We deduce then that:\\
	$\bullet$ \textit{Supp}\textsc{(}$\wedge$$I$\textsc{)} = $min$$\{$\textit{Supp}\textsc{(}$\wedge$$I_1$\textsc{)}$|$ $I_1$ $\in$ $\mathcal{MM}$$ax$$\mathcal{CR}$ and $I_1$ $\subseteq$ $I$$\}$; and\\
	$\bullet$ \textit{bond}\textsc{(}$I$\textsc{)} = $min$$\{$\textit{bond}\textsc{(}$I_1$\textsc{)}$|$ $I_1$ $\in$ $\mathcal{MM}$$ax$$\mathcal{CR}$ and $I_1$ $\subseteq$ $I$$\}$.
\end{proof}
\begin{example}
	Consider the $\mathcal{MM}$$ax$$\mathcal{CR}$ presented in Example \ref{exprep2}.
	The treatment of the first and second cases is similar to the first two cases of the
	$\mathcal{RCPR}$ representation \textsc{(}\textit{cf.} Example \ref{example_3_cas_RMCR}\textsc{)}.
	The case of the pattern $ABE$ is illustrative of the third alternative.
	In fact, it exists two patterns, from the $\mathcal{MM}$$ax$$\mathcal{CR}$ representation, which makes $ABE$ a rare correlated pattern, namely $AB$ and $ABCE$
	\textsc{(}$AB$ $\subseteq$ $ABE$ $\subseteq$ $ABCE$\textsc{)}.
	The elements of the  $\mathcal{MM}$$ax$$\mathcal{CR}$ representation which are included in $ABE$ are $AB$ and $AE$.
	Consequently, \textit{Supp}\textsc{(}$\wedge$$ABE$\textsc{)} = $min$$\{$\textit{Supp}\textsc{(}$\wedge$$AB$\textsc{)},
	\textit{Supp}\textsc{(}$\wedge$$AE$\textsc{)}$\}$ = $min$$\{$2, 2$\}$ = 2, and \textit{bond}\textsc{(}$ABE$\textsc{)} =
	$min$$\{$\textit{bond}\textsc{(}$AB$\textsc{)}, \textit{bond}\textsc{(}$AE$\textsc{)}$\}$ = $min$$\{$$\displaystyle\frac{2}{5}$,
	$\displaystyle\frac{2}{5}$$\}$ = $\displaystyle\frac{2}{5}$.
\end{example}
Since the $\mathcal{MM}$$ax$$\mathcal{CR}$ representation is included in the $\mathcal{RCPR}$ representation, and the latter is a perfect cover of the
$\mathcal{RCP}$ set, then we deduce that the $\mathcal{MM}$$ax$$\mathcal{CR}$ representation is also
a perfect cover of the $\mathcal{RCP}$ set.
In the next sub-section, we present another refinement of the $\mathcal{RCPR}$ representation.

\subsection{The \textbf{$\mathcal{M}$$in$$\mathcal{MCR}$} concise exact representation}

Dually to the previous definition, it is sufficient to maintain in the $\mathcal{RCPR}$ representation, just the minimal elements,
according to the inclusion set, among the $\mathcal{MRCP}$ set.
The pruning of the other elements from the $\mathcal{MRCP}$ set will be shown to be
information lossless during the regeneration of the whole set of rare correlated patterns.
The $\mathcal{M}$$in$$\mathcal{MRCP}$ set of minimal elements among the
$\mathcal{MRCP}$, is thus defined as follows:
\begin{definition}\label{defMM} \textbf{The $\mathcal{M}$$in$$\mathcal{MRCP}$ set of the minimal elements of the $\mathcal{MRCP}$ set}\\
	The $\mathcal{M}$$in$$\mathcal{MRCP}$ set is composed by the patterns which are minimal rare correlated patterns
	\textsc{(}\textit{cf.} Definition \ref{MRCP}, page \pageref{MRCP}\textsc{)} and at the same time
	minimal rare patterns \textsc{(}\textit{cf.} Definition \ref{mrp}, page \pageref{mrp}\textsc{)}.
	Thus, $\mathcal{M}in \mathcal{MRCP}$  = $\mathcal{MRCP}$ $\cap$ $\mathcal{M}in \mathcal{RP}$.
\end{definition}
The $\mathcal{M}in \mathcal{MRCP}$ set is then limited by the
minimal rare correlated patterns which are also minimal rare
\textsc{(}In addition to being the smallest rare patterns\textsc{)}.
\begin{example} \label{expMM}
	Consider the extraction context sketched by Table \ref{Base_transactions} \textsc{(}Page \pageref{Base_transactions}\textsc{)}.
	For  \textit{minsupp} = 4 and \textit{minbond} = 0.2, we have $\mathcal{M}in$$\mathcal{MRCP}$ = $\{$$A$, $D$, $BC$, $CE\}$.
\end{example}

The following definition introduces the representation based on the $\mathcal{M}in\mathcal{MCRP}$ set.
\begin{definition}\label{rep3concise exacte def} \textbf{The $\mathcal{M}in \mathcal{MCR}$ representation}\\
	Let $\mathcal{M}in\mathcal{MCR}$ be the representation based on the
	$\mathcal{M}in \mathcal{MCRP}$ set and the $\mathcal{CRCP}$ set.
	We have  $\mathcal{M}in \mathcal{MCR}$ = $\mathcal{CRCP}$ $\cup$ $\mathcal{M}in\mathcal{MRCP}$.
	For each element $I$ of $\mathcal{M}in\mathcal{MCR}$,  its support \textit{Supp}\textsc{(}$\wedge I$\textsc{)} and its \textit{bond} value \textit{bond}\textsc{(}$I$\textsc{)} are computed.
\end{definition}
\begin{example} \label{exprep3}
	Consider the extraction context sketched by Table \ref{Base_transactions} \textsc{(}Page \pageref{Base_transactions}\textsc{)}.
	For \textit{minsupp} = 4 and \textit{minbond} = 0.2,
	we have the $\mathcal{M}$$in$$\mathcal{MCR}$ = $\{$\textsc{(}$A$, 3, $\displaystyle\frac{3}{3}$\textsc{)},
	\textsc{(}$D$, 1, $\displaystyle\frac{1}{1}$\textsc{)},
	\textsc{(}$AC$, 3, $\displaystyle\frac{3}{4}$\textsc{)},
	\textsc{(}$AD$, 1, $\displaystyle\frac{1}{3}$\textsc{)},
	\textsc{(}$BC$, 3, $\displaystyle\frac{3}{5}$\textsc{)},
	\textsc{(}$CE$, 3, $\displaystyle\frac{3}{5}$\textsc{)},
	\textsc{(}$ACD$, 1, $\displaystyle\frac{1}{4}$\textsc{)},
	\textsc{(}$BCE$, 3, $\displaystyle\frac{3}{5}$\textsc{)},
	and \textsc{(}$ABCE$, 2, $\displaystyle\frac{2}{5}$\textsc{)}$\}$.
	We remark that, this representation has three elements less than the $\mathcal{RCPR}$ representation,
	namely $AB$, $AE$ and $CD$.
\end{example}
The following theorem proves that this representation is also a lossless reduction of the $\mathcal{RCP}$ set.
\begin{theorem}\label{representation3_exacte}
	The  $\mathcal{M}$$in$$\mathcal{MCR}$ representation is a concise exact representation of the  $\mathcal{RCP}$ set of rare correlated patterns.
\end{theorem}
\begin{proof}
	Let $I$ $\subseteq$ $\mathcal{I}$. We distinguish between three different cases:
	
	\textbf{a\textsc{)}} If $I$ $\in$ $\mathcal{M}$$in$$\mathcal{MCR}$, then $I$ is a rare correlated pattern and we know its support as well as its \textit{bond} value.
	
	\textbf{b\textsc{)}} If $\nexists$ $J$ $\in$ $\mathcal{M}$$in$$\mathcal{MCR}$ as $J$ $\subseteq$ $I$ or $\nexists$ $Z$ $\in$ $\mathcal{M}$$in$$\mathcal{MCR}$ as $I$ $\subseteq$ $Z$, then $I$ $\notin$ $\mathcal{RCP}$ since $I$ does not belong to any rare correlated equivalence class.
	
	\textbf{c\textsc{)}} Otherwise,  $I$ $\in$ $\mathcal{RCP}$.
	In fact, according to Proposition \ref{prop_set_CRP}, $I$ is correlated since it is included in a correlated pattern, namely $Z$.
	It is also rare since it includes a rare pattern, namely $J$.
	Since the $\mathcal{CRCP}$ set belongs to $\mathcal{M}$$in$$\mathcal{MCR}$,
	it is enough to localize the closed pattern associated to $I$, namely $F$, equal to:
	$F$ = $min_{\subseteq}$\{$I_1$ $\in$ $\mathcal{M}$$in$$\mathcal{MCR}$ $|$ $I$ $\subseteq$ $I_1$\}.
	Then, \textit{bond}\textsc{(}$I$\textsc{)} = \textit{bond}\textsc{(}$F$\textsc{)} and
	\textit{Supp}\textsc{(}$\wedge I$\textsc{)} = \textit{Supp}\textsc{(}$\wedge F$\textsc{)}.\\\mbox{}
\end{proof}
The treatment of these three cases is similar to those of the $\mathcal{RCPR}$ representation, \textsc{(}\textit{cf.} Example \ref{example_3_cas_RMCR} page \pageref{example_3_cas_RMCR}\textsc{)}. The $\mathcal{M}$$in$$\mathcal{MCR}$ representation also constitutes a perfect cover of the $\mathcal{RCP}$ set, since it is included in the
$\mathcal{RCPR}$ representation.

After the introduction of our exact condensed representations, we deal in the following with the approximate concise representation.
\subsection{The $\mathcal{M}$$in$$\mathcal{MM}$$ax$$\mathcal{CR}$ concise approximate representation}
The approximate concise representation, that we introduce, is defined as follows:
\begin{definition} \textbf{The $\mathcal{M}$$in$$\mathcal{MM}$$ax$$\mathcal{CR}$ representation} \label{rep4concise approxdef}\\
	Let $\mathcal{M}$$in$$\mathcal{MM}$$ax$$\mathcal{CR}$ be the representation
	based on the $\mathcal{M}ax$$\mathcal{CRCP}$ set and the $\mathcal{M}in$$\mathcal{MRCP}$ set.
	We have  $\mathcal{M}$$in$$\mathcal{MM}$$ax$$\mathcal{CR}$  = $\mathcal{M}ax$$\mathcal{CRCP}$ $\cup$ $\mathcal{M}in$$\mathcal{MRCP}$.
	For each element $I$ of $\mathcal{M}$$in$$\mathcal{MM}$$ax$$\mathcal{CR}$, the support \textit{Supp}\textsc{(}$\wedge$$I$\textsc{)} and the \textit{bond} value \textit{bond}\textsc{(}$I$\textsc{)} are computed.
\end{definition}
\begin{example} \label{expRC4}
	We have,  $\mathcal{M}in$$\mathcal{MRCP}$ = $\{$$A$, $D$, $BC$, $CE$$\}$
	\textsc{(}\textit{cf.} Example \ref{expMM}\textsc{)} and
	$\mathcal{M}ax$$\mathcal{CRCP}$ = $\{$$ACD$, $ABCE$$\}$
	\textsc{(}\textit{cf.} Example \ref{example_MaxCCRP}\textsc{)}.
	Therefore,  $\mathcal{M}$$in$$\mathcal{MM}$$ax$$\mathcal{CR}$ =
	$\{$\textsc{(}$A$, 3, $\displaystyle\frac{3}{3}$\textsc{)},
	\textsc{(}$D$, 1, $\displaystyle\frac{1}{1}$\textsc{)},
	\textsc{(}$BC$, 3, $\displaystyle\frac{3}{5}$\textsc{)},
	\textsc{(}$CE$, 3, $\displaystyle\frac{3}{5}$\textsc{)},
	\textsc{(}$ACD$, 1, $\displaystyle\frac{1}{4}$\textsc{)},
	\textsc{(}$ABCE$, 2, $\displaystyle\frac{2}{5}$\textsc{)}$\}$.
\end{example}
In the previous example, the  $\mathcal{M}$$in$$\mathcal{MM}$$ax$$\mathcal{CR}$
representation has six elements less than the
$\mathcal{RCPR}$ set, eleven elements less than the $\mathcal{MM}$$ax$$\mathcal{CR}$ set and one element less than
$\mathcal{M}$$in$$\mathcal{MCR}$.
However, this representation can not exactly derive the support and the \textit{bond} values of a given rare correlated pattern.
\begin{theorem}\label{representation4_approx}
	The $\mathcal{M}$$in$$\mathcal{MM}$$ax$$\mathcal{CR}$ is an \textbf{approximate} concise representation of the
	$\mathcal{RCP}$ set of the rare correlated patterns.
\end{theorem}
\begin{proof}
	For a given pattern  $I$ $\subseteq$ $\mathcal{I}$, we can determine thanks to this representation whether
	$I$ is rare correlated  or not. It suffices to find two patterns $J$ and $Z$ belonging to $\mathcal{M}$$in$$\mathcal{MM}$$ax$$\mathcal{CR}$
	such as  $J$ $\subseteq$ $I$ $\subseteq$ $Z$.
	If $J$ or $Z$ can not be found, then $I$ $\notin$ $\mathcal{RCP}$.
	However, the support and \textit{bond} values can be exactly derived only if
	$I$ $\in$ $\mathcal{M}$$in$$\mathcal{MM}$$ax$$\mathcal{CR}$.
	Otherwise, this representation can not offer an exact derivation of these values, since it may not contain any representative element of the equivalence class  of $I$
	\textsc{(}\textit{i.e.} neither the closed pattern if it does not belong to the $\mathcal{M}ax$$\mathcal{CRCP}$ set
	nor to the associated minimal if they don't belong to the
	$\mathcal{M}in$$\mathcal{MRCP}$ set\textsc{)}.
	We propose in this case an approximate process in order to get these values.
	We define, in this regard, the maximal and minimal borders of the conjunctive, the disjunctive and the \textit{bond} value of a correlated rare
	pattern $I$.
	Let,
	
	$\bullet$ \textit{R1} = $\max$$\{$Supp\textsc{(}$\wedge$$F$\textsc{)}, $F$ $\in$ $\mathcal{M}ax$$\mathcal{CRCP}$ $|$ $I$ $\subseteq$ $F$$\}$,
	
	$\bullet$ \textit{R2} =  $\min$$\{$Supp\textsc{(}$\wedge$$G$\textsc{)}, $G$ $\in$  $\mathcal{M}in$$\mathcal{MRCP}$ $|$ $G$ $\subseteq$ $I$$\}$,
	
	$\bullet$ \textit{R3} = $\min$$\{$Supp\textsc{(}$\vee$$F$\textsc{)}, $F$ $\in$  $\mathcal{M}ax$$\mathcal{CRCP}$ $|$ $I$ $\subseteq$ $F$$\}$ and
	
	$\bullet$ \textit{R4} = $\max$$\{$\textit{Supp}\textsc{(}$\vee$$G$\textsc{)}, $G$ $\in$  $\mathcal{M}in$$\mathcal{MRCP}$ $|$ $G$ $\subseteq$ $I$$\}$.
	
	We define, therefore, the minimal and maximal borders in terms of
	\textit{R1} and of \textit{R2} as follows.
	Let \textit{MinConj} be the minimal border of the conjunctive support of the pattern $I$, \textit{i.e.}, \textit{MinConj} = $\min$\textsc{(}\textit{R1}, \textit{R2}\textsc{)}. Let
	\textit{MaxConj} be the maximal border of the conjunctive support of the pattern $I$, \textit{i.e.},
	\textit{MaxConj} = $\max$\textsc{(}\textit{R1}, \textit{R2}\textsc{)}.
	
	According to the disjunctive support of a given pattern $I$, we define the maximal and minimal borders in terms of
	\textit{R3} and of \textit{R4} as follows.
	Let \textit{MinDisj} be the minimal border of the disjunctive support,
	\textit{MinDisj} = $\min$\textsc{(}\textit{R3}, \textit{R4}\textsc{)} and let
	\textit{MaxDisj} be the maximal border of the disjunctive support,
	\textit{MaxDisj} = $\max$\textsc{(}\textit{R3}, \textit{R4}\textsc{)}.
	
	Consequently, the conjunctive support of a rare correlated pattern $I$ will be included between \textit{MinConj} and \textit{MaxConj}.
	In the same way,  the disjunctive support will be included between \textit{MinDisj} and \textit{MaxDisj}.
	Formally, we have  \textit{Supp}\textsc{(}$\wedge$$I$\textsc{)} $\in$ $[$\textit{MinConj}, \textit{MaxConj}$]$ and
	\textit{Supp}\textsc{(}$\vee$$I$\textsc{)} $\in$ $[$\textit{MinDisj}, \textit{MaxDisj}$]$.
	
	Therefore, we define the minimal and maximal borders of the \textit{bond} value of a rare correlated pattern $I$
	in terms of \textit{MinConj}, \textit{MinDisj}, \textit{MaxConj} and \textit{MaxDisj} as follows.
	
	Since \textit{MinDisj} $\leq$ \textit{Supp}\textsc{(}$\vee$$I$\textsc{)} $\leq$ \textit{MaxDisj},
	then we have
	$\displaystyle\frac{1}{\textit{MaxDisj}}$
	$\leq$ $\displaystyle\frac{1}{\textit{Supp}\textsc{(}\vee I\textsc{)}}$
	$\leq$  $\displaystyle\frac{1}{\textit{MinDisj}}$.
	
	As \textit{Supp}\textsc{(}$\wedge$I\textsc{)} $>$ 0  then
	we can deduce that,
	$\displaystyle\frac{\textit{Supp}\textsc{(}\wedge I\textsc{)}}{\textit{MaxDisj}}$
	$\leq$
	$\displaystyle\frac{\textit{Supp}\textsc{(}\wedge I\textsc{)}}{\textit{Supp}\textsc{(}\vee I\textsc{)}}$  $\leq$
	$\displaystyle\frac{\textit{Supp}\textsc{(}\wedge I\textsc{)}}{\textit{MinDisj}}$.
	This is equivalent to,
	$\displaystyle\frac{\textit{Supp}\textsc{(}\wedge I\textsc{)}}{\textit{MaxDisj}}$  $\leq$
	\textit{bond}\textsc{(}$I$\textsc{)}   $\leq$
	$\displaystyle\frac{\textit{Supp}\textsc{(}\wedge I\textsc{)}}{\textit{MinDisj}}$.
	Already, \textit{MinConj} $\leq$ \textit{Supp}\textsc{(}$\wedge$ $I$\textsc{)}
	, then
	$\displaystyle\frac{\textit{MinConj}}{\textit{MaxDisj}}$  $\leq$
	\textit{bond}\textsc{(}$I$\textsc{)}.
	Besides, \textit{Supp}\textsc{(}$\wedge$$I$\textsc{)} $\leq$ \textit{MaxConj}
	then \textit{bond}\textsc{(}$I$\textsc{)} $\leq$
	$\displaystyle\frac{\textit{MaxConj}}{\textit{MinDisj}}$.
	We then have,
	$\displaystyle\frac{\textit{MinConj}}{\textit{MaxDisj}}$  $\leq$
	\textit{bond}\textsc{(}$I$\textsc{)} $\leq$
	$\displaystyle\frac{\textit{MaxConj}}{\textit{MinDisj}}$.
	
	We conclude then that, \textit{bond}\textsc{(}$I$\textsc{)} $\in$
	$[$\textit{Minbond}, \textit{Maxbond}$]$, with \textit{Minbond} = $\displaystyle\frac{\textit{MinConj}}{\textit{MaxDisj}}$
	and \textit{Maxbond} = $\displaystyle\frac{\textit{MaxConj}}{\textit{MinDisj}}$.
\end{proof}
\begin{example}
	With respect to Example \ref{expRC4}, we have $ABE$ is a correlated rare pattern
	since  \textsc{(}$A$  $\subset$  $ABE$ $\subset$  $ABCE$\textsc{)}.
	The conjunctive, disjunctive and the \textit{bond} value of $ABE$ are approximated as follows:\\
	$\bullet$ \textit{R1} = Supp\textsc{(}$\wedge$$ABCE$\textsc{)} = 2,\\
	$\bullet$ \textit{R2} = Supp\textsc{(}$\wedge$$A$\textsc{)} = 3,\\
	$\bullet$ \textit{R3} = Supp\textsc{(}$\vee$$ABCE$\textsc{)} = 5,\\
	$\bullet$ \textit{R4} = Supp\textsc{(}$\vee$$A$\textsc{)} = 5.\\
	Thus, we have \textit{MinConj} = $\min$\textsc{(}\textit{R1}, \textit{R2}\textsc{)} =  $\min$\textsc{(}2, 3\textsc{)} = 2, \textit{MaxConj} = $\max$\textsc{(}\textit{R1}, \textit{R2}\textsc{)} =  $\max$\textsc{(}2, 3\textsc{)} = 3,
	and $\min$\textsc{(}\textit{R3}, \textit{R4}\textsc{)} =  $\max$\textsc{(}\textit{R3}, \textit{R4}\textsc{)}
	= $\min$\textsc{(}5, 5\textsc{)} = $\max$\textsc{(}5, 5\textsc{)} = 5.
	We have, therefore, \textit{MinDisj} =  \textit{MaxDisj} = 5.
	This implies that, \textit{Minbond} = $\displaystyle\frac{\textit{MinConj}}{\textit{MaxDisj}}$ = $\displaystyle\frac{2}{5}$ and
	\textit{Maxbond} = $\displaystyle\frac{\textit{MaxConj}}{\textit{MinDisj}}$ = $\displaystyle\frac{3}{5}$.
	
	Consequently,  we have Supp\textsc{(}$\wedge$$ABE$\textsc{)} $\in$ $[$2, 3$]$,
	Supp\textsc{(}$\vee$$ABE$\textsc{)} $\in$ $[$5, 5$]$ so Supp\textsc{(}$\vee$$ABE$\textsc{)} = 5 and
	bond\textsc{(}$ABE$\textsc{)}  $\in$ $[$$\displaystyle\frac{2}{5}$, $\displaystyle\frac{3}{5}$$]$.
	
	We remark, according to the extraction context illustrated by Table \ref{Base_transactions} \textsc{(}Page \pageref{Base_transactions} \textsc{)}, that the conjunctive, disjunctive and the \textit{bond} values
	of the pattern $ABE$ corresponds respectively to 2, 5 and $\displaystyle\frac{2}{5}$.
	These values does not contradict the previously obtained approximate values. We affirm that the approximation mechanism offered by
	the approximate concise representation $\mathcal{RM}$$in$$\mathcal{MM}$$ax$$\mathcal{F}$ is valid.
\end{example}
After presenting the condensed representations associated to the rare correlated patterns, we focus on the next section on the condensed representation associated to the $\mathcal{FCP}$ set of frequent correlated patterns

\section{Condensed representation of frequent correlated patterns} \label{section_RCfcp}
Based on the $f_{bond}$ closure operator, a condensed representation which cover the frequent correlated patterns was proposed in \cite{tarekds2010}. This representation is based on the frequent closed correlated patterns. The proposed representation is considered more concise than the representation based on minimal correlated patterns thanks to the fact that a $f_{bond}$
equivalence class always contains only one closed pattern, but
potentially several minimal patterns.

Before introducing the representation, let us define the two discussed sets of frequent closed correlated patterns
and of the frequent minimal correlated pattern associated to the $f_{bond}$ operator.

\begin{definition}  \textbf{Frequent closed correlated pattern}\\
	The set $\mathcal{FCCP}$ of frequent closed correlated patterns is equal to: $\mathcal{FCCP}$ = $\{$ $I$ $\in$ $\mathcal{CCP}$ $|$ \textit{Supp}\textsc{(}$I$\textsc{)} $\geq$ \textit{minsupp}$\}$.
\end{definition}

\begin{definition}\label{définitionMotif minimal} \textbf{Frequent minimal correlated pattern}\\
	Let $\textit{I}$ $\in$ $\mathcal{FCP}$. The pattern \textit{I} is said to be minimal if and only if $\forall$ $i$ $\in$ $\textit{I}$, \textit{bond}\textsc{(}\textit{I}\textsc{)} $<$ \textit{bond}\textsc{(}\textit{I}$\backslash$$\{i\}$\textsc{)} or, equivalently, $\nexists$ $I_1$ $\subset$ $I$ such that $f_{bond}\textsc{(}I\textsc{)}$ = $f_{bond}\textsc{(}I_1\textsc{)}$.
\end{definition}
\begin{example}
	Consider the extraction context sketched by Table
	\ref{Base_transactions} \textsc{(}Page \pageref{Base_transactions} \textsc{)}. For \textit{minsupp} = \textit{4} and
	\textit{minbond} = \textit{0.20},
	the $\mathcal{FCCP}$ set of frequent closed correlated pattern is equal to:
	$\{$\textsc{(} \textsc{(}\textsl{C}, 4, 4\textsc{)},
	\textsc{(}\textsl{BE}, 4, 4\textsc{)}$\}$ while
	the frequent minimal correlated pattern are the items
	\textsc{(}\textsl{B}, 4, 4\textsc{)}, \textsc{(}\textsl{C}, 4, 4\textsc{)} and \textsc{(}\textsl{E}, 4, 4\textsc{)}.
\end{example}
Now, let us define the new concise representation of frequent correlated patterns based on
the frequent closed correlated patterns associated to the \textit{bond} measure.
\begin{definition}\label{rep concise exacte def}
	The representation $\mathcal{RFCCP}$ based on the set of frequent closed
	correlated patterns associated to $f_{bond}$ is defined as
	follows:
	\vspace{-0.1cm}
	\begin{center}
		$\mathcal{RFCCP} = \{$\textsc{(}$\textit{I}$,
		\textit{Supp}\textsc{(}$\wedge$\textit{I}\textsc{)},
		\textit{Supp}\textsc{(}$\vee$\textit{I}\textsc{)}\textsc{)} $|$
		$I\in\mathcal{FCCP}$ $\}$.\end{center}
\end{definition}

\begin{example}
	Consider the extraction context sketched by Table
	\ref{Base_transactions} \textsc{(}Page \pageref{Base_transactions} \textsc{)}.
	For \textit{minsupp} = \textit{4} and
	\textit{minbond} = \textit{0.20}, the representation
	$\mathcal{RFCCP}$ of the $\mathcal{FCP}$ set is equal to:
	$\{$\textsc{(} \textsc{(}\textsl{C}, 4, 4\textsc{)},
	\textsc{(}\textsl{BE}, 4, 4\textsc{)}$\}$.
\end{example}

The next theorem proves that the proposed $\mathcal{RFCCP}$ representation is a condensed exact representation of the $\mathcal{FCP}$ set of frequent correlated patterns.
\vspace{-0.1cm}
\begin{theorem}
	The representation $\mathcal{RFCCP}$ constitutes an exact concise representation of the $\mathcal{FCP}$ set.
\end{theorem}
\vspace{-0.45cm}
\begin{proof}\textit{Thanks to a reasoning by recurrence, we will demonstrate that, for an arbitrary pattern $I\subseteq \mathcal{I}$, its $f_{bond}$ closure, $f_{bond}$\textsc{(}$I$\textsc{)}, belongs to $\mathcal{FCCP}$ if it is frequent correlated. In this regard, let $\mathcal{FMCP}_k$ be the set of frequent minimal correlated patterns of size $k$ and $\mathcal{FCCP}_k$ be the associated set of closures by $f_{bond}$. The hypothesis is verified for single items $i$ inserted in $\mathcal{FMCP}_1$, and their closures $f_{bond}\textsc{(}i\textsc{)}$ are inserted in $\mathcal{FCCP}_1$ if $\textit{Supp}\textsc{(}\wedge i\textsc{)} \geq$ \textit{minsupp} \textsc{(}since $\forall$ $i$ $\in$ $\mathcal{I}$, $\textit{bond}\textsc{(}i\textsc{)}$ $=$ $1$ $\geq$ \textit{minbond}\textsc{)}. Thus,
		$f_{bond}\textsc{(}i\textsc{)}\in\mathcal{FCCP}$. Now, suppose
		that $\forall I\subseteq \mathcal{I}$ such as $|I| = n$. We have
		$f_{bond}\textsc{(}I\textsc{)}\in\mathcal{FCCP}$ if $I$ is frequent correlated. We show that, $\forall I\subseteq \mathcal{I}$ such as $|I|$ $=$ \textsc{(}$n
		+ 1$\textsc{)}, we have $f_{bond}\textsc{(}I\textsc{)}\in\mathcal{FCCP}$ if $I$ is frequent correlated. Let $I$ be a pattern of size \textsc{(}$n + 1$\textsc{)}. Three situations are possible:\\
		\textbf{\textsc{(}a\textsc{)}} if $I \in \mathcal{FCCP}$, then necessarily
		$f_{bond}\textsc{(}I\textsc{)}\in\mathcal{FCCP}$ since $f_{bond}$ is idempotent.\\
		\textbf{\textsc{(}b\textsc{)}} if $I$ $\in \mathcal{FMCP}_{n+1}$, then
		$f_{bond}\textsc{(}I\textsc{)}\in\mathcal{FCCP}_{n+1}$ and, hence,
		$f_{bond}\textsc{(}I\textsc{)}\in\mathcal{FCCP}$.\\
		\textbf{\textsc{(}c\textsc{)}} if $I$ is neither closed nor minimal -- $I \notin \mathcal{FCCP}$ and $I \notin \mathcal{FMCP}_{n+1}$ -- then $\exists I_1 \subset I$ such as $|I_1|$ $=$ $n$ and \textit{bond}\textsc{(}$I$\textsc{)} $=$ \textit{bond}\textsc{(}$I_1$\textsc{)}. In fact, $f_{bond}$\textsc{(}$I$\textsc{)} $=$
		$f_{bond}$\textsc{(}$I_1$\textsc{)}, and $I$ is then frequent correlated. Moreover, using the hypothesis, we have
		$f_{bond}$\textsc{(}$I_1$\textsc{)} $\in$ $\mathcal{FCCP}$ and, hence,
		$f_{bond}$\textsc{(}$I$\textsc{)} $\in$ $\mathcal{FCCP}$.}
\end{proof}



\section{Conclusion} \label{ConcChap4}
In this chapter, we studied the frequent correlated and the rare correlated patterns according to the  \textit{bond} correlation measure. We equally described the equivalence classes induced by the $f_{bond}$
closure operator associated to the \textit{bond} measure. Then, we introduced the condensed exact and approximate representations associated to rare correlated patterns and also to frequent correlated ones.
We proved their theoretical  properties of accuracy and compactness. This chapter was concluded with the
condensed representation associated to the $\mathcal{FCP}$ set of frequent correlated patterns.
In the next chapter, we propose a mining approach, called \textsc{GMJP}, allowing the extraction of both frequent correlated patterns, rare correlated patterns and their associated concise representations.

\chapter{Extraction Approach of Correlated Patterns and associated Condensed Representations}\label{ch_5}
\section{Introduction}
This chapter is dedicated to the presentation of our approach called \textsc{Gmjp}. Section \ref{SeCs} is devoted to the analysis of the different
integration mechanism of the constraints of frequency and of correlation. Section \ref{section_algo} presents the description of the \textsc{Gmjp} approach, going from general to specificities. We describe the three different steps of \textsc{Gmjp}, 
then we present \textsc{Opt-Gmjp} the optimized version of the \textsc{Gmjp} algorithm in Section \ref{se-optgmjp}.
We compute the theoretical approximate time complexity of \textsc{Gmjp} in Section \ref{se-comp}. In Section \ref{sec_Regeneration}, we describe the regeneration strategy of the rare correlated patterns from the  $\mathcal{RCPR}$ representation.
We conclude the chapter in Section \ref{ConcChap5}.
\section{Integration mechanism of the constraints}\label{SeCs}
This chapter is dedicated to the presentation of the extraction approach of both frequent correlated and rare correlated itemsets as well as their associated condensed representations.
For the case of the frequent correlated patterns, the extraction process is straightforward since the set of frequent correlated patterns induces an ideal order on the itemset lattice and fulfills an anti-monotonic constraint. Whereas, for the rare correlated patterns, we have to handle two constraints of distinct types: monotonic and anti-monotonic. The evaluation order of the constraints is of paramount importance given the opposite nature of the handled constraints of rarity and of correlation. Thus, we distinguish two different possible scenarios:

$\bullet$ First Scenario:  We first apply the rarity constraint and the associated conjunctive closure operator, then we apply the correlation constraint. \\

$\bullet$ Second Scenario: We apply the correlation constraint and the associated  $f_{bond}$ closure operator, then we integrate the rarity constraint.

These two scenarios will be analyzed in what follows in order to justify our choice of the adequate scenario in our proposed extraction approach.
\subsection{First Scenario} 
In this case, we firstly extract the rare patterns. Then, we filter the retained rare patterns by the correlation constraint. Thus, only the rare correlated patterns are retained. In this situation, in order to reduce the redundancy among the rare correlated itemsets, we apply the conjunctive closure operator associated
to the conjunctive support \cite{ganter99}. This latter splits the itemset lattice into disjoint equivalence classes where, for each pattern, this closure preserve only the conjunctive support. Consequently, the whole set of the itemsets that appear in the same transactions are merged into the same equivalence class. They share the same conjunctive support, the same conjunctive closure, but they have eventually different disjunctive supports. Thereby,  these rare equivalence classes,  \textit{i.e}, those containing just rare patterns, are evaluated by the anti-monotonic constraint of correlation.  The rare patterns belonging to the same class, are then divided into rare correlated patterns and rare non-correlated ones as shown by Figure \ref{Exp1}.

\begin {figure}
{
	\includegraphics[scale = 0.5]{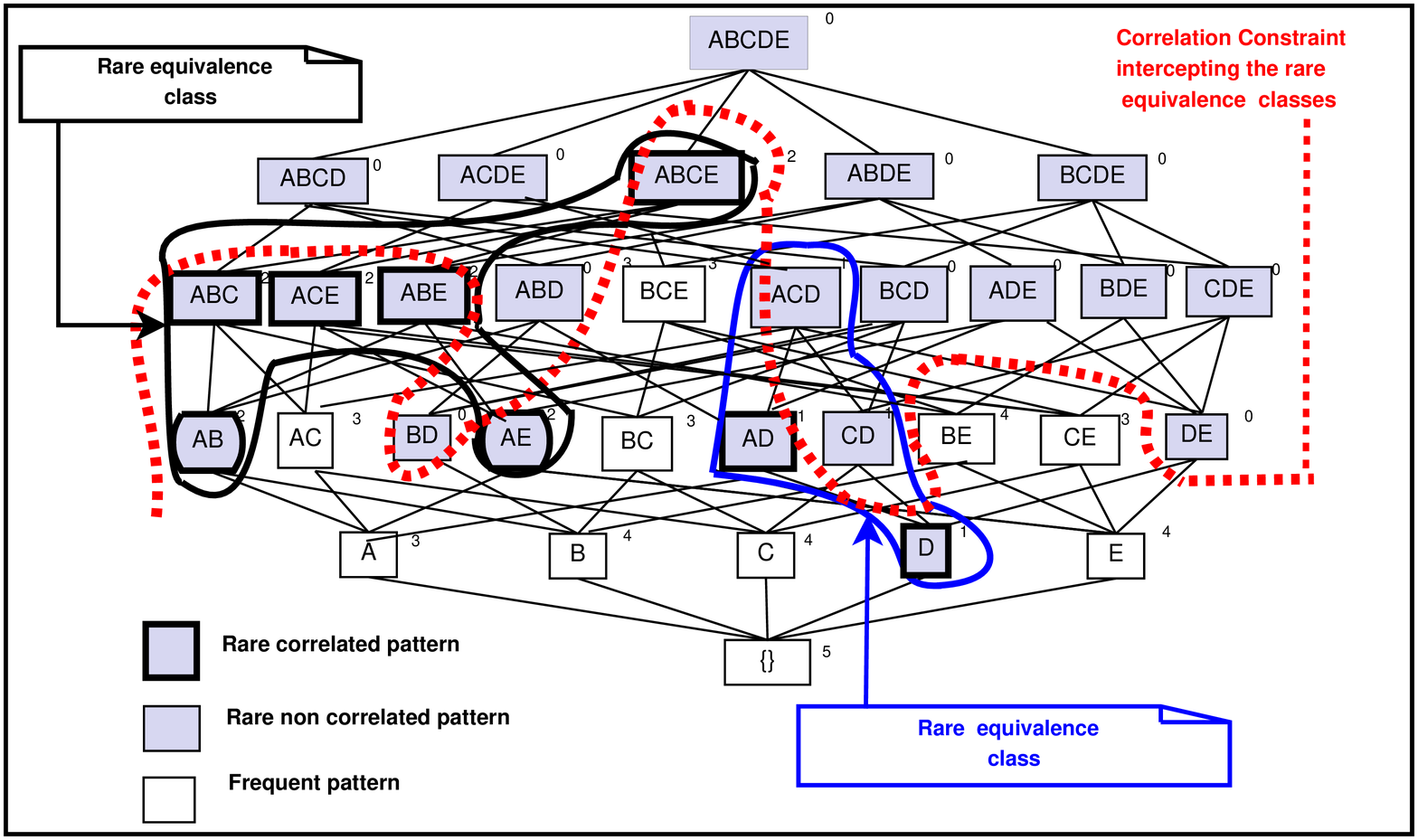}
	\hspace{-3.5cm}
	\caption{Effect of the integration of the correlation constraint for \textit{minsupp} = 3 and \textit{minbond} = 0.3.}
	\label{Exp1}}
\end{figure}

\begin{example}
Let us consider the extraction context given by Table \ref{Base_transactions} 
\textsc{(}Page \pageref{Base_transactions}\textsc{)}. For \textit{minsupp} = 3, we distinguish two rare equivalence classes $\mathcal{C}_1$ and $\mathcal{C}_2$ shown in Figure \ref{Exp1} and composed by the following elements :
\begin{itemize}
	\item $\mathcal{C}_1$ contains the itemsets \texttt{D}, \texttt{AD}, \texttt{CD} and \texttt{ACD}. $\mathcal{C}_1$ has the value of conjunctive support equal to 1 and the conjunctive closed pattern is \texttt{ACD}.
	\item $\mathcal{C}_2$ contains the itemsets \texttt{AB}, \texttt{AE}, \texttt{ABC}, \texttt{ABE}, \texttt{ACE}, and \texttt{ABCE}. $\mathcal{C}_2$ has the value of conjunctive support equal to 2 and the conjunctive closed pattern is \texttt{ABCE}.
\end{itemize}
Let us apply the correlation constraint for a minimal threshold \textit{minbond} = 0.3. For the
$\mathcal{C}_1$ equivalence class,  the patterns $\{$\textsc{(}\texttt{D}, 1, $\displaystyle\frac{1}{1}$\textsc{)}, \textsc{(}\texttt{AD}, 1, $\displaystyle\frac{1}{3}$\textsc{)}$\}$
are rare correlated whereas \textsc{(}\texttt{CD}, 1, $\displaystyle\frac{1}{4}$\textsc{)}, \textsc{(}\texttt{ACD}, 1, $\displaystyle\frac{1}{4}$\textsc{)} are rare non-correlated itemsets.
This is explained by the fact that the elements of the equivalence class $\mathcal{C}_1$ do not share the same disjunctive support, consequently they do not share the same \textit{bond} value. On the other side, all of the elements of the equivalence class $\mathcal{C}_2$ are rare correlated.
\end{example}
\subsection{Second Scenario} 
The second scenario consists in extracting all of the correlated patterns and partitioning them into equivalence classes thanks to the $f_{bond}$ closure operator, then filtering out the obtained equivalence classes by the rarity constraint. In fact, all the itemsets belonging to the same equivalence class share obviously the same conjunctive, disjunctive supports and the same \textit{bond} measure.
Consequently, when considering the anti-monotonic constraint of correlation, we can distinguish two kinds of classes namely: correlated classes and non-correlated ones. The good question that we have to think about it is what will be the effect of applying the monotonic constraint of rarity within this classes?  In other words, how these equivalence classes will be affected by the interception of the rarity constraint?

In fact, the  $f_{bond}$ closure operator preserves the \textit{bond} value, the conjunctive, the disjunctive as well as the negative supports in the same equivalence class. Consequently, all the itemsets belonging to the same equivalence class present the same behavior regarding to the constraints of rarity and those of correlation.
In this respect,  the correlated patterns of an equivalence class, are either rare correlated or  they are frequent correlated. This property also hold for the non correlated equivalence classes. Therefore, these classes are not affected by the application of the rarity constraint. As shown by Figure \ref{Exp2}, we have
correlated frequent classes, correlated rare classes, non correlated frequent classes and
non correlated rare classes.

\begin {figure}\parbox{16cm}{
\hspace{-1.0cm}
\includegraphics[scale = 0.45]{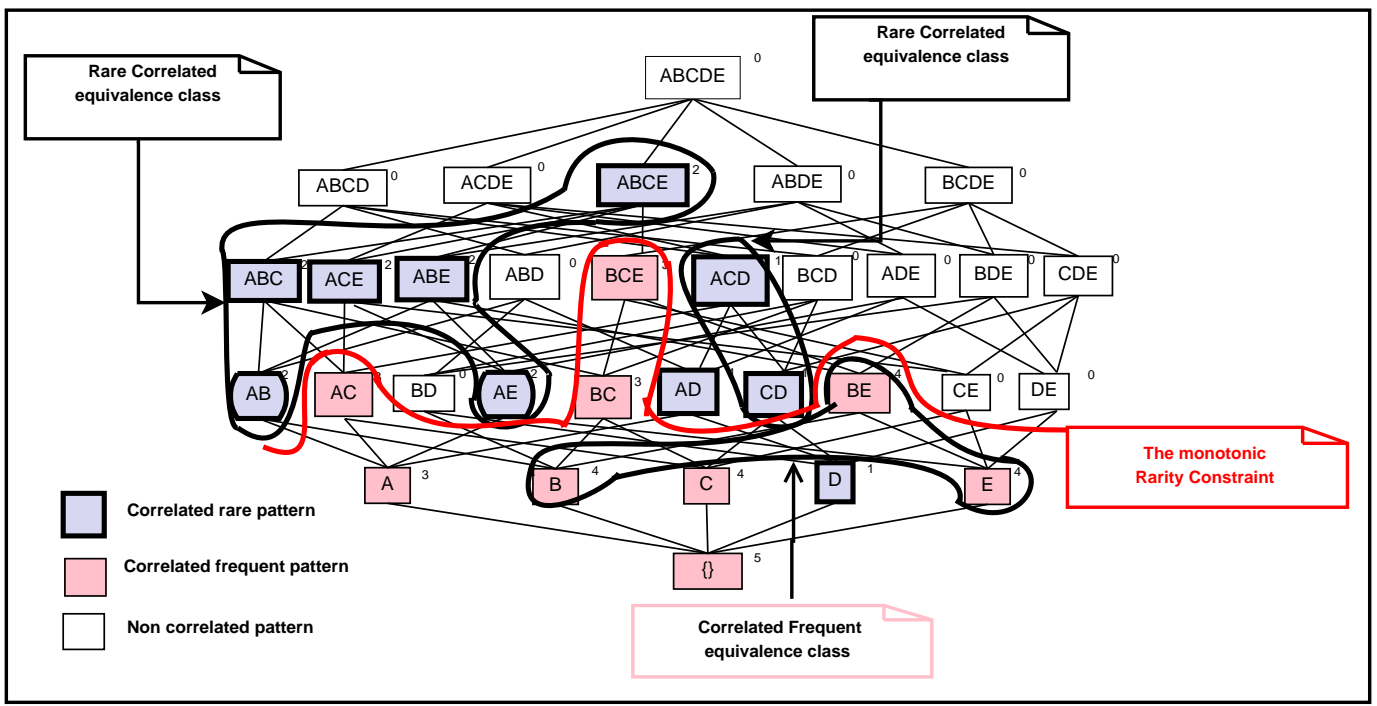}}
\hspace{+0.5cm}
\caption{Effect of the application of the rarity constraint for \textit{minsupp} = 3 and \textit{minbond} = 0.2.}
\label{Exp2}
\end{figure}

\subsection{Summary}

The equivalence classes induced by the $f_{bond}$ closure operator present pertinent characteristics. In fact, this privilege is not offered by the conjunctive closure operator \cite{ganter99}. 
In fact, the state of a given pattern in an equivalence class induced by the conjunctive closure operator is not representative of the state of the other patterns of its same class. In this regard, we are motivated for the application of the second scenario within the design of our mining approach.

We introduce, in the next section, our new \textsc{Gmjp} approach
$^{\textsc{(}}$\footnote{\textsc{Gmjp} stands for \textbf{G}eneric \textbf{M}ining of \textbf{J}accard \textbf{P}atterns. We 
note, by sake of accuracy, that the notation of \textit{Jaccard} measure corresponds to the \textit{bond} measure.}$^{\textsc{)}}$.
\section{The \textsc{Gmjp} approach}\label{section_algo}
We introduce in this section the \textsc{Gmjp} approach which allows, according to the user's input parameters, the extraction of the desired output. As shown by Figure \ref{new_overview}, four different scenarios are possible for running the \textsc{Gmjp} approach: \\
$\bullet$ \textbf{First Scenario}: outputs the whole set $\mathcal{FCP}$ of frequent correlated patterns,\\
$\bullet$ \textbf{Second Scenario}: outputs the $\mathcal{RFCCP}$ concise exact representation of the $\mathcal{FCP}$ set,\\
$\bullet$ \textbf{Third Scenario}: outputs the whole set $\mathcal{RCP}$ of rare correlated patterns,\\
$\bullet$ \textbf{Fourth Scenario}: outputs the $\mathcal{RCPR}$ concise exact representation of the $\mathcal{RCP}$ set.

The \textsc{Gmjp} algorithm takes as an input a dataset $\mathcal{D}$, a minimal support threshold \textit{minsupp} and a minimal correlation threshold \textit{minbond}. We mention that \textsc{Gmjp} determines exactly the \textit{support} and the \textit{bond} values of each pattern of the desired output according to the user's parameters.
\begin{figure}[htbp]
\begin{center}
\includegraphics[scale=0.35]{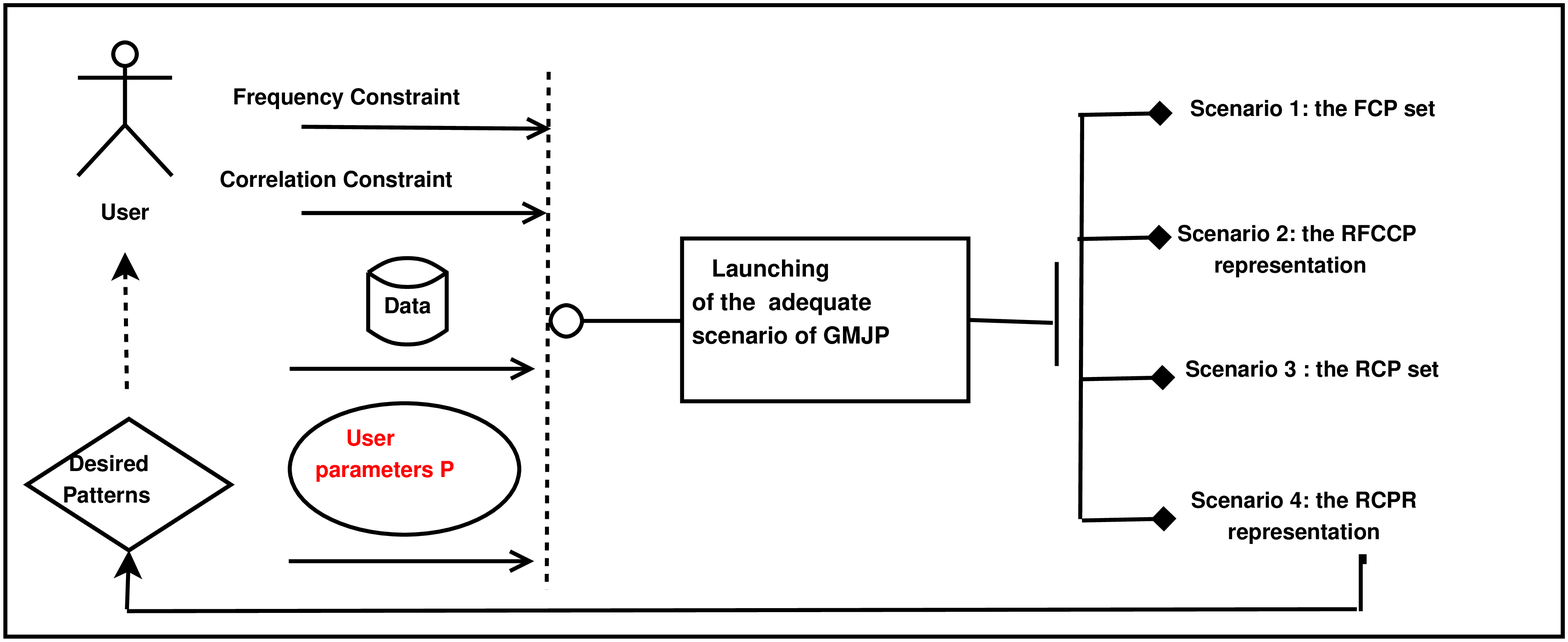}
\caption{Overview of \textsc{Gmjp} approach.}
\label{new_overview}
\end{center}
\end{figure}

\subsection{Overview of the approach}

We illustrate the different steps of \textsc{Gmjp} when running the fourth script aiming to extract the $\mathcal{RCPR}$ representation. Our choice of this fourth scenario is motivated by the fact that the extraction of the $\mathcal{RCPR}$ representation corresponds to the most challenging mining task for \textsc{Gmjp}.

In fact, $\mathcal{RCPR}$ is composed by the set of rare correlated patterns which results from the intersection of two theories \cite{mannila97} induced by the constraints of correlation and rarity. So, this set is neither an order ideal nor an order filter.
Therefore, the localization of the elements of the $\mathcal{RCPR}$ representation is more difficult than the localization of theories corresponding to the conjunction of constraints of the same nature. Indeed, the conjunction of anti-monotonic constraints \textsc{(}\textit{resp.} monotonic\textsc{)} is an anti-monotonic constraint \textsc{(}\textit{resp.} monotonic\textsc{)} \cite{luccheKIS05_MAJ_06}.
For example, the constraint ``being a correlated frequent pattern'' is anti-monotonic, since it results from the conjunction of two anti-monotonic constraints namely,  ``being a correlated pattern'' and ``being a frequent pattern''. This constraint induces, then, an order ideal on the itemset lattice.

In fact, the \textsc{Gmjp} algorithm mainly operates in three steps as depicted by Figure \ref{figure_overview}. The pseudo-code of \textsc{Gmjp} is given by Algorithm \ref{AlgoCRPR}.
\begin{enumerate}
\item A first scan of the dataset is performed in order
to extract all the items and assigning to each item the set of transactions in which it appears.
Then, a second scan of the dataset is carried out in order to identify, for each item, the list of the co-occurrent items
\textsc{(}\textit{cf.} Line 1 Algorithm \ref{AlgoCRPR}\textsc{)}.
\item The second step consists in integrating both of the constraints rarity as well as correlation within a mining process of $\mathcal{RCPR}$. In this situation, this problem is split into independent chunks since each item is separately treated. In fact, for each item \textsc{(}\textit{cf.} Line 2 Algorithm \ref{AlgoCRPR}\textsc{)}, a set of candidates is generated \textsc{(}\textit{cf.} Line (b) Algorithm \ref{AlgoCRPR}\textsc{)}. Once obtained, these candidates are pruned using the following pruning strategies:\\
\textsc{(}\textbf{\textit{a}}\textsc{)} \textbf{The pruning of the candidates which check the cross-support property} 
\textsc{(}\textit{cf.} Line (i) Algorithm \ref{AlgoCRPR}\textsc{)}\cite{tarekds2010}.
In fact, as defined in section \ref{S2chap4} \textsc{(}\textit{cf.} Chapter 4, page \pageref{S2chap4}\textsc{)},
the cross-support property allows to prune non-correlated candidates. 
More clearly, any pattern, containing two items fulfilling the cross-support property w.r.t.
a minimal threshold of correlation, is not correlated. Thus, this
property avoids the computation of its conjunctive and disjunctive
supports, required to evaluate its \textit{bond} value.\\

\textsc{(}\textbf{\textit{b}}\textsc{)} \textbf{The pruning based on the order ideal of the correlated patterns}
\textsc{(}\textit{cf.} Line (ii) Algorithm \ref{AlgoCRPR}\textsc{)}. Recall that the set of correlated patterns induces an order ideal property.
Therefore, each correlated candidate, having a non correlated subset, will be pruned since it will not be a correlated pattern.

Then, the conjunctive, disjunctive supports and the \textit{bond} value of the retained candidates are computed
\textsc{(}\textit{cf.} Line (iii) Algorithm \ref{AlgoCRPR}\textsc{)}.
Thus, the uncorrelated candidates are also pruned. At the level $n$, the local minimal rare correlated patterns of size $n$ are determined among the retained candidates \textsc{(}\textit{cf.} Line (iv) Algorithm \ref{AlgoCRPR}\textsc{)}.  The local closed rare correlated patterns of size $n-1$ are also filtered \textsc{(}\textit{cf.} Line (v) Algorithm \ref{AlgoCRPR}\textsc{)}. This process comes to an end when there is no more candidates to be generated \textsc{(}\textit{cf.} Line (c) Algorithm \ref{AlgoCRPR}\textsc{)}.
\item The third and last step consists of filtering the global minimal rare correlated patterns
\textsc{(}\textit{cf.} Line 3 Algorithm \ref{AlgoCRPR}\textsc{)}  
and the global rare correlated patterns among the two sets of local minimal rare correlated patterns and local closed ones
\textsc{(}\textit{cf.} Line 4 Algorithm \ref{AlgoCRPR}\textsc{)}.
\end{enumerate}
In what follows, we will explain more deeply these different steps of \textsc{Gmjp}. 
\begin{center}
\begin{figure}[htbp]
\includegraphics [scale=0.30] {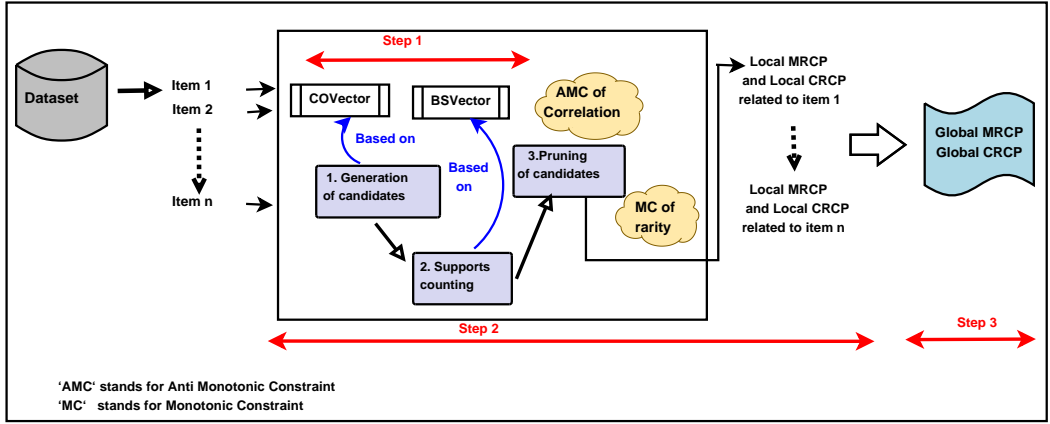}
\caption{Overview of \textsc{Gmjp} when extracting the $\mathcal{RCPR}$ representation.}
\label{figure_overview}
\end{figure}
\end{center}
\incmargin{0.5em}
\linesnumbered
\restylealgo{algoruled}
\begin{algorithm}\label{AlgoCRPR}
\small{ \small{ \caption{\textsc{Gmjp}}}
\SetVline
\setnlskip{-3pt}
\KwData{
	\begin{enumerate}
		\item An extraction context  $\mathcal{C}$.
		\item A minimal correlation threshold \textit{minbond}.
		\item A minimal conjunctive support threshold \textit{minsupp}.
		\item A specification of the desired result `$\mathcal{RCPR}$`.
	\end{enumerate}
} \KwResult{The concise exact representation $\mathcal{RCPR}$ = $\mathcal{MRCP}$ $\cup$ $\mathcal{CRCP}$.}
\small{\Begin{
		\begin{enumerate}
			\item Scan the dataset $\mathcal{C}$ twice to build the BSVector and the COVector \\
			for all the items
			\item For each item $I$ $\in$ $\mathcal{I}$
			\begin{enumerate}
				\item $n$ =2;
				\item Generate the candidates of size $n$ using the COVector of $I$ \\
				
				\item \While {\textsc{(}The number of the generated candidates is not null\textsc{)}}
				{
					\begin{enumerate}
						\item  Prune these candidates w.r.t. the cross-support property \\
						of the \textit{bond} measure
						\item Prune these candidates w.r.t. the order ideal property\\
						of correlated patterns
						\item Compute the conjunctive and disjunctive supports \\
						and the \textit{bond} value  of the maintained candidates
						\item For each candidate $C$
						
						\textbf{If} \textsc{(}IsCorrelated\textsc{(}C\textsc{)} and IsRare\textsc{(}C\textsc{)}\textsc{)} \textbf{then}
						
						{$/*$ Ckeck-Local Minimality of the candidate $C$ $*/$}
						\begin{itemize}
							\item  Update the set of Local Minimal Rare\\
							Correlated Patterns of size $n$
						\end{itemize}

						\item \hspace{-0.20cm} Find Local Closed Rare Correlated Patterns of size $n$$-$$1$
						\item \hspace{-0.20cm} $n$ = $n$$+$$1$
						\item \hspace{-0.20cm} Generate candidates of size $n$ using the \textsc{Apriori-Gen}\\
						procedure
					\end{enumerate}
				}
				
			\end{enumerate}
			\item Find all Global Minimal Rare Correlated Patterns
			\item Find all Global Closed Rare Correlated Patterns
			\item \Return{$\mathcal{RCPR}$\;}
		\end{enumerate}
	}
}
}
\end{algorithm}
\subsection{First Step: The power of the bit vectors and of co-occurrent vectors}
Initially, the dataset is scanned in order to extract the items and to build, for each item, the bitset called here
``BSVector''. In fact, a bitset is a container that can store a huge number of bits while optimizing the memory consumption
\textsc{(}For example, 32 elements are stored in a memory block of 4 bytes\textsc{)}.
Each block of memory is treated in just one \texttt{CPU} operation by a 32 bits processor.
Therefore, we were motivated for these kinds of structures within the \textsc{Gmjp} algorithm in order to optimize the
conjunctive and the disjunctive supports computations.

Then, the dataset is scanned again in order to identify, for each item $I$, the list of the co-occurrent items which
corresponds to the items occurring in the same transactions as the item $I$. These latter ones are stored in a vector of integers, called here ``COVector''.
We note that one of the main challenges of the \textsc{Gmjp} algorithm is that it allows pushing two constraints of distinct types and to deliver the output with only two scans of the dataset.
We also uphold that the bitsets, when incorporated into the mining process within the \textsc{Gmjp} algorithm, sharply decrease the size of the memory required to store immediate results and significantly save execution costs.
\subsection{Second Step: Getting the Local Minimal and the Local Closed Rare Correlated Patterns without closure computations}
Worth of mention, the main thrust of the \textsc{Gmjp} algorithm is to break the search space into independent sub-spaces.
In fact, for each item $I$, a level-wise mining process is performed using the COVector containing the co-occurrent items of $I$.
At each level $n$, starting by the second level, a set of candidates is generated, then pruned according to the different pruning strategies
described previously. The minimal rare correlated patterns of size $n$, associated to the item $I$, are called \textbf{Local Minimal Rare Correlated Patterns} and they
are determined by comparing their \textit{bond} values versus those of their respective immediate subsets. Similarly, the closed rare correlated patterns of size $n-1$ associated to the item $I$ are called \textbf{Local Closed Rare Correlated Patterns}, and they
are determined by comparing their \textit{bond} values to those of their respective immediate supersets.

It is also important to mention that the implementation of the different stages of this second step \textsc{(}candidate generation,
evaluation and pruning\textsc{)} was based on simple vectors of integers. Thus, we do not require more complex data structure during the
implementation of the \textsc{Gmjp} algorithm. This feature makes \textsc{Gmjp} a practical approach
for handling both monotonic and anti-monotonic constraints even for large datasets.

One of the major challenges in the design of the \textsc{Gmjp} algorithm is how to perform subset and superset checking to efficiently identify Local Minimal and Local Closed patterns$?$ The answer is to construct and manage a multi-map hash structure,
$^{\textsc{(}}$\footnote{We used in our implementation the \textsc{C++} STL Standard Template Library multi-map.}$^{\textsc{)}}$
in order to store at each level $n$ the rare correlated patterns of size $n$. This technique has been shown to be very powerful since it makes the subset and the superset checking practical even on dense datasets.

Thus, our proposed efficient solution \textsc{(}as we prove it experimentally later\textsc{)}
is to integrate both of the monotonic constraint of rarity and the anti-monotonic constraint of
correlation within the mining process and to identify the local closed rare correlated patterns without closure computing.
\subsection{Third Step: Filtering the Global Minimal and the Global Closed Rare Correlated patterns}
After identifying the local minimal and the local closed rare correlated patterns associated to each item $I$ of the dataset $\mathcal{D}$, the third step consists
in filtering the $\mathcal{MRCP}$ set of Global Minimal Rare Correlated patterns and the $\mathcal{CRCP}$ set of Global Closed Rare Correlated patterns.
This task is performed using two distinct multi-map hash structures. In fact, for each local minimal rare correlated pattern $LM$ previously identified, we check
whether it has a direct subset \textsc{(}belonging to the whole set of local minimal patterns\textsc{)} with the same \textit{bond} value.  If it is not the case, then the local minimal pattern $LM$
is a global minimal rare pattern and it is added to the  $\mathcal{MRCP}$  set. Similarly, for each local closed rare correlated pattern $LC$ previously identified, we check whether
it has a direct superset \textsc{(}belonging to the whole set of local closed patterns\textsc{)} with the same \textit{bond} value.
If it is not the case, then the local closed pattern $LC$ is a global closed rare pattern and it is added to the  $\mathcal{CRCP}$ set of Closed rare correlated patterns.
\begin{remark}
We note that we are limited to the description of the extraction of the $\mathcal{RCPR}$ representation since the post-processing operation of the representations $\mathcal{MM}$ax$\mathcal{CR}$, $\mathcal{M}$in$\mathcal{MCR}$ from the $\mathcal{RCPR}$ representation is obviously done and we prove that the needed execution time is negligible.
\end{remark}

In what follows, we illustrate with a running example of the \textsc{Gmjp} algorithm.
\subsection{Running example}
Let us consider the extraction context $\mathcal{C}$ sketched by Table \ref{Base_transactions}
\textsc{(}Page \pageref{Base_transactions}\textsc{)}. 
First, the BSVectors and the COVectors associated to each item of this dataset are constructed, as we plot by Figure \ref{figure_COVector}.
\begin{figure}[htbp]
\begin{center}
\includegraphics [scale = 0.3]{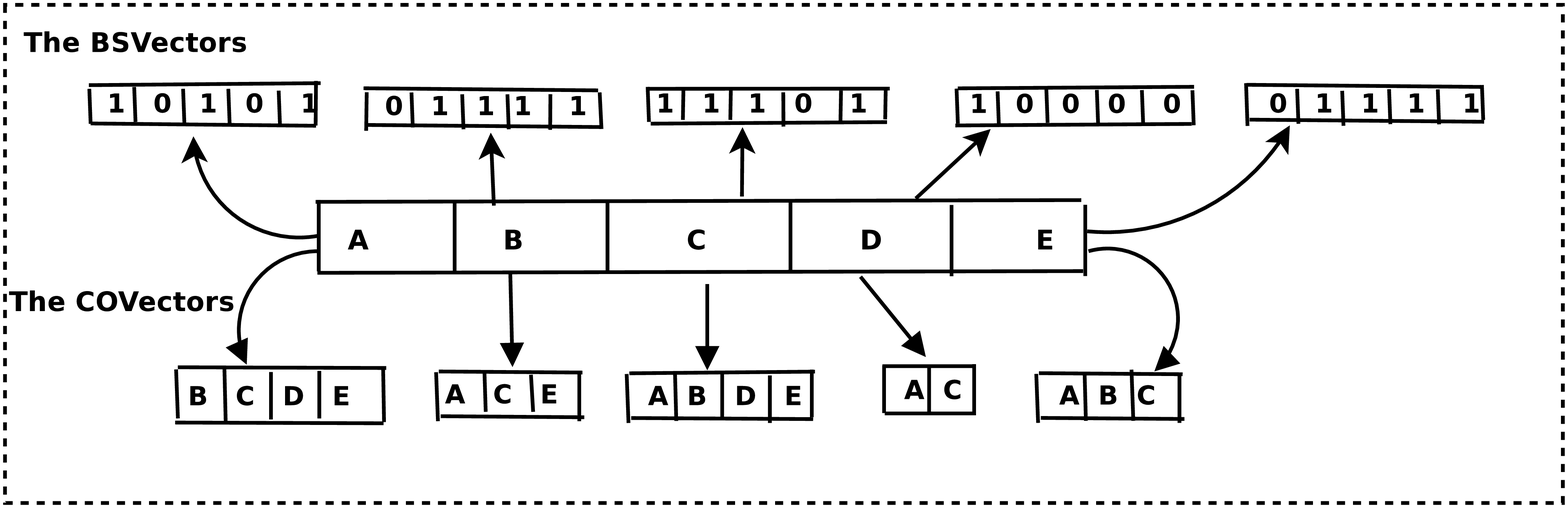}
\caption{The BSVectors and the COVectors associated to the items of the extraction context $\mathcal{C}$.}
\label{figure_COVector}
\end{center}
\end{figure}
\vspace{-0.2cm}
These BSVectors are next used to compute the conjunctive and the disjunctive supports. We have, for example, the item $A$  which belongs to the transactions $\{1, 3, 5\}$ and the item $C$ which belongs to the transactions $\{1, 2, 3, 5\}$.
We, then, have \textit{Supp}\textsc{(}$\wedge$\textit{AC}\textsc{)} = 3 and \textit{Supp}\textsc{(}$\vee$\textit{AC}\textsc{)}\textsc{)} = 4.

The local minimal and the local closed correlated rare patterns associated to each item  $I$ of the dataset $\mathcal{D}$, are extracted. A detailed example of the process of the item $A$ is given by Figure \ref{figure_ExpItemA}.
\begin{figure}[htbp]
\begin{center}
\includegraphics[scale = 0.33]{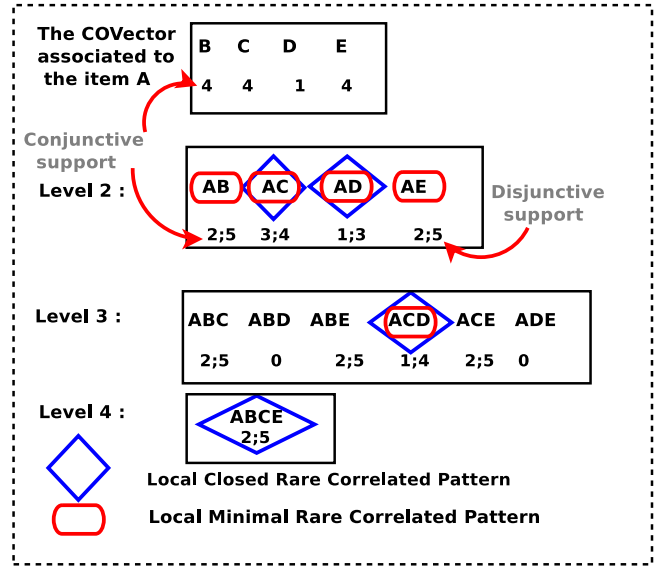}
\caption{Mining Local Minimal and Local Closed Rare Correlated Patterns for the item $A$.}
\label{figure_ExpItemA}
\end{center}
\end{figure}
The finally obtained $\mathcal{RCPR}$ representation, for \textit{minsupp} = 4 and for \textit{minbond} = 0.20, is composed by the following global minimal and global closed correlated patterns:
$\mathcal{RCPR}$ = $\{$ \textsc{(}$A$, $3$, $\displaystyle\frac{3}{3}$\textsc{)},
\textsc{(}$D$, $1$, $\displaystyle\frac{1}{1}$\textsc{)},
\textsc{(}$AB$, $2$, $\displaystyle\frac{2}{5}$\textsc{)},
\textsc{(}$AC$, $3$, $\displaystyle\frac{3}{4}$\textsc{)},
\textsc{(}$AD$, $1$, $\displaystyle\frac{1}{3}$\textsc{)},
\textsc{(}$AE$, $2$, $\displaystyle\frac{2}{5}$\textsc{)},
\textsc{(}$BC$, $3$, $\displaystyle\frac{3}{5}$\textsc{)},
\textsc{(}$CD$, $1$, $\displaystyle\frac{1}{4}$\textsc{)},
\textsc{(}$CE$, $3$, $\displaystyle\frac{3}{5}$\textsc{)},
\textsc{(}$ACD$, $1$, $\displaystyle\frac{1}{4}$\textsc{)},
\textsc{(}$BCE$, $3$, $\displaystyle\frac{3}{5}$\textsc{)}
and \textsc{(}$ABCE$, $2$, $\displaystyle\frac{2}{5}$\textsc{)}$\}$.


Last, it is important to notice that \textsc{Gmjp} is not an exclusive approach in the sense that it can be coupled with other efficient approaches to mine statistically significant patterns.

In the next section, we present \textsc{Opt-Gmjp} the optimized version of the \textsc{Gmjp} algorithm.
\section{\textsc{Opt-Gmjp}: The optimized version of \textsc{Gmjp}} \label{se-optgmjp}
In this section, we present \textsc{Opt-Gmjp} the optimized version of \textsc{Gmjp} approach.
Our improvements cover the four different scenarios S1, S2, S3 and S4 of the \textsc{Gmjp} approach.
Nevertheless, we describe specifically  the optimization of the third scenario S3. In fact, the
latter deals with two constraints of distinct types namely the anti-monotonic constraint of correlation and the monotonic constraint of rarity.

We start  by presenting a generic overview of the \textsc{Opt-Gmjp} algorithm which is illustrated by Figure \ref{figure-optgmjp}.  The pseudo code of \textsc{Opt-Gmjp}, when running the third scenario S3, is given by Algorithm \ref{AlgoCori} while all of the used notations are illustrated in table \ref{notationsTab}. We note that the ``MC`` notation stands for ``Monotonic Constraint`` while the ``AMC`` notation stands for ``Anti-Monotonic Constraint``.
\begin{table}\begin{center}
\footnotesize{
	\begin{tabular}{|lcl|}
		\hline
		$\mathcal{I}$& :& The set of the distinct items.\\
		$C_{n}$& :&A candidate itemset  $C$ of size $n$.\\
		$\mathcal{CAND}_{n}$& :&The Candidate itemsets of size $n$.\\
		$\mathcal{RCP}$& :&The Rare correlated pattern set.\\
		\hline
\end{tabular}}
\caption{The notations used within the \textsc{OptGmjp} algorithm.} \label{notationsTab}
\end{center}
\end{table}

\begin{table}
\parbox{.4\linewidth}{
\centering
\twlrm{
	\begin{tabular}{llllll}
		\hline\noalign{\smallskip}
		& $\texttt{A}$ & $\texttt{B}$ & $\texttt{C}$ & $\texttt{D}$ & $\texttt{E}$\\
		\noalign{\smallskip}
		\hline
		\noalign{\smallskip}
		1 & $\times$  &          &$\times$    & $\times$&         \\
		2 &           & $\times$ &$\times$    &         & $\times$ \\
		3 & $\times$  & $\times$ &$\times$    &         & $\times$ \\
		4 &           & $\times$ &            &         & $\times$ \\
		5 & $\times$  & $\times$ &$\times$    &         & $\times$  \\
		\hline
	\end{tabular}
} 
\caption{The Initial extraction context $\mathcal{C}$.}\label{BD}
}
\hfill
\parbox{.4\linewidth}{
\centering
\hspace{-0.6cm}
\twlrm{
	\begin{tabular}{llllll}
		\hline\noalign{\smallskip}
		& T1 & T2 & T3 & T4 & T5\\
		\noalign{\smallskip}
		\hline
		\noalign{\smallskip}
		$\texttt{A}$ & 1  &  0        &1    & 0&  1       \\
		$\texttt{B}$ &       0    & 1 &1    &     1    & 1 \\
		$\texttt{C}$ & 1  & 1 & 1    &   0      & 1 \\
		$\texttt{D}$ &       1    & 0 &       0     &     0    & 0 \\
		$\texttt{E}$ & 0  & 1 & 1    &   1      & 1  \\
		\hline
	\end{tabular}
	\caption{The transformed extraction context $\mathcal{C}*$.}\label{BDnew}
} 
} 
\end{table}

In fact, the proposed optimizations are of two types: the transformation of the initial extraction context and the reduction of the number of distinct constraints evaluation.
In fact,  the measurement of the impact of pushing the monotonic constraint
of rarity and the anti-monotonic constraint of correlation early within the \textsc{Opt-Gmjp} algorithm helps to measure the selectivity power of each type of constraint during the mining process.

\begin{enumerate}
\item \textbf{Transformation of the initial extraction context} \\
Initially, the extraction context is scanned once to build the new transformed extraction context and to construct an in-memory structure
\textsc{(}\textit{cf.} Line 1 Algorithm \ref{AlgoCori}\textsc{)}.
In fact,  we assign to each item a bitset, each column of this bitset indicates the presence or the absence
of the item in a specified transaction. For example, if the third column of this list contains 0, then the item $I$ is not present in the third transaction. The transformed extraction context associated to the initial context of Table \ref{BD} is given by Table \ref{BDnew}.

\item \textbf{Initialization of the tree-data structure }\\
Initially, the $\mathcal{RCP}$ set of rare correlated pattern is set to the empty-set \textsc{(}\textit{cf.} Line 2 Algorithm \ref{AlgoCori}\textsc{)}. Then, we compute the conjunctive support of the items.
The items are then sorted in an ascending order of their support \textsc{(}\textit{cf.} Line 3 Algorithm \ref{AlgoCori}\textsc{)}.
All the items are added to the nodes of the first level of our tree structure. Thus, they constitute the set of 1-itemsets candidates $\mathcal{CAND}_{1}$ \textsc{(}\textit{cf.} Line 4 Algorithm \ref{AlgoCori}\textsc{)}.
All the items are evaluated according to the monotone constraint of rarity: The rare items are printed to the output set, \textsc{(}\textit{cf.} Line  \textsc{(}b\textsc{)} Algorithm \ref{AlgoCori}\textsc{)},  while frequent ones are not pruned, \textit{i.e.}, they
are maintained.

\item  \textbf{Solving recursively the mining problem}\\
This step is the main optimization of our algorithm. The idea consists in dividing our mining problem into sub-problems.
For each item \texttt{I},   a sub-tree is constructed and a depth-first traversal  is therefore performed.
The candidates of size $n$ are generated by building intersection of itemsets of size $n-1$
\textsc{(}\textit{cf.} Line  \textsc{(}d\textsc{)} Algorithm \ref{AlgoCori}\textsc{)}. They are then evaluated as follows:
\begin{itemize}
\item 	 \textbf{Evaluation of the anti-monotone constraint of correlation}: if the candidate is correlated, we have to distinguish two possible cases: 
\begin{itemize}
	\item \textbf{if the candidate is rare correlated}, then it is added to the result set  \textsc{(}\textit{cf.} Line  \textsc{(}i.1.\textsc{)} Algorithm \ref{AlgoCori}\textsc{)}.
	\item \textbf{if the candidate is frequent correlated}, then it will not be pruned. In fact, a frequent pattern can have rare supersets. In this case, the candidate is maintained and we continue to develop its sub-tree.
\end{itemize}
\item \textbf{if the candidate is not correlated}, then all its supersets will not be correlated according to the monotonicity property of the non-correlation constraint. In this case, the candidate and the associated sub-tree are pruned
\textsc{(}\textit{cf.} Line  \textsc{(}i.2.\textsc{)} Algorithm \ref{AlgoCori}\textsc{)}.
\end{itemize}

This generation and evaluation process is continued, the dedicated procedure is recursively called while there is a number of candidates to be generated \textsc{(}\textit{cf.} Line  \textsc{(}g\textsc{)} Algorithm \ref{AlgoCori}\textsc{)}. 
Finally,  the used  memory is freed and the $\mathcal{RCP}$ is outputted
\textsc{(}\textit{cf.} Line  6 Algorithm \ref{AlgoCori}\textsc{)}. The used data structure is illustrated by Figure \ref{figure-exp-tree}.
\item \textbf{A running Example} 

The example shown in Figure \ref{figure-exp-tree}, illustrates how \textsc{Opt-Gmjp} works. The tree is made from five sorted items, \texttt{D},
\texttt{A}, \texttt{B}, \texttt{C}, and \texttt{E}, with their respective supports 1, 3, 4, 4 and 4 respectively.
We  begin by the node containing the item having the lowest support: which is \texttt{D}.
The anti-monotone constraint  'Amc` is having the \textit{bond} value  $\geq$ 0.20 and the
monotone constraint  'Mc` is having the conjunctive support $<$ 4.
Intersecting \texttt{D} with \texttt{A}, \texttt{B}, \texttt{C}, and \texttt{E} produces \texttt{DA}, \texttt{DB}, \texttt{DC} and \texttt{DE}.
The candidates \texttt{DB} and \texttt{DE} are pruned since they have null support. Thus, their supersets will also have null support.
The candidates \texttt{DA} and \texttt{DC} are rare correlated, thus they pass both constraints Amc and Mc. Their intersection produces
\texttt{DAC} with support 1 and bond value equal to  $\displaystyle\frac{1}{4}$, which also fulfills both constraints.
The rare correlated itemsets are added to the output set and the \texttt{D}-sub-tree is deleted  from our tree.
The \texttt{A}-sub-tree, \texttt{B}-sub-tree, \texttt{C}-sub-tree and \texttt{E}-sub-tree are successively built.
The same process is repeated recursively until there's no more candidates to be generated and evaluated.
Finally,  the obtained result set of correlated rare itemsets is composed by:
\textsc{(}$D$, $1$, $\displaystyle\frac{1}{1}$\textsc{)},
\textsc{(}$A$, $3$, $\displaystyle\frac{3}{3}$\textsc{)},
\textsc{(}$AB$, $2$, $\displaystyle\frac{2}{5}$\textsc{)},
\textsc{(}$AC$, $3$, $\displaystyle\frac{3}{4}$\textsc{)},
\textsc{(}$AE$, $2$, $\displaystyle\frac{2}{5}$\textsc{)},
\textsc{(}$BC$, $3$, $\displaystyle\frac{3}{5}$\textsc{)},
\textsc{(}$CE$, $3$, $\displaystyle\frac{3}{5}$\textsc{)},
\textsc{(}$DA$, $1$, $\displaystyle\frac{1}{3}$\textsc{)},
\textsc{(}$DC$, $1$, $\displaystyle\frac{1}{4}$\textsc{)},
\textsc{(}$ABC$, $2$, $\displaystyle\frac{2}{5}$\textsc{)},
\textsc{(}$ABE$, $2$, $\displaystyle\frac{2}{5}$\textsc{)},
\textsc{(}$ACE$, $3$, $\displaystyle\frac{3}{4}$\textsc{)},
\textsc{(}$BCE$, $3$, $\displaystyle\frac{3}{5}$\textsc{)},
\textsc{(}$DAC$, $1$, $\displaystyle\frac{1}{4}$\textsc{)},
and \textsc{(}$ABCE$, $2$, $\displaystyle\frac{2}{5}$\textsc{)}.
\end{enumerate}
\incmargin{0.5em}
\restylealgo{algoruled}
\begin{algorithm}
\small{ \small{ \caption{\textsc{Opt-Gmjp}}}
\SetVline
\setnlskip{-4pt}
\KwData{
	An Extraction Context $\mathcal{C}$, a minimal correlation threshold \textit{minbond} and a minimal conjunctive support threshold \textit{minsupp}.
} \KwResult{The $\mathcal{RCP}$ set of Rare Correlated Patterns with their respective supports and \textit{bond} values.}
\small{\Begin{
		\begin{enumerate} 
			\item  Scan the extraction context $\mathcal{C}$  once to find $\mathcal{I}$ the set of all the distinct items and to build their associated bitsets
			\item $\mathcal{RCP}$ := $\emptyset$  
			\item Computing the conjunctive support of the items and sorting them in an ascendant order of their support value.
			\item $\mathcal{CAND}_{1}$ := $\mathcal{I}$ 
			\item \textbf{for} each item $I$ $\in$ $\mathcal{I}$ \textsc{(}strating with the item with the lowest support\textsc{)} \textbf{do}
			\begin{enumerate} 
				\item Computing the disjunctive support using the bitsets
				\item \textbf{if} $I$ is rare \textbf{then} 
				
				$\mathcal{RCP}$ := $\mathcal{RCP}$ $\cup$  $I$ 
				
				\textbf{end if}
				\item n := 2
				\item $\mathcal{CAND}_{n}$:= generate-Candidate ($\mathcal{CAND}_{n-1}$) $\blacklozenge$
				\item \textbf{for} each candidate $C_{n}$ $\in$  $\mathcal{CAND}_{n}$  \textbf{do}
				\begin{enumerate} 
					\item  \textbf{if} $C_n$ is correlated \textbf{then}

					\textbf{if} $C_n$ is rare \textbf{then}

					(i.1.)   $\mathcal{RCP}$ := $\mathcal{RCP}$ $\cup$   $C_n$ 
					
					\textbf{end if}
					
					\textbf{else}
					
					(i.2.) 	$\mathcal{CAND}_{n}$ := $\mathcal{CAND}_{n}$ $\backslash$ $C_n$
					
					\textbf{end if}
					
				\end{enumerate} 
				
				\textbf{end for}
				
				\item n := n + 1
				\item Loop to ($\blacklozenge$) while  \textsc{(}$\mathcal{CAND}_{n}$ $\neq$ $\emptyset$\textsc{)}
				
				\textbf{end for}
				
			\end{enumerate} 
			\item \textbf{\Return{$\mathcal{RCP}$}}
		\end{enumerate} 
	}
} 
}
\label{AlgoCori}
\end{algorithm}
\begin{landscape}
\pagestyle{empty}
\begin{figure}
\hspace{-1.5cm}
\includegraphics[scale = 0.65]{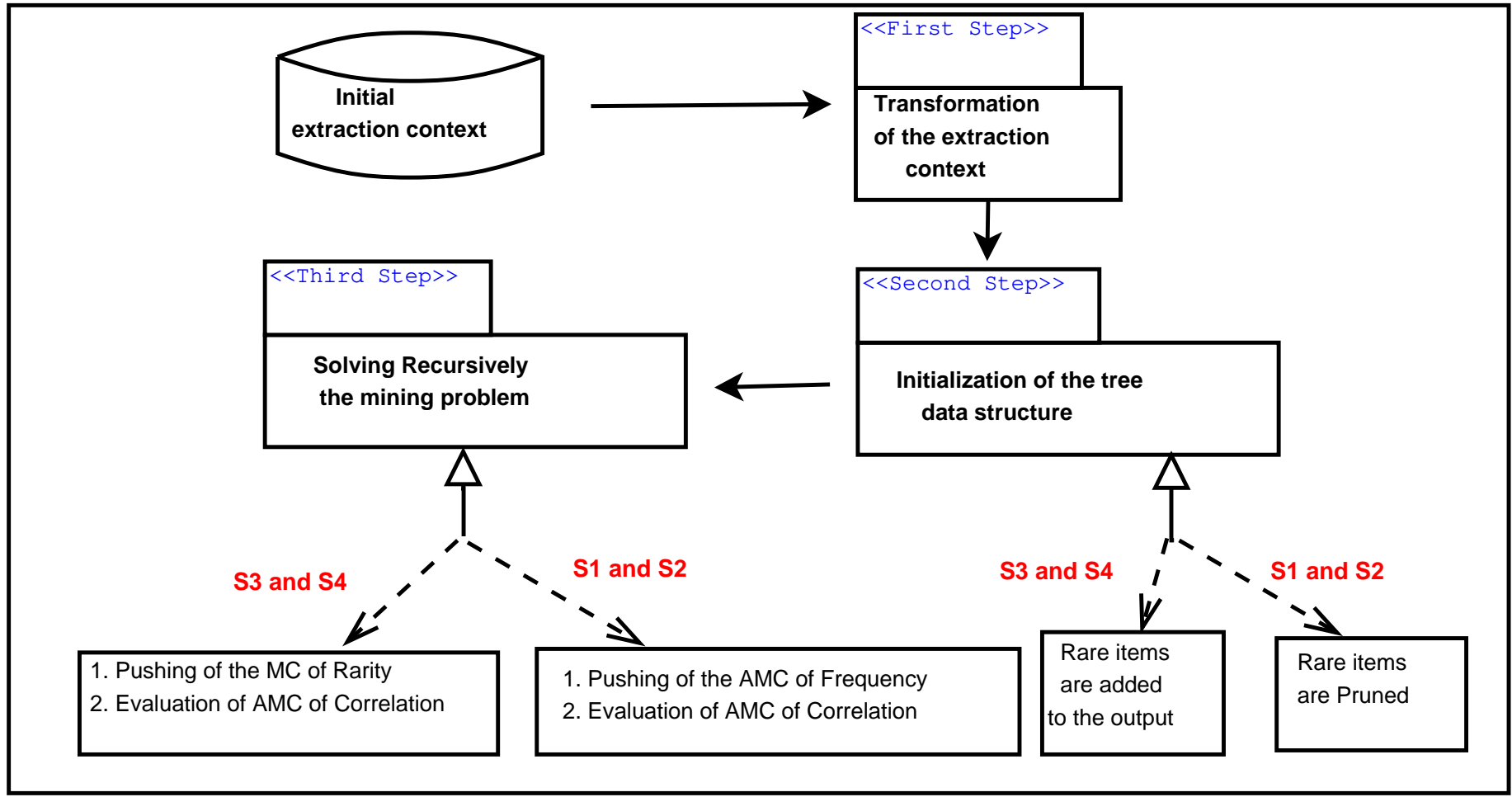}
\caption{An Overview of the \textsc{Opt-Gmjp} algorithm.}\label{figure-optgmjp}
\end{figure}
\begin{figure}
\hspace{-2.0cm}
\includegraphics[scale = 0.55]{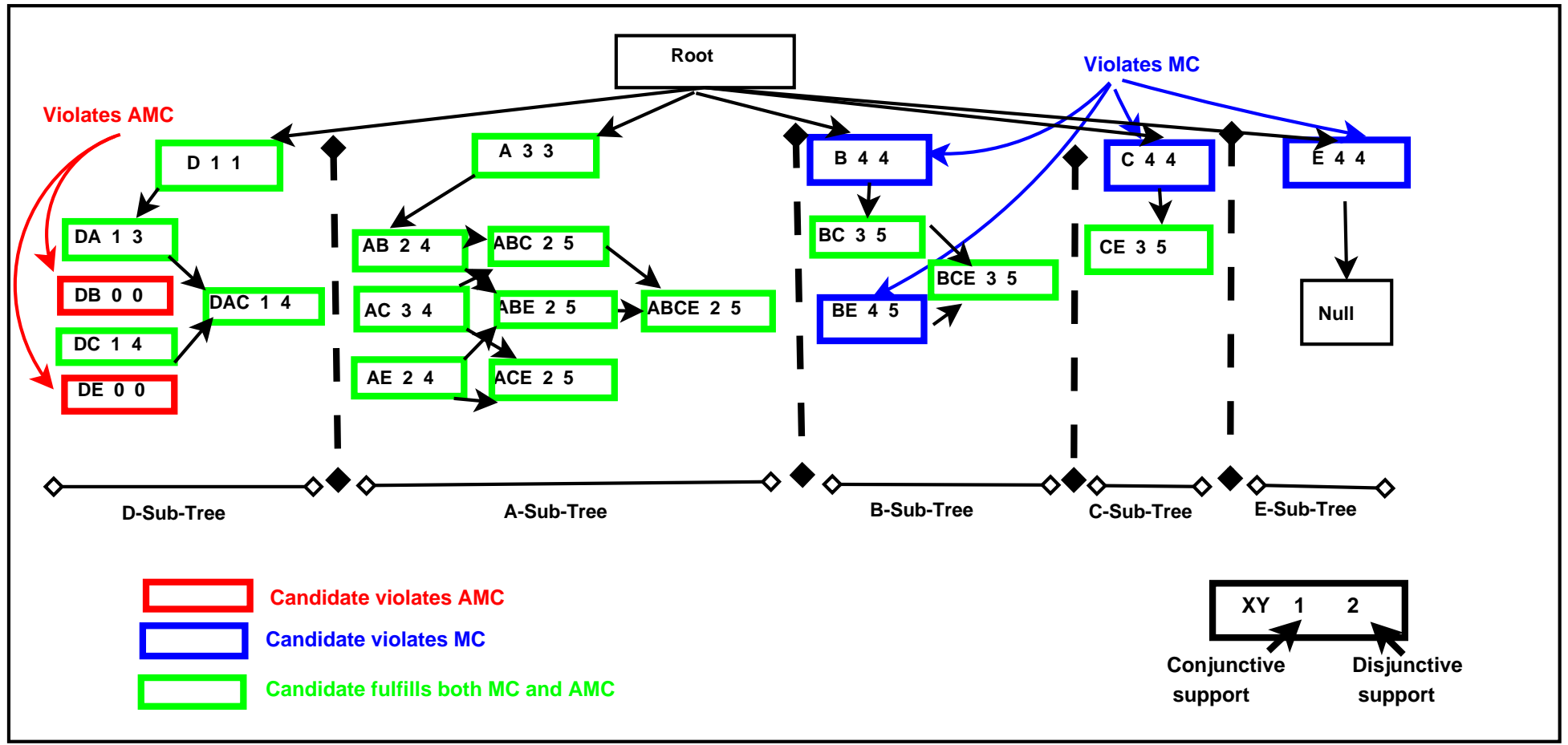}
\caption{The tree data-structure used within the \textsc{Opt-Gmjp} approach and associated to the extraction context given by Table \ref{BDnew}.}
\label{figure-exp-tree}
\end{figure}
\end{landscape}

In the next section, we present our analysis of the theoretical time complexity of the \textsc{Gmjp} algorithm.
\section{Theoretical Time Complexity} \label{se-comp}
Proposition \ref{comp} gives the theoretical time complexity of the \textsc{Gmjp} algorithm when running the fourth scenario dedicated to the extraction of the $\mathcal{RCPR}$ representation.

\begin{proposition} \label{comp}
The worst case time complexity of the first step is bounded by $O\textsc{(}N \times M\textsc{)}$, that of the second step is bounded by $O\textsc{(}\textsc{(}N^3 + \textsc{(}N^2 \times M\textsc{)}\textsc{)}\times 2^N\textsc{)}$, while that of the third step is bounded by  $O\textsc{(}N^2\textsc{)}$, where
$M$ = $|\mathcal{T}|$ and $N$ = $|\mathcal{I}|$. The theoretical complexity in the worst case of the \textsc{Gmjp} algorithm is bounded by the sum of those of its three steps.
\end{proposition}

\begin{proof}
First of all, let us recall the respective roles of the distinct steps of the \textsc{Gmjp} algorithm.
\begin{itemize}
\item \textbf{The first step:} \textit{Scanning the extraction context twice in order to build the bitset vector and the co-occurring vector associated to each
	item $I$.}\\
The complexity $\mathcal{C}_{_{1}}$ of this step, is equal to, $\mathcal{C}_{_{1}}$ =  2 $\times$ $O\textsc{(}N \times M\textsc{)}$ $\approx$ $O\textsc{(}N \times M\textsc{)}$.

\item \textbf{The second step:} \textit{Extracting the local minimal and the local closed rare correlated patterns.}\\
The cost of this step is equal to those of its associated instructions which are as follows:
\begin{enumerate}
	\item The cost of the initialization of the integer $n$ carried out in line \textsc{(}a\textsc{)} \textsc{(}\textit{cf.} Algorithm \ref{AlgoCRPR}\textsc{)} is in $O\textsc{(}1\textsc{)}$.
	\item The generation of the candidates of size $n$, \textsc{(}\textit{cf.} line \textsc{(}b\textsc{)} in Algorithm \ref{AlgoCRPR}\textsc{)}, is done in $O\textsc{(}N-1\textsc{)}$ since in the worst case the
	number of generated candidates is $N-1$.

	\item The cost of the pruning of candidates \textit{w.r.t.} the cross-support property of the bond measure is done in $O\textsc{(}N^2\textsc{)}$ \textsc{(}\textit{cf.} line \textsc{(}i\textsc{)} in Algorithm \ref{AlgoCRPR}\textsc{)}.
	
	\item The cost of the pruning of candidates \textit{w.r.t.} the ideal order property of correlated patterns is done in $O\textsc{(}N^2\textsc{)}$ \textsc{(}\textit{cf.} line \textsc{(}ii\textsc{)} in Algorithm \ref{AlgoCRPR}\textsc{)}.
	
	\item  The cost of the computation of the conjunctive and the disjunctive supports of the itemset candidates is bounded by $O\textsc{(}N \times M\textsc{)}$ \textsc{(}\textit{cf.} line \textsc{(}iii\textsc{)} in Algorithm \ref{AlgoCRPR}\textsc{)}.
	
	\item The checking of the constraints of rarity and of correlation is done in
	$O\textsc{(}1\textsc{)}$, while the checking of the local minimality of the set of candidates is done in $O\textsc{(}N^2\textsc{)}$ and the updating of the
	$\mathcal{MRCP}$ set of minimal rare correlated patterns
	is done in  $O\textsc{(}1\textsc{)}$ \textsc{(}\textit{cf.} line \textsc{(}iv\textsc{)} in Algorithm \ref{AlgoCRPR}\textsc{)}.

	\item The extraction of the local closed rare correlated patterns of size $n-1$ is bounded by $O\textsc{(}N^2\textsc{)}$ \textsc{(}\textit{cf.} line \textsc{(}v\textsc{)} in Algorithm \ref{AlgoCRPR}\textsc{)}.
	
	\item The cost of increasing the integer $n$ is done in $O\textsc{(}1\textsc{)}$ \textsc{(}\textit{cf.} line \textsc{(}vi\textsc{)} in Algorithm \ref{AlgoCRPR}\textsc{)}.
	
	\item There are, in the worst case, $\textsc{(}2^{N} - N - 1\textsc{)}$ candidates to be generated using the \textsc{Apriori-Gen} procedure \textsc{(}\textit{cf.} line \textsc{(}vii\textsc{)} in Algorithm \ref{AlgoCRPR}\textsc{)}. The cost of this step is bounded by $O\textsc{(}2^{N} - N\textsc{)}$.
\end{enumerate}

Consequently, the cost $\mathcal{C}_{_{2}}$ of this second step, is approximatively equal to:

$\mathcal{C}_{_{2}}$ =
$O\textsc{(}1\textsc{)}$ $+$ $O$\textsc{(}$N$ - $1$\textsc{)} $\times$
\textsc{[}
$O$\textsc{(}$2^{N}$ - $N$\textsc{)} $\times$
\textsc{[}
$O\textsc{(}N^2\textsc{)}$ $+$  $O\textsc{(}N^2\textsc{)}$
$+$  $O\textsc{(}N \times M\textsc{)}$ $+$
$O\textsc{(}1\textsc{)}$ $+$ $O\textsc{(}N^2\textsc{)}$ $+$
$O\textsc{(}1\textsc{)}$  $+$ $O\textsc{(}N^2\textsc{)}$ $+$  $O\textsc{(}1\textsc{)}$
\textsc{]}\textsc{]}

$\approx$  $O\textsc{(}1\textsc{)}$ $+$ $O$\textsc{(}$N$ - $1$\textsc{)} $\times$
\textsc{[}
$O$\textsc{(}$2^{N}$\textsc{)} $\times$
\textsc{[}
$O\textsc{(}N^2\textsc{)}$ $+$  $O\textsc{(}N \times M\textsc{)}$
\textsc{]} \textsc{]}

$\approx$ $O$\textsc{(}$2^{N}$\textsc{)} $\times$
\textsc{(}
$O\textsc{(}N^3\textsc{)}$ $+$  $O\textsc{(}N^2 \times M\textsc{)}$\textsc{)}

$\approx$ $O$\textsc{(}\textsc{(}$N^3$ + \textsc{(}$N^2$ $\times$ $M$\textsc{)}\textsc{)} $\times$ $2^{N}$\textsc{)}.

\item \textbf{The third step:} \textit{Filtering the global minimal and the global closed patterns among the local ones}.

In fact, this step consists in checking for each local minimal \textsc{(}\textit{resp.} closed\textsc{)} pattern, whether it has a subset
\textsc{(}\textit{resp.} superset\textsc{)} with the same \textit{bond} value or not. The complexity, $\mathcal{C}_{_{3}}$, of this step is then bounded by
$O\textsc{(}N^2\textsc{)}$ \textsc{(}\textit{cf.} lines $3$ and $4$ in Algorithm \ref{AlgoCRPR}\textsc{)}.
\end{itemize}
Consequently, the complexity of the \textsc{Gmjp} algorithm is equal to:
$\mathcal{C}_{_{1}}$ $+$ $\mathcal{C}_{_{2}}$ $+$ $\mathcal{C}_{_{3}}$

$\approx$ $O{\textsc{(}}N \times M\textsc{)}$ $+$ $O$\textsc{(}\textsc{(}$N^3$ + \textsc{(}$N^2$ $\times$ $M$\textsc{)}\textsc{)} $\times$ $2^{N}$\textsc{)} $+$ $O\textsc{(}N^2\textsc{)}$

$\approx$ $O$\textsc{(}\textsc{(}$N^3$ + \textsc{(}$N^2$ $\times$ $M$\textsc{)}\textsc{)} $\times$ $2^{N}$\textsc{)}
$\approx$ $2^{N}$.
\end{proof}

It is important to mention that the complexity in the worst case of the
\textsc{Gmjp} algorithm is not reachable in practice. Indeed, there is not a context that simultaneously gives the respective worst case complexities of the three steps. Hence, the worst case complexity of \textsc{Gmjp} is roughly bounded by the sum of those of its three steps.

In the next section, we describe the query process of the $\mathcal{RCPR}$
representation and the regeneration of the $\mathcal{RCP}$ set from the $\mathcal{RCPR}$ representation.
\section{The query and the regeneration strategies} \label{sec_Regeneration}
We begin the first sub-section with the querying strategy of the $\mathcal{RCPR}$ representation.
\subsection{Querying of the $\mathcal{RCPR}$ representation} \label{sub_sec_query}
In the following, we introduce the \textsc{Regenerate} algorithm, whose pseudo code is given by Algorithm \ref{Rege1}, dedicated to the query of the $\mathcal{RCPR}$ representation. In fact, the interrogation of the $\mathcal{RCPR}$ representation allows determining the nature of a given pattern. If it is a rare correlated pattern, then, its conjunctive, disjunctive, negative supports as well as its \textit{bond} value are faithfully derived from the $\mathcal{RCPR}$ representation.

\bigskip
$\bullet$ \textbf{Description of the \textsc{Regenerate} algorithm}

\bigskip
The \textsc{Regenerate} algorithm takes as an input the number of the transactions $|\mathcal{T}|$, the $\mathcal{RCPR}$ representation and an arbitrary itemset $I$ and it proceeds in two distinct ways depending on the state of $I$: 
\begin{itemize}
\item If the itemset $I$ belong to the $\mathcal{RCPR}$ representation 
\textsc{(}\textit{cf.} Line 2 Algorithm \ref{Rege1}\textsc{)}, then $I$ is a rare correlated itemset. In this regard, we have the conjunctive support and the \textit{bond} values, thus we compute the disjunctive and the negative supports. The disjunctive support is equal to the ratio of the conjunctive support by the \textit{bond} value, \textsc{(}\textit{cf.} Line 3 Algorithm \ref{Rege1}\textsc{)}. The negative support is equal to  the number of transactions $|\mathcal{T}|$ minus the disjunctive support \textsc{(}\textit{cf.} Line 4
Algorithm \ref{Rege1}\textsc{)}. The algorithm returns the values of the different supports as well as the \textit{bond} value of the 
rare correlated itemset $I$, \textsc{(}\textit{cf.} Line 5 Algorithm \ref{Rege1}\textsc{)}.
\item If the itemset $I$ do not belong to the $\mathcal{RCPR}$ representation, then we distinguish two different cases:
\begin{itemize}
\item  If it exists two itemsets $J$ and $Z$ belonging to $\mathcal{RCPR}$ such as, $J$ $\subset$ $I$ and $Z$ $\supset$ $I$, \textsc{(}\textit{cf.} Line 7 Algorithm \ref{Rege1}\textsc{)}, then $I$ is a rare correlated pattern and we have to determine its supports and correlation values. The closed itemset $F$ associated to $I$ is the minimal itemset, according to the inclusion set, that covers $I$;  \textsc{(}\textit{cf.} Line 9 Algorithm \ref{Rege1}\textsc{)}. The conjunctive support and the \textit{bond} value of $I$ are equal to those of its closure $F$,  \textsc{(}\textit{cf.} Lines 10 and 11 Algorithm \ref{Rege1}\textsc{)}. The disjunctive and the supports
are computed in the same manner as the first case, \textsc{(}\textit{cf.} Lines 12 and 13 Algorithm \ref{Rege1}\textsc{)}. Finally the algorithm outputs the values of the different supports as well as the \textit{bond} value of the  rare correlated itemset $I$, \textsc{(}\textit{cf.} Line 14 Algorithm \ref{Rege1}\textsc{)}.

\item  If it not exists two itemsets $J$ and $Z$ belonging to $\mathcal{RCPR}$ such as, $J$ $\subset$ $I$ and $Z$ $\supset$ $I$, then 
the algorithm outputs the emptyset to indicates that the itemset $I$ is not a rare correlated pattern, \textsc{(}\textit{cf.} Line 16 Algorithm \ref{Rege1}\textsc{)}.
\end{itemize}
\end{itemize}
\incmargin{0.5em} 
\linesnumbered
\restylealgo{algoruled}
\begin{algorithm}[!t]
\SetVline \caption{\textsc{Regenerate} \label{Rege1}}
\setnlskip{-3pt}
\KwData{\begin{enumerate}
	\item An arbitrary pattern $I$.
	\item $\mid$$\mathcal{T}$$\mid$: The number of transactions.
	\item The $\mathcal{RCPR}$ representation = $\mathcal{MRCP}$ $\cup$ $\mathcal{CRCP}$.
\end{enumerate}}
\KwResult{The conjunctive, disjunctive, negative supports and the \textit{bond} value of the pattern $I$ if it is rare correlated, the empty set otherwise.}
\Begin{
\If{\textsc{(}$I$ $\in$ $\mathcal{RCPR}$\textsc{)} \label{Apparrmcr}}
{
	\textit{Supp}\textsc{(}$\vee$$I$\textsc{)} = $\displaystyle\frac{\displaystyle Supp\textsc{(}\wedge I\textsc{)}}
	{\displaystyle \textit{bond}\textsc{(}I\textsc{)}}$ \label{calculdisj}\;
	
	\textit{Supp}\textsc{(}$\neg$$I$\textsc{)} = $\mid$$\mathcal{T}$$\mid$ $-$ \textit{Supp}\textsc{(}$\vee$$I$\textsc{)} \label{calculneg}\;
	
	\Return{ $\{$$I$, \textit{Supp}\textsc{(}$\wedge$$I$\textsc{)},
		\textit{Supp}\textsc{(}$\vee$$I$\textsc{)},
		\textit{Supp}\textsc{(}$\neg$$I$\textsc{)},
		\textit{bond}\textsc{(}$I$\textsc{)}$\}$ \label{res1}}\;
}
\Else
{
	\If{\textsc{(}$\exists$ $J$, $Z$ $\in$ $\mathcal{RCPR}$ $\mid$ $J$ $\subset$ $I$ and $I$ $\subset$ $Z$\textsc{)}\label{si1}}
	{
		{$/*$ $F$ is the closed pattern associated to the candidate $I$ $*/$}	
		
		$F$:= $\min_{\subseteq}$$\{$$I_{1}$ $\in$  $\mathcal{RCPR}$ $\mid$  $I$  $\subset$ $I_{1}$$\}$ \label{ferme}\;
		
		\textit{Supp}\textsc{(}$\wedge$$I$\textsc{)} = \textit{Supp}\textsc{(}$\wedge$$F$\textsc{)}\;
		
		\textit{bond}\textsc{(}$I$\textsc{)} = \textit{bond}\textsc{(}$F$\textsc{)}\;
		
		\textit{Supp}\textsc{(}$\vee$$I$\textsc{)} = $\displaystyle\frac{\displaystyle Supp\textsc{(}\wedge I\textsc{)}}
		{\displaystyle \textit{bond}\textsc{(}I\textsc{)}}$\;
		
		\textit{Supp}\textsc{(}$\neg$$I$\textsc{)} = $\mid$$\mathcal{T}$$\mid$ $-$  \textit{Supp}\textsc{(}$\vee$$I$\textsc{)}\;
		
		\Return{ $\{$$I$,  \textit{Supp}\textsc{(}$\wedge$$I$\textsc{)},
			\textit{Supp}\textsc{(}$\vee$$I$\textsc{)},
			\textit{Supp}\textsc{(}$\neg$$I$\textsc{)},
			\textit{bond}\textsc{(}$I$\textsc{)}$\}$ \label{res2}}\;
	}
	\Else
	{\Return{$\emptyset$}\label{res3}\;}
}}

\end{algorithm}
\decmargin{1em}
\begin{example}\label{exp_regeneration1}
Let us consider the $\mathcal{RCPR}$ representation given by Example \ref{exprep1}, 
\textsc{(}Page \pageref{exprep1}\textsc{)}, for \textit{minsupp} = 4  and \textit{minbond} = 0.2.
Consider the pattern $ACE$.
When comparing the pattern $ACE$ with the elements of the $\mathcal{RCPR}$ representation, we remark that
$AE$ $\subset$ $ACE$ and $ACE$ $\subset$ $ABCE$.
Then, the pattern $ACE$ is a rare correlated pattern and the associated closed pattern is $ABCE$.
Consequently,
\textit{Supp}\textsc{(}$\wedge$$ACE$\textsc{)} = \textit{Supp}\textsc{(}$\wedge$$ABCE$\textsc{)} = 2,
\textit{Supp}\textsc{(}$\vee$$ACE$\textsc{)} = \textit{Supp}\textsc{(}$\vee$$ABCE$\textsc{)} = 5,
\textit{Supp}\textsc{(}$\neg$$ACE$\textsc{)} = $|\mathcal{T}|$ - \textit{Supp}\textsc{(}$\vee$$ACE$\textsc{)} = 5 - 5 = 0 and
\textit{bond}\textsc{(}$ACE$\textsc{)} = \textit{bond}\textsc{(}$ABCE$\textsc{)} = $\displaystyle\frac{2}{5}$. Consider the pattern $BE$. In fact, $BE$ $\notin$ $\mathcal{RCPR}$
and there is no element of $\mathcal{RCPR}$ which is included in $BE$.
Therefore, the \textsc{Regenerate} algorithm returns the empty set and indicates that the pattern $BE$
is not a rare correlated pattern.
\end{example}
In what follows, we introduce the strategy of regeneration of the $\mathcal{RCP}$
set, \textit{i.e.} the set of all rare correlated patterns, from this representation.

\subsection{Regeneration of the $\mathcal{RCP}$ set from the $\mathcal{RCPR}$ representation} \label{sub_sec_regeneration}
The regeneration of the $\mathcal{RCP}$ set from the $\mathcal{RCPR}$ representation is achieved through the \textsc{RcpRegeneration} algorithm which pseudo-code is given by Algorithm \ref{CRPRegeneration}. This latter algorithm inputs the $\mathcal{RCPR}$ representation and provides the $\mathcal{RCP}$ set of rare correlated patterns. The conjunctive support and the \textit{bond} value of each pattern are exactly determined.

\bigskip
$\bullet$ \textbf{Description of the \textsc{RcpRegeneration} algorithm}

\bigskip
The \textsc{RcpRegeneration} algorithm takes as an input the number of the transactions $|\mathcal{T}|$ and the $\mathcal{RCPR}$ representation which is composed by the two sets $\mathcal{MRCP}$ and  $\mathcal{CRCP}$. The algorithm generates the $\mathcal{RCP}$ set as described in the following: 
\begin{enumerate}
\item Initially, the $\mathcal{RCP}$ set is assigned with the empty set, \textsc{(}\textit{cf.} Line 2 Algorithm \ref{CRPRegeneration}\textsc{)}. Then, all the itemsets of the  $\mathcal{RCPR}$ representation \textsc{(}The elements of the $\mathcal{MRCP}$ and the $\mathcal{CRCP}$ sets\textsc{)} are added to the $\mathcal{RCP}$ set, \textsc{(}\textit{cf.} Lines 4 and 5 Algorithm \ref{CRPRegeneration}\textsc{)}. 
\item For each minimal generator $M$ of the $\mathcal{MRCP}$ set, we determine its closure $F$ among the $\mathcal{CRCP}$ set. In fact, $F$ corresponds to the minimal itemset, according to inclusion set, that covers $M$ \textsc{(}\textit{cf.} Line 9 Algorithm \ref{CRPRegeneration}\textsc{)}.
\item At this step, we derive all the patterns that are included between the minimal generator $M$ and its closure $F$. 
Thus, we need an intermediate itemset $D$ that contains the set of items belonging to $F$ and not to $M$: 	
$D$ =  $F$$\backslash$$M$, \textsc{(}\textit{cf.} Line 11 Algorithm \ref{CRPRegeneration}\textsc{)}. Then, each
item $j$ included in the itemset $D$, is concatenated with the generator $M$ in order to form a new itemset $X$: $X$ = $M$ $\cup$
$j$, \textsc{(}\textit{cf.} Line 13 Algorithm \ref{CRPRegeneration}\textsc{)}. The conjunctive support and the \textit{bond} value of $X$ are equal to those of its closure $F$ \textsc{(}\textit{cf.} Lines 14 and 15 Algorithm \ref{CRPRegeneration}\textsc{)}. 
\item Finally, the rare correlated itemset $X$ is added to the $\mathcal{RCP}$ set after the non-redundancy checking  
\textsc{(}\textit{cf.} Line 18 Algorithm \ref{CRPRegeneration}\textsc{)}. The whole $\mathcal{RCP}$ set is outputted in Line 19 of algorithm \ref{CRPRegeneration}.
\end{enumerate}
\incmargin{1em}
\restylealgo{algoruled}
\linesnumbered
\begin{algorithm}[htbp]\label{CRPRegeneration}\small{
\SetVline
\setnlskip{-3pt}
\KwData{\begin{enumerate}
		\item  The $\mathcal{RCPR}$ representation =  $\mathcal{MRCP}$ $\cup$ $\mathcal{CRCP}$.
		\item  The number of transactions: $|$$\mathcal{T}$$|$.
	\end{enumerate}
}
\KwResult{The $\mathcal{RCP}$ set of rare correlated patterns.}
\Begin{
	$\mathcal{RCP}$:= $\emptyset$\;
	{$/*$ The elements of $\mathcal{RCPR}$ are appended to the $\mathcal{RCP}$ set $*/$}
	
	\For{\textsc{(}$M \in \mathcal{RCPR}$\textsc{)}}
	{
		$\mathcal{RCP}$:= $\mathcal{RCP}$ $\cup$ $\{$$M$, \textit{Supp}\textsc{(}$\wedge$$M$\textsc{)},
		\textit{bond}\textsc{(}$M$\textsc{)}$\}$ \label{majMCR1}\;
	}
	{$/*$ Process of each element $M$ of $\mathcal{MRCP}$ separately $*/$}
	
	\For{\textsc{(}$M \in \mathcal{MRCP}$\textsc{)}\label{pour2}}
	{
		{$/*$ Identification of the closed pattern associated to the minimal $M$ $*/$}
		
		$F$:= 	$\min_{\subseteq}$$\{$$M_{1}$ $\in$  $\mathcal{CRCP}$ $\mid$  $M$  $\subset$ $M_{1}$$\}$\label{ferme2}\;	
		
		{$/*$ Derivation of the rare correlated patterns included between $M$ and $F$ $*/$}
		
		$D$:=  $F$ $\backslash$ $\{{i}$,  $\forall$ ${i}$ $\in$ $M$$\}$
		
		\For{\textsc{(}$j$ $|$ $j$ $\in$ $D$\textsc{)}\label{pour-New}}
		{
			\For{\textsc{(}$X$ $\mid$ $X$ = $M$ $\cup$ ${j}$ \textsc{)}\label{pour3}}
			{

				\textit{Supp}\textsc{(}$\wedge$$X$\textsc{)} = \textit{Supp}\textsc{(}$\wedge$$F$\textsc{)}\;  		          		
				\textit{bond}\textsc{(}$X$\textsc{)} = \textit{bond}\textsc{(}$F$\textsc{)} \;
				
				{$/*$ Check of the uniqueness of the elements of the $\mathcal{RCP}$ set $*/$}
				
				\If{\textsc{(}$X$ $\notin$ $\mathcal{RCP}$\textsc{)}}
				{
					$\mathcal{RCP}$:= $\mathcal{RCP}$ $\cup$ $\{$$X$, \textit{Supp}\textsc{(}$\wedge$$X$\textsc{)},
					\textit{bond}\textsc{(}$X$\textsc{)}$\}$ \label{majMCR2}\;
				}
			}
		}
	}
	\Return{$\mathcal{RCP}$}\;\label{resultat}}
\caption{\textsc{RcpRegeneration}}}
\end{algorithm}
\decmargin{1em}

\begin{example}
Consider the $\mathcal{RCPR}$ representation illustrated by example \ref{exprep1} 
, \textsc{(}Page \pageref{exprep1}\textsc{)}, for \textit{minsupp} = 4 and \textit{minbond} = 0.2.
The regeneration of the $\mathcal{RCP}$
set is carried out as follows.
Firstly, the $\mathcal{RCP}$ set is initialized to the empty set. Then, the elements of the
$\mathcal{RCPR}$ representation are appended to the $\mathcal{RCP}$ set.
We have so,
$\mathcal{RCP}$ = $\{$\textsc{(}$A$, $3$, $\displaystyle\frac{3}{3}$\textsc{)},
\textsc{(}$D$, $1$, $\displaystyle\frac{1}{1}$\textsc{)},
\textsc{(}$AB$, $2$, $\displaystyle\frac{2}{5}$\textsc{)},
\textsc{(}$AC$, $3$, $\displaystyle\frac{3}{4}$\textsc{)},
\textsc{(}$AD$, $1$, $\displaystyle\frac{1}{3}$\textsc{)},
\textsc{(}$AE$, $2$, $\displaystyle\frac{2}{5}$\textsc{)},
\textsc{(}$BC$, $3$, $\displaystyle\frac{3}{5}$\textsc{)},
\textsc{(}$CD$, $1$, $\displaystyle\frac{1}{4}$\textsc{)},
\textsc{(}$CE$, $3$, $\displaystyle\frac{3}{5}$\textsc{)},
\textsc{(}$ACD$, $1$, $\displaystyle\frac{1}{4}$\textsc{)},
\textsc{(}$BCE$, $3$, $\displaystyle\frac{3}{5}$\textsc{)},
\textsc{(}$ABCE$, $2$, $\displaystyle\frac{2}{5}$\textsc{)}$\}$.
Then, the patterns $ABE$ and $ABC$ included between the minimal pattern \textsc{(}$AB$, 2, $\displaystyle\frac{2}{5}$\textsc{)} and its closure \textsc{(}$ABCE$, 2, $\displaystyle\frac{2}{5}$\textsc{)} are generated.
In addition, the pattern
$ACE$, included between the minimal pattern \textsc{(}$AE$, 2, $\displaystyle\frac{2}{5}$\textsc{)} and its closure \textsc{(}$ABCE$, 2, $\displaystyle\frac{2}{5}$\textsc{)}, is also derived.
The patterns $ABE$, $ABC$ and $ACE$ are also generated, as they share the same conjunctive support and
\textit{bond} value of the closed pattern $ABCE$,  are then inserted in the
$\mathcal{RCP}$ set. This latter is updated and contains all the rare correlated patterns:
$\mathcal{RCP}$ = $\{$\textsc{(}$A$, $3$, $\displaystyle\frac{3}{3}$\textsc{)},
\textsc{(}$D$, 1, $\displaystyle\frac{1}{1}$\textsc{)},
\textsc{(}$AB$, $2$, $\displaystyle\frac{2}{5}$\textsc{)},
\textsc{(}$AC$, $3$, $\displaystyle\frac{3}{4}$\textsc{)},
\textsc{(}$AD$, 1, $\displaystyle\frac{1}{3}$\textsc{)},
\textsc{(}$AE$, $2$, $\displaystyle\frac{2}{5}$\textsc{)},
\textsc{(}$BC$, $3$, $\displaystyle\frac{3}{5}$\textsc{)},
\textsc{(}$CD$, 1, $\displaystyle\frac{1}{4}$\textsc{)},
\textsc{(}$CE$, $3$, $\displaystyle\frac{3}{5}$\textsc{)},
\textsc{(}$ABC$, 2, $\displaystyle\frac{2}{5}$\textsc{)},
\textsc{(}$ABE$, 2, $\displaystyle\frac{2}{5}$\textsc{)},
\textsc{(}$ACD$, 1, $\displaystyle\frac{1}{4}$\textsc{)},
\textsc{(}$ACE$, 2, $\displaystyle\frac{2}{5}$\textsc{)},
\textsc{(}$ABCE$, 2, $\displaystyle\frac{2}{5}$\textsc{)}$\}$.
\end{example}
\section{Conclusion} \label{ConcChap5}
We introduced, in this chapter, the \textsc{Gmjp} extraction approach to mine correlated patterns in a generic way \textsc{(}i.e., with two types of constraints: anti-monotonic constraint of frequency and monotonic constraint of rarity\textsc{)}.
Our approach is based on the key notion of bitsets codification that supports efficient correlated patterns computation thanks to an adequate condensed representation of patterns.
In the next chapter, we present our experimental evaluation of \textsc{Gmjp} and of \textsc{Opt-Gmjp} according to both quantitative and qualitative aspects.

\part{Experiments and Classification Process}
\chapter{Experimental Validation}\label{ch_6}
\section{Introduction} \label{IntroChap6}

This chapter is devoted to the experimental evaluation of the proposed \textsc{Gmjp} algorithm.
Our evaluation is performed on two principal axes.
In Section \ref{se1Chap6}, we present the experimental environment, specifically the characteristics of the used datasets as well as the experimental protocol.  Then, we present in Section \ref{sec_XP1} the qualitative evaluation of the proposed condensed representations which is measured by the compactness rates  offered by each proposed exact and approximate concise representation.  Section \ref{sec_XP2}, is dedicated to the evaluation of the performance of both \textsc{Gmjp}, as well as the   optimized version that highlights the important measured improvements.

\section{Experimental Environment} \label{se1Chap6}

In this chapter, we aim to show, through extensive carried out experiments,
that the different proposed concise representations provide interesting compactness rates compared
to the whole set of  correlated patterns. In addition to this, we aim to prove the efficiency of the proposed \textsc{Gmjp} approach. In our experiments, we used two evaluation measures: the conjunctive support to measure the frequency \textsc{(}respectively the rarity\textsc{)} and the \textit{bond} measure to evaluate the correlation of a pattern.

All experiments were carried out on a PC equipped with a $2.4$ GHz Intel Core TM i3 processor and $4$ GB of main memory, running the Linux Ubuntu 12.04. The used datasets are described in what follows.
\subsection{Datasets}

The experiments were carried out on different dense and sparse benchmark datasets
$^{\textsc{(}}$\footnote{Available at \textsl{http://fimi.cs.helsinki.fi/data}.}$^{\textsc{)}}$.
Table \ref{Caract_bases_benchmark} summarizes the characteristics of the considered datasets. A brief description of the content of each dataset is given below:

\begin{itemize}
	\item \textbf{\textsc{Connect}}: This dataset contains all legal positions in the game of connect-4 in which neither player has won yet, and in which the next move is not forced.
	
	\item \textbf{\textsc{Mushroom}}: This dataset includes descriptions of hypothetical samples corresponding to 23 species of gilled mushrooms.
	
	\item \textbf{\textsc{Pumsb}}: This dataset contains Census Data from PUMS \textsc{(}Public Use Microdata Samples\textsc{)}. Each object represents the answers to a census questionnaire.
	
	\item \textbf{\textsc{Pumsb*}}: This dataset is obtained after deleting all frequent items for a minimum support threshold set to 80$\%$ in the original \textsc{Pumsb}.
	
	\item \textbf{\textsc{Retail}}: The \textsc{Retail} dataset contains information about Market Basket of clients in a Belgian Supermarket.
	
	\item \textbf{\textsc{Accidents}}: This dataset represents traffic accidents obtained from the National Institute of Statistics \textsc{(}NIS\textsc{)} for the Region of Flanders \textsc{(}Belgium\textsc{)}.
	
	\item \textbf{\textsc{T10I4D100K}}: This is a synthetic dataset generated using the generator from the IBM Almaden Quest Research Group. The goal of this generator is to create objects similar to those obtained in a supermarket environment.
	
	\item \textbf{\textsc{T40I10D100K}}: Identically to \textsc{T10I4D100K}, this dataset is also generated from the IBM generator. The differences between this dataset and  \textsc{T10I4D100K} are the number of items and the average size of the objects.
\end{itemize}
\begin{table}
	\begin{center}
		\begin{tabular}{lrrrr}
			\hline
			\noalign{\smallskip}
			\textbf{Dataset}&\textbf{Property}  &\textbf{Number} & \textbf{Number of}& \textbf{Average length}\\
			&                   &\textbf{of items} &\textbf{transactions}& \textbf{of transactions}\\
			\noalign{\smallskip}\hline\noalign{\smallskip}
			\textsc{Connect} &  \textbf{Dense} &{129} & {67, 557} & {43.00}\\
			\textsc{Chess}   &  \textbf{Dense}  &{75} & {3 196} & {37.00}\\
			\textsc{Mushroom} &  \textbf{Dense} &{119} & {8, 124} & {23.00}\\
			\textsc{Pumsb}&   \textbf{Dense} &{7, 117} & {49, 046} & {74.00}\\
			\textsc{Pumsb*}&   \textbf{Dense} & {7, 117} &{49, 046} & {50.00}\\
			\textsc{Retail} &   \textbf{Sparse} &{16, 470} &{88, 162} & {10.00}\\
			\textsc{Accidents}&   \textbf{Sparse}& {468} &{340, 183} & {33.81}\\
			\textsc{T10I4D100K}&   \textbf{Sparse} & {870} & {100, 000} &{10.10} \\
			\textsc{T40I10D100K}&   \textbf{Sparse}& {942} &{100, 000} &{39.61}\\
			\noalign{\smallskip}\hline
		\end{tabular}
	\end{center}
	\caption{Characteristics of the benchmark datasets.}
	\label{Caract_bases_benchmark}
\end{table}

\subsection{Experimental Protocol}
Our objective is to prove, through extensive carried out experiments, the efficiency of the proposed \textsc{Gmjp} algorithm.

\checkmark  Our first batch of experiments focus on evaluating the compactness rates by different condensed representations of rare correlated patterns. We also build a quantitative comparison between the $\mathcal{FCP}$, the $\mathcal{RCP}$ sets and their associated condensed representations.

\checkmark  Our second batch of experiments focus on studying running times of the proposed \textsc{Gmjp} algorithm while running the four different scenarios.

\section{Evaluation of the compactness rates offered by the proposed representations} \label{sec_XP1}
The compactness rate offered by a concise representation measures the reduction of the size of the representation compared to the size of the whole set of patterns. For example, for the $\mathcal{RCPR}$
representation, the compactness rate is equal to:  1 - $\frac{|\mathcal{RCPR}|}{|\mathcal{RCP}|}$. Worth of cite, our experimental study confirms that the $\mathcal{RCPR}$ representation is a perfect cover of the $\mathcal{RCP}$ set. In fact, the obtained results show that the size of $\mathcal{RCPR}$ is always
smaller than that of the $\mathcal{RCP}$ set over the entire range of the considered support and \textit{bond} thresholds. Our study concerns both dense and sparse datasets.

\subsection{Effect of \textit{minsupp} variation}

For example, considering the \textsc{Mushroom}
dataset for \textit{minsupp} = \textit{35}$\%$ and
\textit{minbond} = \textit{0.15}: $|\mathcal{RCPR}|$ = \textit{1, 810} while $|\mathcal{RCP}|$ = \textit{100, 156},
with a reduction rate reaching approximately \textit{98}$\%$.
This is explained by the nature of the induced equivalence classes. In fact, we have in this case,
$|\mathcal{MRCP}|$ = \textit{1, 412} and  $|\mathcal{CRCP}|$ = \textit{652}.
Since the $\mathcal{RCPR}$ representation corresponds to the union without redundancy of the
$\mathcal{MRCP}$ and the $\mathcal{CRCP}$ sets, we have always $|\mathcal{RCPR}|$ $\leq$ $|\mathcal{MRCP}|$ + $|\mathcal{CRCP}|$.

In this respect, Figure \ref{Fct_minsupp} shows that all the compression rates proportionally vary to \textit{minsupp} and disproportionately with respect to \textit{minbond} values. This is due to the fact that, the size of the different representations increases as far as \textit{minsupp} increases and decreases whenever \textit{minbond} increases.

\begin{figure}[h]
	\begin{center}
		\includegraphics[scale = 1.3]{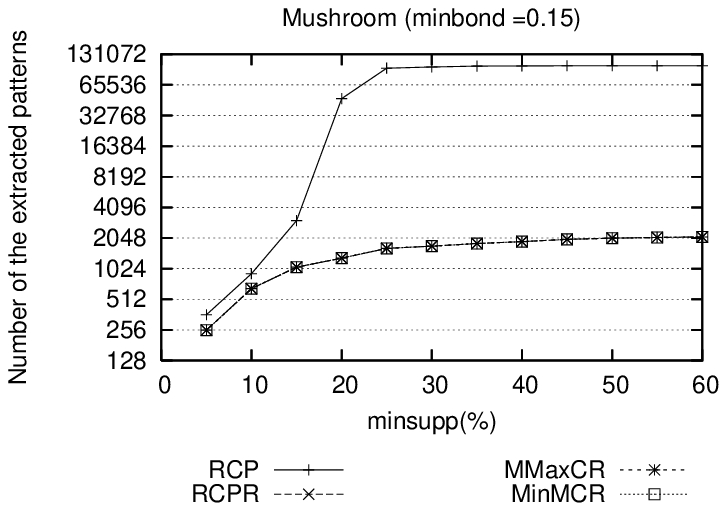}
		\includegraphics[scale = 1.3]{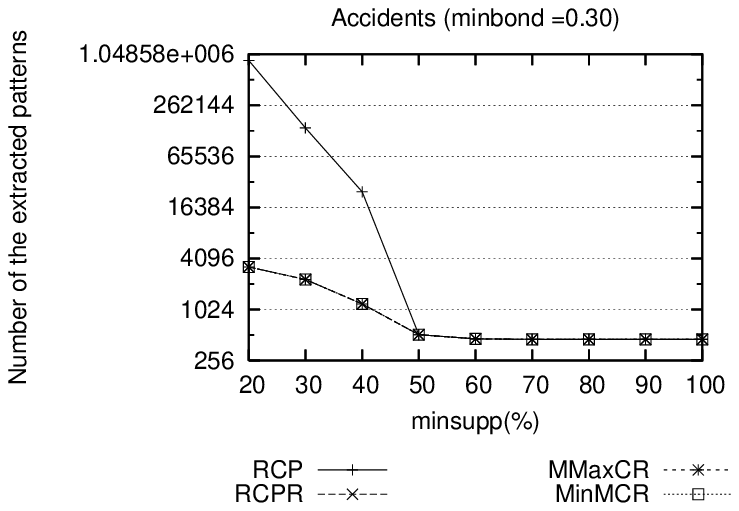}
		\includegraphics[scale = 1.3]{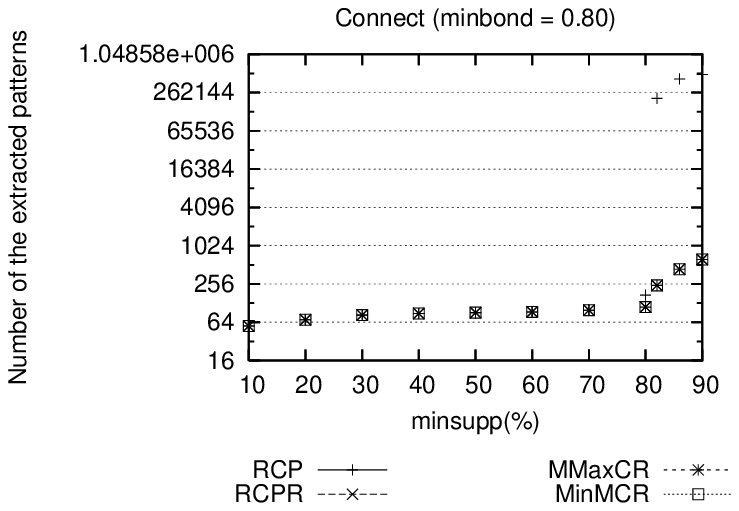}
	\end{center}
	\caption{Sizes of the different representations when \textit{minsupp} varies and \textit{minbond} is fixed.}\label{Fct_minsupp}
\end{figure}

We also find that the respective sizes of the concise exact representations
$\mathcal{MM}$$ax$$\mathcal{CR}$ and $\mathcal{M}$$in$$\mathcal{MCR}$
never exceed the size of the $\mathcal{RCPR}$ representation.
This is justified by the nature of the elements composing both representations.
In fact, the $\mathcal{MM}$$ax$$\mathcal{CR}$ is composed by the
$\mathcal{MRCP}$ set of minimal rare correlated patterns and the
$\mathcal{M}ax$$\mathcal{CRCP}$ set of maximal closed rare correlated patterns.
Nevertheless, we confirm that the size of the $\mathcal{M}ax$$\mathcal{CRCP}$ set
is always lower than that of the $\mathcal{CRCP}$ set.
According to the $\mathcal{M}$$in$$\mathcal{MCR}$ representation,
it is based on the $\mathcal{M}in$$\mathcal{MRCP}$ set of the minimal elements of the $\mathcal{MRCP}$ set and on the $\mathcal{CRCP}$ set.
In fact, we remark that the size of the $\mathcal{M}in$$\mathcal{MRCP}$ set never exceeds the size of the $\mathcal{MRCP}$ set.

In what follows, we evaluate the compactness of the representations based on the variation of
the correlation threshold \textit{minbond}.

\subsection{Effect of \textit{minbond} variation}

Let us consider the results depicted by Figure \ref{Fct_minbond}. The first intuition is that
all the compression rates vary disproportionately to \textit{minbond} values.
For the \textsc{Pumsb*} dataset, for  \textit{minsupp} = \textit{80}$\%$ and for \textit{minbond} =
\textit{0.30}, we have
$|\mathcal{RCP}|$ = \textit{65, 536} $>$
$|\mathcal{RCPR}|$  $=$  $|$$\mathcal{RMM}$$ax$$\mathcal{F}|$  $=$
$|\mathcal{M}in$$\mathcal{MRCP}|$ = \textit{2, 048}.
Now, while increasing \textit{minbond} from \textit{0.30} to \textit{0.60}, we remark that the size of the
$\mathcal{RCP}$ increase as well the $|\mathcal{RCPR}|$ representation, whereas
the size of $\mathcal{RMM}$$ax$$\mathcal{F}$ and of  $\mathcal{M}in$$\mathcal{MRCP}$ are unchanged. Thus,
these representations offer constant reduction rates in spite of \textit{minbond} variation.
For the same example, we have  $|\mathcal{RCP}|$ = \textit{130, 000} $>$
$|\mathcal{RCPR}|$ = \textit{4, 096} $>$
$|$$\mathcal{RMM}$$ax$$\mathcal{F}|$  = \textit{2, 048} $=$
$|\mathcal{M}in$$\mathcal{MRCP}|$ = \textit{2, 048}.

\begin{figure}
	\begin{center}
		\includegraphics[scale = 1.3]{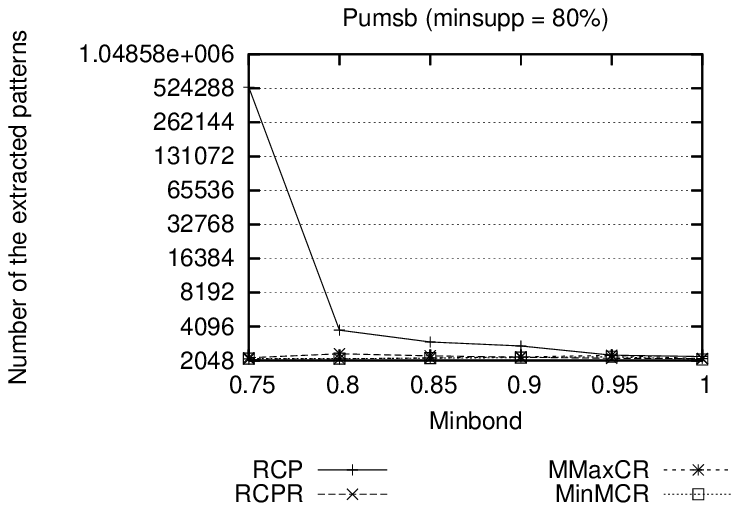}
		\includegraphics[scale = 1.3]{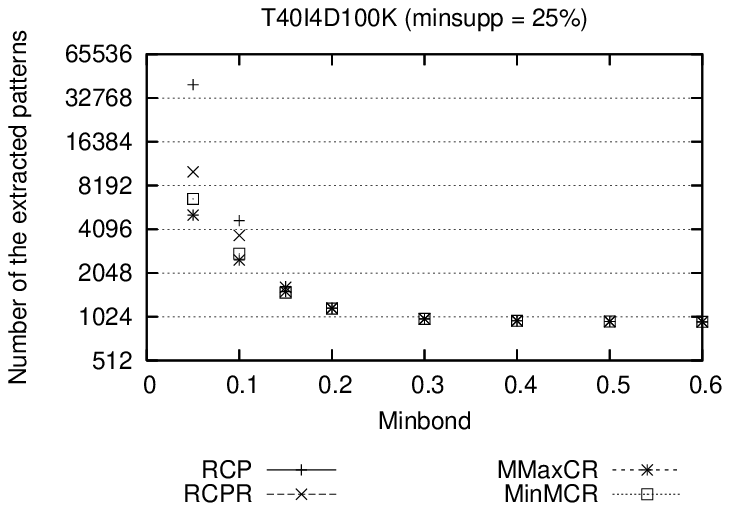}
		\includegraphics[scale = 1.3]{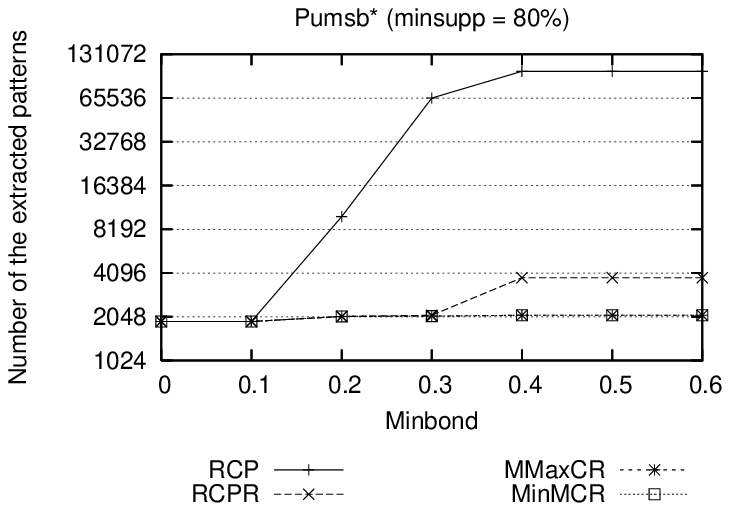}
	\end{center}
	\caption{Sizes of the different representations  when \textit{minbond} varies
		and \textit{minsupp} is fixed.}\label{Fct_minbond}
\end{figure}

We sketch in Table \ref{labelsize1}, the experimental results associated to both the $\mathcal{FCP}$ set of frequent correlated patterns and to the $\mathcal{RCP}$ set of rare correlated patterns.
To summarize, the concise representations $\mathcal{FCCPR}$ and $\mathcal{RCPR}$ present very encouraging reduction rates over several datasets and for different ranges of \textit{minsupp} and \textit{minbond} thresholds. We note that, the `gain` corresponds to the ``reduction rate`` said also ``compactness rate``.
\begin{center}
	\begin{sidewaystable}
		\begin{tabular}{lclrrrcrrcr}
			\hline
			\noalign{\smallskip}
			\textbf{Dataset}& \textbf{\textit{minsupp}} & \textbf{\textit{minbond}}& \textbf{$\#$ $\mathcal{FCP}$} & \textbf{$\#$ $\mathcal{FCCPR}$} & \textbf{Gain of}           &   &\textbf{$\#$ $\mathcal{RCP}$} & \textbf{$\#$ $\mathcal{RCPR}$} & &\textbf{Gain of}\\
			&                           &                          &                              &                                  & \textbf{$\mathcal{FCCPR}$}  &  &\textbf{}                     &                                &  &\textbf{$\mathcal{RCPR}$} \\
			\noalign{\smallskip}\hline\noalign{\smallskip}
			\textsc{Mushroom} &{30$\%$}     & {0.15} & {2, 701} & {427}& {\textbf{84.19$\%$}}&  &{98, 566}& {1, 704} & &{\textbf{98.27$\%$}}\\\hline
			&{45$\%$}     & {0.15} & {307} & {83}  & {\textbf{72.96$\%$}}&   &{100, 960}& {1, 985} & &{\textbf{98.03$\%$}}\\\hline
			\textsc{Pumsb*}    &{40$\%$}     & {0.45} & {10, 674} & {1646} & {\textbf{84.57$\%$}}&  &{448, 318}& {3, 353}& & {\textbf{99.25$\%$}}\\\hline
			&{40$\%$}     & {0.50} & {9, 760} & {1325} & {\textbf{86.42$\%$}}&  & {82, 413}& {3, 012}& &{\textbf{96.34$\%$}}\\\hline
			\textsc{Connect}    &{10$\%$}     & {0.80} & {534, 026} & {15, 152} & {\textbf{97.16$\%$}}&  & {56}& {56} &  &{\textbf{0$\%$}}\\\hline
			&{50$\%$}     & {0.80} & {533, 991} & {15, 117} & {\textbf{97.16$\%$}}&  & {91}& {91}& &{\textbf{0$\%$}}\\\hline
			\textsc{Accidents}& {40$\%$} & {0.30}&{32, 529} &{32, 528} & {\textbf{0$\%$}}  & & {117, 805} &{1, 722}&  &{\textbf{98.53$\%$}}\\\hline
			& {60$\%$} & {0.30}&{2, 057} &{2, 047} & {\textbf{0$\%$}}&  &{148, 259} &{2, 743}&  &{\textbf{98.14$\%$}}\\
			\noalign{\smallskip}\hline
		\end{tabular}
		\caption{Compactness Rates associated to the $\mathcal{FCP}$ set \textit{vs.} the $\mathcal{RCP}$ set.}
		\label{labelsize1}
	\end{sidewaystable}
\end{center}

At this stage, we have analyzed the variation of the size of the different concise exact representations according to
both \textit{minsupp} and  \textit{minbond} variations. In the next section, we put the focus on the evaluation of the running time of the proposed \textsc{Gmjp} approach.

\section{Evaluation of the running time of \textsc{Gmjp}} \label{sec_XP2}

\subsection{Overall Performance Evaluation of \textsc{Gmjp}}
We emphasize that, according to the results given by Table \ref{labelCPUTime1}
$^\textsc{(}$\footnote{We note that `S1` stands for the First Scenario, `S2` stands for the Second Scenario,
	`S3` stands for the Third Scenario and `S4` stands for the Fourth Scenario.}$^{\textsc{)}}$
, that the execution time varies depending on the number of distinct items of the considered dataset.
This is explained by the principle of \textsc{Gmjp} which is based on the idea of processing each item separately and based on the list of the co-occurrent of each item.

For example, the computation costs are relatively high for the \textsc{T40I10D100K} dataset,
and they are low for the \textsc{Mushroom} dataset.
This is explained by the fact that, the \textsc{Mushroom} dataset only contains \textit{119} items
while the \textsc{T40I10D100K} dataset contains \textit{942} items. We  also note that the highest execution times are obtained with the \textsc{Retail} dataset, since this latter contains a high number of
distinct items, equal to {16, 470}.

\begin{center}
	\begin{sidewaystable}
		\hspace{+0.7cm}
		\begin{tabular}{lrrrrrrrr}
			\hline
			\noalign{\smallskip}
			\textbf{Dataset}& \textbf{Number} & \textbf{Average}     & \textbf{Average}& \textbf{Average} & \textbf{Average} & \textbf{Average} & \textbf{Average} \\
			& \textbf{of Items} & \textbf{\textit{minsupp}} & \textbf{\textit{minbond}}& \textbf{Time S1} & \textbf{Time S2} & \textbf{Time S3} & \textbf{Time S4} \\
			\noalign{\smallskip}\hline\noalign{\smallskip}
			\textsc{Mushroom}  &{119} & {58$\%$}    & {0.30} & {7}    & {11.4} & {20}& {19.6}\\
			& & {40$\%$}       & {0.57} & {3.75} & {5.25} & {11}& {709}\\\hline
			\textsc{Accidents} &{468} & {7.8$\%$}   & {0.50} & {709} & {703} &{793}&{784.2}\\\hline
			\textsc{Retail}    &{16, 470} & {25.83$\%$}    & {0.50} & {5.83} & {13.16} &{1903}&{1902}\\\hline
			\textsc{T10I4D100K} &{870} & {5$\%$}          & {0.20} & {2} & {3} &{163}&{163}\\\hline
			\textsc{T40I10D100K} &{942} & {8.2$\%$}    & {0.50} & {148} & {182.6} &{491}&{490.4}\\
			\noalign{\smallskip}\hline
		\end{tabular}
		\caption{Performance Analysis of \textsc{Gmjp} on UCI benchmarks \textsc{(}Time in seconds\textsc{)}.}
		\label{labelCPUTime1}
	\end{sidewaystable}
\end{center}
It is worth of mention that the computation time of the fourth scenario dedicated to the extraction of the $\mathcal{RCPR}$ representation are the highest ones among the other scenarios. This can be explained by the fact that the extraction of the $\mathcal{RCPR}$ representation is the most complex mining task within the \textsc{Gmjp} approach.  Thereby, we focus on the next subsection on studying the cost of the three different steps of \textsc{Gmjp} when extracting the $\mathcal{RCPR}$ representation.
\subsection{Performance Evaluation of the $\mathcal{RCPR}$ representation Mining}
We study, in this subsection, the running time of need for the extraction of the $\mathcal{RCPR}$ representation.
It is worth noting that the running times of the \textsc{Gmjp} algorithm vary according to the characteristics of the dataset. For example, the computation costs are relatively high for the \textsc{T10I10D100K} dataset, and they are low for the \textsc{Mushroom} dataset.
This is due to the difference in the characteristics of these two datasets. In fact, the \textsc{Mushroom} dataset only contains \textit{119} items and \textit{8, 124} transactions while the \textsc{T10I10D100K} dataset contains \textit{870} items and \textit{100, 000} transactions. We also note, for the different datasets, that the extraction costs are slightly  sensitive to the changes of the \textit{minsupp} and \textit{minbond} values.

We present, in Tables \ref{labelRCPR1} and \ref{labelRCPR2}, the CPU time corresponding to each step of the \textsc{Gmjp} algorithm and depending respectively on the variation of \textit{minsupp} and on the variation of \textit{minbond}.

\begin{table}
	\begin{tabular}{lrrrrrr}
		\hline
		\noalign{\smallskip}
		\textbf{Dataset}& \textbf{\textit{minbond}} & \textbf{\textit{minsupp}}  & \textbf{First Step} & \textbf{Second Step} & \textbf{Third Step}\\
		\noalign{\smallskip}\hline\noalign{\smallskip}
		\textsc{Mushroom} &{0.3}& {20$\%$} & {6} & {13} &{0}\\
		&     & {40$\%$} & {6} & {13} &{0}\\
		&     & {60$\%$} & {6} & {14} &{0}\\
		&     & {80$\%$} & {6} & {15} &{0}\\\hline
		\textsc{Retail}   &{0.5}& {5$\%$} &{276} &{1, 627} &{0}\\
		&     & {10$\%$}&{273} &{1, 627} &{0}\\
		&     & {30$\%$}&{274} &{1, 628} &{0}\\
		&     & {50$\%$}&{275} &{1, 628} &{0}\\\hline
		\textsc{Accidents}&{0.5}& {1$\%$} &{724} &{67} &{0}\\
		&     & {3$\%$}&{714} &{67} &{0}\\
		&     & {5$\%$}&{716} &{67} &{0}\\
		&     & {10$\%$}&{715} &{67} &{0}\\
		&     & {15$\%$}&{717} &{67} &{0}\\
		&     & {20$\%$}&{717} &{67} &{0}\\
		\noalign{\smallskip}\hline
	\end{tabular}
	\caption{Impact of the variation of \textit{minsupp}, for the three steps of the \textsc{Gmjp} algorithm \textsc{(}Time in seconds\textsc{)}.}
	\label{labelRCPR1}
\end{table}

We conclude, according to these results, that the obtained execution times are slightly sensitive to the variation of the \textit{minbond} and \textit{minsupp}
values.
\begin{table}
	\hspace{-0.5cm}
	\begin{tabular}{lrrrrrr}
		\hline
		\noalign{\smallskip}
		\textbf{Dataset}& \textbf{\textit{minsupp}} & \textbf{\textit{minbond}}& \textbf{First Step} & \textbf{Second Step} & \textbf{Third Step}\\
		\noalign{\smallskip}\hline\noalign{\smallskip}
		\textsc{Mushroom} &{50$\%$}     & {0.4} & {6} & {4} &{0}\\
		&             & {0.7} & {6} & {1} &{0}\\
		&             & {1} & {6} & {1} &{0}\\\hline
		\textsc{T10I4D100K}& {5$\%$} & {0.2}&{25} &{137} &{0}\\
		& & {0.4}&{26} &{138} &{0}\\
		& & {0.6}&{26} &{138} &{0}\\
		& & {0.8}&{25} &{138} &{0}\\
		& & {1}&{25} &{137} &{0}\\
		\noalign{\smallskip}\hline
	\end{tabular}
	\caption{Impact of the variation of \textit{minbond}, for the three steps of the \textsc{Gmjp} algorithm \textsc{(}Time in seconds\textsc{)}.}
	\label{labelRCPR2}
\end{table}

We have, for example, the CPU time needed for the execution of the first step for the \textsc{T10I4D100K} dataset is about \textit{26} seconds, while
for the \textsc{Retail} dataset, the execution of the first step needs about \textit{275} seconds. This is justified by the fact that
the \textsc{T10I4D100K} dataset contains only \textit{870} items while the \textsc{Retail} dataset contains \textit{16, 470} items.
We also remark, for the \textsc{Accidents} dataset, that the execution times are relatively high compared to the other datasets. This is justified
by its high number of transactions, equal to \textit{340, 183} transactions which induces that the first step becomes more costly.
In this regard, the first step of the building of the BSVector and the COVector of the items needs about \textit{720} seconds and it lasts
more than the second step, which needs only \textit{67} seconds.

It is also important to mention that the CPU time dedicated to the third step, allowing to filter the global minimal and the global closed patterns among the sets of the identified local minimal and local closed ones, is negligible and equal to null. This confirms the very good choice of the suitable multimap data structures during the third step.

We also note, that the execution times needed for the post-processing of the representations $\mathcal{MM}$ax$\mathcal{CR}$, $\mathcal{M}$in$\mathcal{MCR}$ from the $\mathcal{RCPR}$ representation are negligible.

We have, thus, evaluated the performance of the \textsc{Gmjp} approach while running the four different execution scenarios. We focused specially on the fourth scenario dedicated to the extraction of the $\mathcal{RCPR}$ representation.  In the next section, we evaluate \textsc{Opt-Gmjp} the optimized version of the \textsc{Gmjp} approach.
\section{Optimizations and Evaluations} \label{sec_Opt}
Our aim in the next subsections is to evaluate the impact of varying the thresholds of both the correlation and the rarity constraints.

In the remainder, we study the impact of the rarity constraint threshold variation on the execution time of the \textsc{Opt-Gmjp} version.
\subsection{Effect of \textit{minsupp} variation}
Table \ref{labeltimexp2} presents our results while fixing the \textit{minbond} threshold and varying the monotone constraint of rarity threshold,  \textit{minsupp}.
We consider as an example the \textsc{Mushroom} dataset, while varying \textit{minsupp}  from {20$\%$} to {80$\%$}, 
the size of the result set |$\mathcal{RCP}$| varies from 261 itemsets to 3352 itemsets while the CPU-time and the memory consumption underwent a slight variation.
Whereas, for the \textsc{T40I10D100K} dataset, the variation of \textit{minsupp}  from {2$\%$} to {15$\%$} induces an increase in the CPU-time of the second and third steps
from {60.09} to {79.49} seconds. The size of the output result increase also from {341} to {932} itemsets.

The \textsc{Chess} dataset presents a specific behavior according to \textit{minsupp}  variation.
The variation of \textit{minsupp}  from {30$\%$} to {50$\%$} induces an increase in the CPU-time of the second and third steps
from {5.604} to {300.163} seconds. The size of the $\mathcal{RCP}$ set increase in a very significant way  from {618} to {36010, 648} itemsets.

\subsection{Effect of \textit{minbond} variation}
Table \ref{labelbondxp3} presents the results obtained when varying the correlation threshold  \textit{minbond}  for a fixed
\textit{minsupp} threshold.  In this experiment, we found that for the
\textsc{Mushroom}  dataset,  the \textit{minbond} threshold was  chosen to be increasingly  selective, from {0.2} to the highest value, equal to {1}.
This variation  affects very slightly  the CPU-time and the memory consumption, while the size of the output set decreases sharply from
{54, 395}  to 126 itemsets.
Whereas, for the sparse
\textsc{T10I4D100K}  dataset, the \textit{minbond} variation act slightly  on the execution time.
It increases just by 2 seconds, and the output's size decreases from 915 to 860 itemsets.
However,  for the \textsc{Chess} dataset,
the size of $\mathcal{RCP}$ set and the CPU-time are very sensitive to  the \textit{minbond}  variation.
For example, a slight variation of \textit{minbond} from 0.40 to 0.45 induces an important decrease of the $\mathcal{RCP}$ set from
5167,  090  to  1560,  073 itemsets. The CPU-time  is also lowered from 40.124 to 0.451 seconds when \textit{minbond} decrease from 0.4 to 0.5.


The most interesting observation we found from the previous experiments was that the choice of very selective correlation threshold  do not affect significantly the CPU-time and the memory consumption, while it affects the size of our result set. Whereas,  the fact of pushing more selective the rarity constraint increases the execution time needed for the second and the third steps.  This confirms that monotone and anti-monotone constraints are mutually of use in the selectivity. It is also important to mention that, the first step of transforming the database is not affected by both constraints variation.

In the next sub-section, we evaluate the proposed optimization and we compare the optimized version of \textsc{Gmjp} \textit{vs.} the \textsc{Jim} approach \cite{borgelt}.
\begin{sidewaystable} [htbp]
	\begin{center}
		\hspace{+0.7cm}
		\begin{tabular}{lrrrrll}
			\hline
			\noalign{\smallskip}
			\textbf{Dataset}           & \textbf{\textit{minbond}}             & \textbf{\textit{minsupp}} &   \textbf{|$\mathcal{RCP}$|}      & \textbf{CPU Time}       &  \textbf{CPU Time}  &  \textbf{Avg. Memory}\\
			&                                                              &                                               &                                                           & \textbf{ Step 1}           &   \textbf{Steps 2 and 3}      & \textbf{Consumption (Ko)}\\
			\noalign{\smallskip}\hline\noalign{\smallskip}
			\textsc{Mushroom}                        &{0.30}            & {20$\%$}  &    {261}                 &           {0.20}                       & \textbf{{0.208}} & \\
			&                        & {40$\%$}  &        {2810}                           &           {0.20}        & \textbf{{0.246}} & {18, 850}\\
			&                    & {80$\%$}  & {3352}                                            &           {0.20}                       & \textbf{{0.247}} & \\\hline
			\textsc{T40I10D100K}                    & {0.50}        & {2$\%$}   &      {341}                                     &                {16}                 &\textbf{{60.09}}  & \\
			&                     & {11$\%$}     &          {889}                               &                {16}                 &\textbf{{63.028}}  & {131,516}\\
			&                    & {15$\%$}        &            {932}                          &     {16}                 &\textbf{{79.49}}  & \\\hline
			\textsc{Chess}                                &{0.30}            & {10$\%$}  &    {16}                                 &           {0.068}                       & \textbf{{0.116}} & \\
			&                       & {30$\%$}  &        {618}                           &           {0.068}                       & \textbf{{5.604}} & {13,509}\\
			&                    & {50$\%$}  & {36010, 648}                     &           {0.068}                       & \textbf{{300.163}} & \\\hline
			\noalign{\smallskip}
		\end{tabular}
		\caption{Impact of the rarity  threshold \textit{minsupp} variation.}
		\label{labeltimexp2}
	\end{center}
\end{sidewaystable}
\begin{sidewaystable} [htbp]
	\begin{center}
		\hspace{+0.7cm}
		\begin{tabular}{lrrrlll}
			\hline
			\noalign{\smallskip}
			\textbf{Dataset}           & \textbf{\textit{minsupp}}             & \textbf{\textit{minbond}}  &  \textbf{|$\mathcal{RCP}$|}      & \textbf{CPU Time}       &  \textbf{CPU Time}  &  \textbf{Avg. Memory} \\
			&                                                              &                                                 &                                                         & \textbf{ Step 1}           &   \textbf{Steps 2 and 3}      & \textbf{Consumption (Ko)}\\
			\noalign{\smallskip}\hline\noalign{\smallskip}
			\textsc{Mushroom}   & {40$\%$}                   & {0.2} &     {54, 395}     &  {0.20}                         & \textbf{{0.977}} & {18, 590}   \\
			&                                         & {1} &        {126}          &{0.20}                         & \textbf{{0.198}} &     \\\hline
			\textsc{Chess}         & {50$\%$}                   & {0.40} &          {5167,  090}     &  {0.068}                         & \textbf{{40.124}} &  \\
			&                                     & {0.45} &        {1560, 073}          &{0.068}                         & \textbf{{12.127}} &     \\
			&                                     & {0.50} &           {162}          &{0.068}                         & \textbf{{0.451}} &    {13, 556}  \\
			&                                     & {0.60} &           {40}          &{0.068}                         & \textbf{{0.073}} &     \\
			&                                     & {1} &           {38}          &{0.068}                         & \textbf{{0.054}} &     \\\hline
			\textsc{T10I4D100K}   & {5$\%$}                 & {0.40} &         {915}      & {16}              &\textbf{{47.95}}  & {131, 572}\\
			&                                  & {1}      &           {860}     &{16}              &\textbf{{49.48}}  & \\\hline
			
			\noalign{\smallskip}
		\end{tabular}
		\caption{Impact of the correlation threshold \textit{minbond} variation.}
		\label{labelbondxp3}
	\end{center}
\end{sidewaystable}
\subsection{Performance of \textsc{Opt-Gmjp} \textit{vs.} \textsc{Gmjp}}

\textbf{6.5.4.1 Comparison of \textsc{Opt-Gmjp} \textit{vs.} \textsc{Gmjp}} \\

The goal of our evaluation is to compare the scalability level of  \textsc{Opt-Gmjp} \textit{vs.} \textsc{Gmjp}.
In fact, scalability is an important criteria for constrained itemset mining approaches. Our \textsc{Opt-Gmjp} algorithm demonstrates  good scalability as far as we increase the size of the datasets according to two dimensions: the number of transactions
|$\mathcal{T}$| and the number of items |$\mathcal{I}$|.
While  \textsc{Gmjp} reached a point where it consumed about seven times more CPU-time than \textsc{Opt-Gmjp}.
Tables \ref{Cmp1} and \ref{CORI2} reported our results while varying respectively \textit{minsupp} and  \textit{minbond}.

As example,  while testing the \textsc{Mushroom}  dataset containing 8, 124 transactions,
\textsc{Opt-Gmjp} finishes in average in {0.432} seconds while \textsc{Gmjp}  needs  in average $11$ seconds.
As another experiment example, we tested also on \textsc{Accidents} with  340, 183 transactions and 468 items. \textsc{Opt-Gmjp} finishes in $47.617$ seconds
while \textsc{Gmjp}  needs  $793$ seconds.  \textsc{Gmjp}  finished the \textsc{Mushroom}  dataset with about 8K transactions  in $20$ seconds while
\textsc{Opt-Gmjp} finished, in average, in  $0.278s$, $52.591s$ and  $47.617s$ for the 8K, 100K and 340K transactions datasets,  respectively.
\begin{table} [htbp]
	\begin{center}
		
		\hspace{+0.7cm}
		\begin{tabular}{lrrrlr}
			\hline
			\noalign{\smallskip}
			\textbf{Dataset} &\textbf{\textit{minsupp}}& \textbf{\textit{minbond}}&\textbf{CPU Time}& \textbf{CPU Time}  \\
			& &                              & \textbf{\textsc{Gmjp}} & \textbf{ \textsc{Opt-Gmjp}}  \\
			\noalign{\smallskip}\hline\noalign{\smallskip}
			\textsc{Mushroom}   & {20$\%$} & {0.30}     &  {20}                       & \textbf{{0.14}}\\
			& {40$\%$} &            &     {19}                         & \textbf{{0.18}}\\
			& {60$\%$} &            &     {19}                         & \textbf{{0.18}}\\
			& {80$\%$} &            &     {21}                         & \textbf{{0.18}}\\
			\hline
			\textsc{Accidents} & {1$\%$}                   & {0.50}   &{802}         &\textbf{{22}}\\
			& {3$\%$}                   &          &{802}         &\textbf{{22}}\\
			& {5$\%$}                   &         &{790}         &\textbf{{21}}\\
			& {10$\%$}                   &    &{783}         &\textbf{{21}}\\
			& {12$\%$}                   &    &{788}         &\textbf{{22}}\\
			\hline
			\textsc{T40I10D100K}   & {2$\%$}           & {0.50}  &        {489}    &\textbf{{51}}\\
			& {5$\%$}           &   &        {494}    &\textbf{{53}}\\
			& {8$\%$}           &   &        {493}    &\textbf{{51}}\\
			& {11$\%$}           &   &        {489}    &\textbf{{51}}\\
			& {15$\%$}           &   &        {490}    &\textbf{{51}}\\
			
			\noalign{\smallskip}\hline
		\end{tabular}
		\caption{Performance comparison  of  \textsc{Opt-Gmjp} \textit{vs.} \textsc{Gmjp} while varying \textit{minsupp} \textsc{(}Time in seconds\textsc{)}.}
		\label{Cmp1}
	\end{center}
\end{table}

\begin{table} [htbp]
	\begin{center}
		\hspace{+0.7cm}
		\begin{tabular}{lrrrlr}
			\hline
			\noalign{\smallskip}
			\textbf{Dataset} &\textbf{\textit{minsupp}}& \textbf{\textit{minbond}}&\textbf{CPU Time}& \textbf{CPU Time}  \\
			& &                              & \textbf{\textsc{Gmjp}} & \textbf{ \textsc{Opt-Gmjp}}  \\
			\noalign{\smallskip}\hline\noalign{\smallskip}
			\textsc{T10I4D100K} & {5$\%$}                 & {0.20} &      {163}              &\textbf{{39}}\\
			&                 & {0.40} &      {164}              &\textbf{{40}}\\
			&                 & {0.60} &      {163}              &\textbf{{39}}\\
			&                 & {0.80} &      {163}              &\textbf{{39}}\\
			&                 & {1} &      {163}              &\textbf{{39}}\\
			\hline
			\textsc{Mushroom} & {40$\%$}            & {0.20} &      {21}              &\textbf{{0.90}}\\
			&                 & {0.40} &      {9}              &\textbf{{0.14}}\\
			&                 & {0.70} &      {7}              &\textbf{{0.13}}\\
			&                 & {1} &      {7}              &\textbf{{0.13}}\\
			\noalign{\smallskip}\hline
		\end{tabular}
		\caption{Performance comparison  of  \textsc{Opt-Gmjp} \textit{vs.} \textsc{Gmjp} while varying \textit{minbond} \textsc{(}Time in seconds\textsc{)}.}
		\label{CORI2}
	\end{center}
\end{table}

\begin{landscape}
	\pagestyle{empty}
	\font\xmplbx = cmbx7.5 scaled \magstephalf
	\begin{table}
		\hspace{-2.0cm}
		\begin{tabular}{|l|c|c|rr|rr|rr|rr|}
			\hline
			\noalign{\smallskip}
			\textbf{Dataset}& \textbf{Avg}& \textbf{Avg}& \multicolumn{2}{c|} {\textbf{S1}} & \multicolumn{2}{c|} {\textbf{S2}} & \multicolumn{2}{c|} {\textbf{S3}} & \multicolumn{2}{c|} {\textbf{S4}} \\
			
			& \textbf{\textit{minsupp}} & \textbf{\textit{minbond}}&   &  &   &  & & &  & \\
			\midrule
			&                          &                           &   \textbf{Gmjp}& \textbf{Opt-Gmjp} &  \textbf{Gmjp}& \textbf{Opt-Gmjp} & \textbf{Gmjp}& \textbf{Opt-Gmjp} & \textbf{Gmjp}& \textbf{Opt-Gmjp} \\
			\noalign{\smallskip}\hline\noalign{\smallskip}
			\textsc{Mushroom}  & {58$\%$}    & {0.30} & {7} & {0.114}   & {11.4} & {0.052}& {20} & {0.172}& {19.6} & {0.206} \\
			& {40$\%$}   & {0.57} & {3.75} & {0.096} & {5.25} & {0.058} & {11} & {0.325}& {709} & {0.525}\\\hline
			\textsc{Accidents} &{7.8$\%$}   & {0.50} & {709} &{7.094} & {703} &{7.978} &{793} &{22.034} &{784.2} &{22.430}\\\hline
			\textsc{T10I4D100K} & {5$\%$}   & {0.20} & {2} &{0.132} & {3} &{0.14}&{163} & {39.804}&{163} & {39.424}\\\hline
			\textsc{T40I10D100K} & {8.2$\%$}    & {0.50} & {148} & {4.222}& {182.6} & {7.566} &{491} & {51.798}&{490.4} & {51.740}\\
			\noalign{\smallskip}\hline
		\end{tabular}
		\caption{Summarized Comparison of the Performance of \textbf{\textsc{Gmjp}} \textit{vs.} \textbf{Optimized \textsc{Gmjp}} \textsc{(}Time in seconds\textsc{)}.}
		\label{addlabel}
	\end{table}
\end{landscape}

We highlight that \textsc{Opt-Gmjp} outperformed  \textsc{Gmjp} in the different evaluated bases. This is dedicated to the efficient integration of the monotone and anti-monotone constraints in an early  stages of the mining process.
We also present in Table \ref{addlabel} a summarized comparison of the performances of  \textsc{Gmjp} \textit{vs.} \textbf{Optimized \textsc{Gmjp}} over the four different scenarios S1, S2, S3 and S4. We thus conclude that the optimized version of \textsc{Gmjp} offers important reduction of the running time over all the tested benchmark datasets and for wide range of constraints threshold. In what follows, we evaluate our optimized version
\textit{vs.} the \textsc{Jim} approach \cite{borgelt}.\\

\textbf{6.5.4.2 Comparison of \textsc{Opt-Gmjp} \textit{vs.} \textsc{Jim}} \\


The goal of these experiments is to prove the competitive performances of \textsc{Opt-Gmjp} compared to other state-of the art approaches dealing with frequent correlated itemsets. Our comparative study covers the \textsc{Jim} approach \cite{borgelt} which is implemented in the C language and is publicly available.
\font\xmplbx = cmbx7.5 scaled \magstephalf
\begin{table}[htbp]
	\centering
	\twlrm{\begin{tabular}{ l l l || l  l || l  l }
			\hline
			\textbf{Dataset}                          & \textbf{\textit{minsupp}} &  \textbf{\textit{minbond}}      &    \textbf{Opt \textsc{Gmjp}}  &  \textbf{\textsc{Jim} } &    \textbf{Opt \textsc{Gmjp}} &     \textbf{\textsc{Jim}} \\
			&                                               &                                                       &   \textbf{S1}&                                     &                         \textbf{S2}               & \\
			\hline
			\textsc{T10I4D100K}                                 & {5$\%$}                   & {0.20}     &    \textbf{0.133}                         & 0.20&    \textbf{0.15}       &0.19\\
			&                                       & {0.40} &    \textbf{0.133}                          & 0.19&        \textbf{0.13}   & 0.18\\
			&                                       & {0.60} &    \textbf{0.132}                            & 0.18 &       \textbf{0.13}     & 0.18\\
			&                                       & {0.80} &     \textbf{0.129}                            & 0.18&       \textbf{0.13}      & 0.18\\
			&                                      & {1}        &       \textbf{0.135}                          & 0.19 &      \textbf{0.13}  & 0.19\\
			\hline
			\textsc{Mushroom}                                                                            & {20$\%$}  & {0.30} &       0.082                     &  \textbf{0.06}&      0.140     &\textbf{0.03}\\
			& {40$\%$} &               &             \textbf{0.029}                   & 0.03&       0.060    & \textbf{0.03}\\
			& {60$\%$} &                &           0.200                      & \textbf{0.02} &        0.020    & 0.02\\
			& {80$\%$}&                 &               0.200                 & \textbf{0.02}&      0.023       & \textbf{0.02}\\
			& {90$\%$}   &                &                0.210                 & \textbf{0.02} &    \textbf{ 0.019}   & 0.02\\
			\hline
			\textsc{Retail}                                                                           & {5$\%$}  & {0.50} &          \textbf{0.249}                 &  0.25&    0.46       & \textbf{0.25}\\
			& {10$\%$} &               &            0.250                   & 0.23&     0.36      &  \textbf{0.24}\\
			& {20$\%$} &                &               0.249              & 0.22 &    0.26        &  \textbf{0.22}\\
			& {40$\%$}&                 &                0.240            & 0.23&       0.25     &  \textbf{0.22}\\
			& {50$\%$}   &                &              0.240            & 0.22 &      0.36  &  \textbf{0.20}\\
			\hline
	\end{tabular}}
	\smallskip
	\caption{Performance comparison  of our Improved  \textbf{\textsc{Opt-Gmjp}} \textit{vs.}  \textsc{Jim} \cite{borgelt} \textsc{(}Time in seconds\textsc{)}.}
	\label{TabCmpJim}
\end{table}
We report in Table \ref{TabCmpJim}
$^\textsc{(}$\footnote{We note that ``S1'' stands for the First Scenario and ``S2'' stands for the Second Scenario.}$^{\textsc{)}}$
a comparison between our improved \textsc{Gmjp} approach with the  \textsc{Jim} approach \cite{borgelt}.
Our comparative study is restricted to the first Scenario S1 and second scenario S2, since the  \textsc{Jim} approach does not consider the rare correlated patterns. Therefore, we are not able to compare the third and
the fourth scenarios S3 and S4.
We highlight that our running time are competitive to those achieved by \textsc{Jim} for different ranges of frequency and correlation thresholds. Note-worthily, for the \textsc{T10I4D100K} dataset, our obtained results are even better than \textsc{Jim} for both first and second scenarios. While, for the  \textsc{Mushroom} dataset, the results of the first scenario are very close to those of \textsc{Jim}. Whereas, \textsc{Jim} outperformed our \textsc{Gmjp} in the second scenario when extracting the frequent closed correlated itemsets.
\section{Conclusion} \label{ConcChap6}
We presented in this chapter the experimental evaluation of our \textsc{Gmjp} mining approach. The evaluation is based on two main axes, the first is related to the compactness rates of the condensed representations  while the second axe concerns the running time. We measured the global performance of \textsc{Gmjp} then we focused on the performance of the fourth execution scenario S4. The optimized version \textsc{Opt-Gmjp} presents much better performance than \textsc{Gmjp} over different benchmark datasets. The two main features which constitute the thrust of the improved version: (\textit{i}) only one scan of the database is performed to build the new transformed dataset; (\textit{ii}) it offers a resolution of the problem of handling both rarity and correlation constraints. In the next chapter, we present the classification process based on correlated patterns.

\chapter{Associative-Classification Process based on Correlated Patterns}\label{ch_7}

\section{Introduction} \label{IntroChap7}

In this chapter,  we put the focus on the classification process based on correlated patterns.
The second section is devoted to the description of the framework of the association rules. We continue in the third section with a specific kind of association called ``Generic Bases of Association Rules``.
The fourth section presents the description of the associative-classification
based on correlated patterns. We evaluate the classification accuracy of  frequent correlated patterns \textit{vs.} rare correlated patterns. In Section \ref{se-IDS}, we present the application of rare correlated patterns on the classification of intrusion detection data derived from the \textsc{KDD 99} dataset. In Section \ref{SecGene}, we propose the process of applying the $\mathcal{RCPR}$ representation on the extraction of rare correlated association rules from Micro-array gene expression data.
\section{Overview of association rules} \label{se2}
The extraction of association rules is one of the most important techniques in data mining \cite{BoukerSYN14,GasmiYNB07}. The leading approach
of generating association rules is based on the extraction of frequent patterns \cite{Agra94}.
We clarify the basic notions related to association rules through the following definitions.

\begin{definition} \textbf{Association Rule} \\
	An association rule $R$ is a relation between itemsets, in the form $R$ : $A$ $\Rightarrow$ $B$$\backslash$$A$, with $A$ and $B$ are two itemsets and $A$ $\subset$ $B$. The itemset $A$ is called  `Premise` \textsc{(}or `Antecedent`\textsc{)} whereas
	the itemset $B$$\backslash$$A$ is called `Conclusion` \textsc{(}or `Consequent`\textsc{)} of the association rule $R$.
\end{definition}
Each association rule, $R$ : $A$ $\Rightarrow$ $B$$\backslash$$A$, is characterized by:
\begin{enumerate}
	\item \textbf{The value of the Support:} Corresponding to the number of times where the
	association holds reported by the number of occurrence of the itemset $B$. The support metric assesses the frequency of the association rule.
	\item \textbf{The value of the Confidence:} Corresponding to the number of times where the
	association holds reported by the number of occurrence of the itemset $A$. The confidence expresses the reliability of the rule.
\end{enumerate}
The support and the confidence are formally defined as follows:
\begin{definition} \textbf{Support, Confidence of an association rule} \\
	Let an association rule $R$ : $A$ $\Rightarrow$ $B$$\backslash$$A$, its support, denoted by  \textit{Supp}\textsc{(}$R$\textsc{)} = \textit{Supp}\textsc{(}$B$\textsc{)}, where as the confidence, denoted by,
	\textit{Conf}\textsc{(}$R$\textsc{)} = $\displaystyle \displaystyle\frac{\textit{Supp}\textsc{(}B\textsc{)}}{\textit{Supp}\textsc{(}A\textsc{)}}$.
\end{definition}
\begin{definition} \textbf{Valid, Exact and Approximative Association Rule} \\
	An association rule $R$ is said Valid whenever:\\
	$\bullet$ The value of the confidence is greater than or equal to the minimal threshold of confidence \textit{minconf}, and \\
	$\bullet$ The value of its support is greater than or equal to the minimal threshold of support  \textit{minsupp}.
	If the confidence of the rule $R$, Conf\textsc{(}R\textsc{)}, is equal to 1 then the rule $R$ is said an Exact association rule, otherwise it is said approximative.
\end{definition}
The extraction of the association rules consists in determining the set of valid rules \textit{i.e.}, whose support and confidence are at least equal, respectively, to a minimal threshold of support \textit{minsupp} and a minimal threshold
of confidence \textit{minconf} predefined by the user. This problem is decomposed into two subproblems \cite{Agra94} as follows:

1. Extraction of frequent itemsets;

2. Generation of valid association rules based on the frequent extracted itemset set: the generated rules
are in the form $R$ : $A$ $\Rightarrow$ $B$$\backslash$$A$, with $A$ $\subset$ $B$ and Conf\textsc{(}$R$\textsc{)}$\geq$ \textit{minconf}.

The association rule extraction problem suffers from the high number of the generated association rules from the frequent itemset set. In fact, the number of the extracted frequent itemsets
can be exponential in function of the number of items $|\mathcal{I}|$.
In fact, from a frequent itemset $F$, we can generate $2^{|F|}-1$ association rules. The huge number of association rules leads to a deviation regarding to the principal objective namely, the discovery of reliable knowledge and with a manageable size. To palliate this problem, many techniques derived from the Formal Concept Analysis \textsc{(}FCA\textsc{)}, were proposed. These techniques aimed to reduce, without information loss, the set of association rules. The main idea is to determine a minimal set of association rules allowing to derive the redundant association rules, this set is called ```Generic bases of association rules''.

\section{Extraction of the generic bases of association rules} \label{se3}
The approaches derived from the \textsc{FCA} allows to extract the generic bases of association rules. These generic bases allow to derive the set of redundant association rules without information loss. In fact, these bases constitute a compact set of association rules easily interpretable by final user.
Every generic base constitutes an information lossless representation of the whole set of association rules
if it fulfills the following properties \cite{marzena022}:
\begin{itemize}
	\item \textbf{Lossless:} The generic base must enable the derivation of all valid association rules,
	\item \textbf{Sound:}  The generic base must forbid the derivation of the non valid association rules, and,
	\item \textbf{Informative:} The generic base must allow to exactly retrieve the support and confidence values of all the generated rules.
\end{itemize}
The majority of the generic bases of association rules express implications between minimal generators and closed frequent itemsets \cite{marzena022,tarekcla06_paper_2_revised_version,pasquierIGIBook09}. In this thesis,
we focus on the $\mathcal{IGB}$ generic base \cite{IGB} defined in what follows.
\begin{definition} \textbf{The $\mathcal{IGB}$ Generic Base} \cite{IGB}\\
	Let $\mathcal{FCP}$ be the set of frequent closed patterns
	and let  $\mathcal{FMG}$ be the set of frequent minimal generators of all the frequent closed itemsets include or equal to
	a frequent closed itemset $F$. The $\mathcal{IGB}$ base is defined as follows:\\
	$\mathcal{IGB}$ = $\{$$R$:  $fmg$ $\Rightarrow$ \textsc{(}$F$$\backslash$$fmg$\textsc{)}
	$\lvert$ $F$ $\in$ $\mathcal{FCP}$, $fmg$ $\in$ $\mathcal{FMG}$, \textsc{(}$F$$\backslash$$fmg$\textsc{)} $\neq$ $\emptyset$,
	Conf\textsc{(}$R$\textsc{)} $\geq$ \textit{minconf}, $\nexists$ $g_{1}$ $|$ $g_{1}$ $\in$ $\mathcal{FMG}$ and
	Conf\textsc{(}$g_{1}$ $\Rightarrow$ \textsc{(}$F$$\backslash$$g_{1}$\textsc{)}\textsc{)} $\geq$ \textit{minconf}.$\}$
\end{definition}
Thus, the generic rules of the $\mathcal{IGB}$ generic base represent implications between
the minimal premises, according to the size on number of items, and the maximal conclusions.

\section{Association rule-based classification process} \label{chap7se4}
\subsection{Description} \label{subse-descrip}
We present in the following, the application of the $\mathcal{RCPR}$ and the $\mathcal{RFCCP}$ representations in the
design of an association rules based classifier. In fact, we used the $\mathcal{MRCP}$ and the $\mathcal{CRCP}$ sets, composing the $\mathcal{RCPR}$ representation,
within the generation of the generic
$^\textsc{(}$\footnote{By ``generic'', it is meant that these rules are with minimal premises and maximal conclusions, w.r.t. set-inclusion.}$^{\textsc{)}}$
rare correlated rules.
The $\mathcal{RFCCP}$ representation is used to generate generic
frequent correlated rules, of the form $Min$ $\Rightarrow$ $Closed$ $\setminus$ $Min$, with $Min$ is a minimal generator and $Closed$ is a closed pattern.
The procedure allowing the extraction of the generic correlated association rules is an adapted version of the original \textsc{GEN-IGB} \cite{IGB} that we implemented as a \textsc{C++} program.

Then, from the generated set of the generic rules, only the classification rules will be retained, \textit{i.e.}, those having the label of the class in its conclusion part.
Subsequently, a dedicated associative-classifier is fed with these rules and has to perform the classification process and returns the accuracy rate for each class.

The aim of the evaluation of the classification process is the comparison of the
effectiveness of frequent correlated patterns \textit{vs.} rare correlated patterns within the classification process.  The comparison is carried out through two directions:\\
$\bullet$ Study of the impact of \textit{minbond} variation \\
$\bullet$ Study of the impact of \textit{minconf} variation.

\subsection{Effect of \textit{minbond} variation}

The accuracy rate of the classification, is equal to
$\frac{NbrCcTr}{TotalNbrTr}$, with $NbrCcTr$ stands for the number of the correctly classified transactions and $TotalNbrTr$ is equal to the whole number of the classified transactions.
The classification results reported in Table \ref{minbondcls1} corresponds to the variation of the correlation constraint for a fixed \textit{minsupp} and \textit{minconf} thresholds, with \textit{minconf} corresponds to the minimum threshold of the confidence measure \cite{Agra94}.

We remark, for the frequent correlated patterns, that as far as we increase the \textit{minbond} threshold, the number of exact and approximate association rules decreases while maintaining always an important accuracy rate. Another benefit
for the \textit{bond} correlation measure integration, is the improvement of the response time, that varies from \textit{1000} to \textit{0.01} seconds.
Whereas, for the rare correlated patterns, we highlight that the increase of the \textit{minbond} threshold induces a reduction in the accuracy rate. This is explained by a decrease in the number of the obtained classification rules.

\begin{landscape}
	\thispagestyle{empty}
	\begin{table}
		\hspace{-2.3cm}
		\begin{tabular}{lrrrrrrrrr}
			\hline
			\noalign{\smallskip}
			\textbf{Dataset}& \textbf{\textit{minsupp}} & \textbf{\textit{minconf}}&  \textbf{\textit{minbond}}  & \textbf{$\#$ Exact} & \textbf{$\#$ Approximate} & \textbf{$\#$ Classification} & \textbf{Accuracy} & \textbf{Response} & \textbf{Property of}\\
			&  & &   & \textbf{Rules} & \textbf{Rules} & \textbf{Rules} & \textbf{rate} & \textbf{Time \textsc{(}sec\textsc{)}} & \textbf{Patterns}\\
			\noalign{\smallskip}\hline\noalign{\smallskip}
			\textsc{Wine}     & {1$\%$}& {0.60} & {0} & {387} & {5762} &{650} & {97.75$\%$} & {1000} &{Frequent}\\
			&     &          & {0.10} & {154} &{2739}   &{340} & {95.50$\%$} &{13.02}&{Frequent}\\
			&     &          & {0.20} & {60} &{1121}   &{125} & {94.38$\%$} & {1.00}&{Frequent}\\
			&     &          & {0.30} & {20} &{319}   &{44} & {87.07$\%$} & {0.01}&{Frequent}\\\hline
			\textsc{Zoo}      & {50$\%$}& {0.70} & {0.30} & {486} & {2930} &{235} & {89.10$\%$} & {40}&{Rare}\\
			&     &            & {0.40} & {149} &{436}   &{45} & {89.10$\%$} &{3}&{Rare}\\
			&     &            & {0.50} & {38} &{88}   &{11} & {83.16$\%$} & {0.01}&{Rare}\\
			&     &            & {0.60} & {12} &{31}   &{6} & {73.26$\%$} & {0.01}&{Rare}\\\hline
			\textsc{TicTacToe}&  {10$\%$}   &{0.80}&   {0} & {0} & {16} &{16} & {69.40$\%$}&-&{Frequent}\\
			&&&{0.05}  & {0} & {16} &{16} & {69.40$\%$}&-&{Frequent}\\
			&&&{0.07}  & {0} & {8} &{8} & {63.25$\%$}&-&{Frequent}\\
			&&&{0.1}  & {0} & {1} &{1} & {60.22$\%$}&-&{Frequent}\\
			&&&{0}    &  {1, 033} & {697} &{192} & {100.00$\%$}&-&{Rare}\\
			&&&{0.05}  & {20} & {102} &{115} & {100.00$\%$}&-&{Rare}\\
			&&&{0.07}  & {8} & {66} &{69} & {97.07$\%$}&-&{Rare}\\
			&&&{0.1}  & {2} & {0} &{1} & {65.34$\%$}&-&{Rare}\\
			
			\noalign{\smallskip}\hline
		\end{tabular}
		\caption{Evaluation of the classification accuracy \textit{versus} \textit{minbond} variation for frequent and rare correlated patterns.}
		\label{minbondcls1}
	\end{table}
\end{landscape}
\begin{landscape}
	\thispagestyle{empty}
	\vspace{-2.0cm}
	\begin{table}
		\hspace{-2.0cm}
		\begin{tabular}{lrrrrrrrr}
			\hline
			\noalign{\smallskip}
			\textbf{Dataset}& \textbf{\textit{minbond}} & \textbf{\textit{minsupp}}&  \textbf{\textit{minconf}}  & \textbf{$\#$ Exact} & \textbf{$\#$ Approximate} & \textbf{$\#$ Classification} & \textbf{Accuracy} & \textbf{Property of}\\
			&  & &   & \textbf{Rules} & \textbf{Rules} & \textbf{Rules} & \textbf{rate} & \textbf{\textit{Correlated}}\\
			&  & &   & &  &  &                                                          & \textbf{patterns}\\
			\noalign{\smallskip}\hline\noalign{\smallskip}
			\textsc{Wine}     &{0.1}& {20$\%$} & {0.60} & {7} & {274} &{25} & {76.40$\%$} & {Frequent}\\
			&     &          & {0.80} & {7} &{86}   &{10} & {86.65$\%$} &{Frequent}\\
			&     &          & {0.90} & {7} &{30}   &{4} & {84.83$\%$} & {Frequent}\\\hline
			&{0.1}& {20$\%$} &{0.60} &{91} &{1516} &{168} & {\textbf{95.50$\%$}} & {\textbf{Rare}}\\
			&     &          &{0.80} &{91} &{449} &{84} & {92.69$\%$} & {Rare}\\
			&     &          &{0.90} &{91} &{100} &{48} & {91.57$\%$} & {Rare}\\\hline\hline
			\textsc{Iris}     &{0.15}& {20$\%$} & {0.60} & {3} & {22} &{7} & {\textbf{96.00$\%$}}&{\textbf{Frequent}}\\
			&     &          & {0.95} & {3} &{6}   &{3} & {95.33$\%$}&{Frequent}\\\hline
			&{0.15}& {20$\%$} & {0.60} & {17} & {32} &{8} & {80.06$\%$}&{Rare}\\
			&     &          & {0.95} & {17} &{7}   &{5} & {80.00$\%$}&{Rare}\\\hline\hline
			&{0.30}& {20$\%$} &{0.60} &{3} &{22} &{7} & {\textbf{96.00$\%$}}&{\textbf{Frequent}}\\
			&     &           &{0.95} &{3} &{6} &{3} & {95.33$\%$}&{Frequent}\\\hline
			&{0.30}& {20$\%$} &{0.60} &{8} &{14} &{4} & {70.00$\%$}&{Rare}\\
			&     &           &{0.95} &{8} &{6} &{3} & {69.33$\%$}&{Rare}\\\hline
			\noalign{\smallskip}\hline
		\end{tabular}
		\caption{Evaluation of the classification accuracy of frequent patterns \textit{vs} rare patterns when  \textit{minconf} varies.}
		\label{minconfCls2}
	\end{table}
\end{landscape}
\subsection{Effect of \textit{minconf} variation}
We note according to the results sketched by Table \ref{minconfCls2}, that for the datasets \textsc{Wine} and \textsc{TicTacToe}, the highest values of the accuracy rate are achieved
with the rare correlated rules. Whereas, for the \textsc{Iris} dataset, the frequent correlated rules performed higher accuracy than the rare ones.
In this regard, we can conclude that for some datasets, the frequent correlated patterns have better informativity than rare ones.
Whereas, for other datasets, rare correlated patterns bring more rich knowledge.
This confirms the beneficial complementarity of our approach in inferring new knowledge from both frequent and rare \textit{correlated} patterns.

In the next section, we present the application of the rare correlated associative rules on intrusion detection data.
\section{Classification of Intrusion Detection Data} \label{se-IDS}
The intrusion detection problem \cite{brahmiYAP10,BrahmiYAP11} is a common problem. In this context, We present, in this section, the experimental evaluation of the correlated classification association rules, previously extracted in Section \ref{chap7se4}, when applied to the \textsc{KDD 99} dataset of intrusion detection data.
\subsection{Description of the \textsc{KDD 99} Dataset} \label{DataDescrip}
In the \textsc{KDD 99} dataset $^\textsc{(}$\footnote{The \textsc{KDD 99} dataset is available at the following link: \textsl{http://kdd.ics.uci.edu/databases/kddcup99/kddcup99.html}.}$^\textsc{)}$, each line or connexion represents a data stream between two defined instants
between a source and a destination, each of them identified by an IP address under a given protocol(TCP, UDP).
Every connection is  labeled either normal or attack and has 41 discrete and continuous attributes that are divided into three groups \cite{farid2010}. The first group of attributes is the basic features of network connection, which include the duration, prototype, service, number of bytes from IP source addresses or from destination IP addresses. The second group of features is composed by the 
content features within a connection suggested by domain knowledge. The third group is  composed by traffic features computed using a two-second time window.


\textsc{KDD 99} defines 38 attacks categories partitioned into four \texttt{Attack} classes, which are \textsc{Dos}, \textsc{Probe}, \textsc{R2L} and \textsc{U2R}, and one \textsc{Normal} class. These categories are described in \cite{NahlaSac2004} and in \cite{farid2010} as follows:\\

$\bullet$ \textbf{Denial of Service Attacks (DOS)}: in which an attacker overwhelms the victim host with a huge number of requests. Such attacks are easy to perform and can cause a shutdown of the host or a significant slow in its performance. Some examples of DOS attack:  Neptune, Smurf, Apache2 and Pod.

$\bullet$ \textbf{User of Root Attacks (U2R)}: in which an attacker or a hacker tries to get the access rights from a normal host in order, for instance, to gain the root access to the system. Some examples of U2R attack: Httptunnel, Perl, Ps, Rootkit.

$\bullet$ \textbf{Remote to User Attacks (R2L)}: in which the intruder tries to exploit the system vulnerabilities in order to control the remote machine through the network as a local user. Some examples of R2L attack: Ftp-write, Imap, Named, Xlock.

$\bullet$ \textbf{Probe}: in which an attacker attempts to gather useful information about machines and services available on the network in order to look for exploits. Some examples of Probe attack: Ipsweep, Mscan, Saint, Nmap.

The \textsc{KDD 99} dataset contains 4, 940, 190 objects in the learning set.
We consider 10$\%$ of the training set in the construction step of the classifier, containing 494, 019 objects. The learning set contains 79.20$\%$ \textsc{(}respectively, 0.83$\%$, 19.65$\%$, 0.22$\%$ and 0.10$\%$\textsc{)} of \textsc{Dos} \textsc{(}respectively, \textsc{Probe}, \textsc{Normal}, \textsc{R2L} and \textsc{U2R}\textsc{)}.

\subsection{Experimentations and Discussion of Obtained Results} \label{RecapXp}
Table \ref{TabNewKDD3} summarizes the obtained results, where \textsc{AR} and \textsc{DR}, respectively, denote ``Association Rule'' and ``Detection Rate'', with Detection Rate =
$\frac{NbrCcCx}{TotalNbrCx}$, with $NbrCcCx$ stands for the number of the correctly classified connections and $TotalNbrCx$ is equal to the whole number of the classified connections,  while \textit{minconf} is the minimum threshold of the confidence measure \cite{Agra94}.

In addition, by ``Construction step'', we mean that the step associated to the extraction of the $\mathcal{RCPR}$ representation while ``Classification step'' represents the step in which the classification association rules are derived starting from $\mathcal{RCPR}$ and applied for detecting intrusions.

We note that the highest value of the detection rate is achieved for the classes \textsc{Normal} and \textsc{Dos}. In fact, this is related to the high number of connections of these two classes. This confirms that our proposed approach presents interesting performances even when applied to voluminous datasets.
We also remark that the detection rate varies from an attack class to another one. In fact, for the \textsc{U2R} class, this rate is relatively low when compared to the others classes.

To sum up, according to Table \ref{TabNewKDD3}, the computational cost varies from one attack class to another one. It is also worth noting that, for all the classes, the construction step is much more time-consuming than the classification step. This can be explained by the fact that the extraction of the $\mathcal{RCPR}$ concise representation is a sophisticate problem.

Furthermore, the results shown by Table \ref{TabCmpNahla} prove that our proposed classifier is more competitive than the decision trees as well as the Bayesian networks \cite{NahlaSac2004}.
In fact, our approach presents better results for the attack classes \textsc{Dos}, \textsc{R2L} and \textsc{U2R} than these two approaches. For the \textsc{Normal} class, the obtained results using our approach are close to those
obtained with the decision trees.
The Bayesian networks based approach presents better detection rate only for the \textsc{Probe} attack class.
We thus deduce that the proposed rare correlated association rules constitute an efficient classification tool when were applied to the intrusion detection in a computer network.
\font\xmplbx = cmbx7. scaled \magstephalf
\begin{sidewaystable}
	\begin{center}
		\begin{tabular}{|l||c|c|c||c|c|c||c|}
			\hline\noalign{\smallskip}
			\textbf{Attack}& \textbf{\textit{minsupp}}&\textbf{\textit{minbond}} &\textbf{\textit{minconf}} &\textbf{$\#$ of generic}&\textbf{$\#$ of generic } & \textbf{$\#$ of generic}&\textbf{CPU Time}\\
			\textbf{class} &  \textbf{\textsc{(}$\%$\textsc{)}} &  & &\textbf{exact}&\textbf{approximate}& \textbf{\textsc{AR}s of}&\textbf{\textsc{(}in seconds\textsc{)}} \\
			& &               &     &\textbf{\textsc{AR}s} &\textbf{\textsc{AR}s}&\textbf{classification}&  \\
			\noalign{\smallskip}
			\hline \textsc{Dos}   & {\xmplbx80}  & {\xmplbx0.95}& {\xmplbx0.90} & {\xmplbx4} & {\xmplbx31}  &{\xmplbx17}&{\xmplbx121}\\
			\textsc{Probe} & {\xmplbx60}  &{\xmplbx0.70} &{\xmplbx0.90} & {\xmplbx232} & {\xmplbx561}  &{\xmplbx15}&{\xmplbx56}\\
			\textsc{Normal}& {\xmplbx85}  &{\xmplbx0.95} &{\xmplbx0.95} & {\xmplbx0} & {\xmplbx10}  &{\xmplbx3}&{\xmplbx408}\\
			\textsc{R2L}   & {\xmplbx80}  &{\xmplbx0.90} &{\xmplbx0.70} & {\xmplbx2}& {\xmplbx368} &{\xmplbx1}&{\xmplbx1, 730}\\
			\textsc{U2R}   & {\xmplbx60}  &{\xmplbx0.75} &{\xmplbx0.75} & {\xmplbx106}& {\xmplbx3}  &{\xmplbx5}&{\xmplbx33}\\
			\hline
		\end{tabular}
	\end{center}
	\caption{Evaluation of the rare correlated association rules for the \textsc{KDD 99} dataset.}\label{TabNewKDD3}
\end{sidewaystable}
\font\xmplbx = cmbx7.5 scaled \magstephalf
\begin{sidewaystable}
	\begin{center}
		\begin{tabular}{|l||c||c|c|}
			\hline\noalign{\smallskip}
			\textbf{Attack class}&\textbf{Rare correlated}     &\textbf{Decision trees}  &\textbf{Bayesian networks}\\
			& \textbf{generic \textsc{AR}s}& \cite{NahlaSac2004}     & \cite{NahlaSac2004}\\
			\hline
			\textsc{Dos}   & \textbf{98.68} &97.24& 96.65\\
			\textsc{Probe} & 70.69 & 77.92&\textbf{88.33}\\
			\textsc{Normal}& {\textbf{100.00}} &99.50&97.68\\
			\textsc{R2L}   & {\textbf{81.52}} &0.52&8.66\\
			\textsc{U2R}   & {\textbf{38.46}} &13.60&11.84\\
			\hline
		\end{tabular}
	\end{center}
	\caption{Comparison between the proposed rare correlated association rules based classifier versus the state of the art approaches.}
	\label{TabCmpNahla}
\end{sidewaystable}

\section{Application of the $\mathcal{RCPR}$ representation on Micro-array gene expression data} \label{SecGene}
We present, in this section, the application of the $\mathcal{RCPR}$ condensed representation of rare correlated patterns on Micro-array gene expression data. In fact, the $\mathcal{RCPR}$ representation
\textsc{(}\textit{cf.} Definition \ref{rmcr} Page \pageref{rmcr}\textsc{)}, is composed by the  
$\mathcal{CRCP}$ set of Closed Rare Correlated Patterns as well as the associated $\mathcal{MRCP}$ set of Minimal Rare Correlated Patterns. From these two sets, we extract the generic rare correlated associated rules, as described in Sub-section \ref{subse-descrip} \textsc{(}\textit{cf.} Page \pageref{subse-descrip}\textsc{)}. The extracted association rules will be then analyzed in order to evaluate the relevance of the obtained biological knowledge.  
\subsection{Our Motivations}
Since many years, gene expression technologies have offered a huge amount of micro-array data by measuring expression levels of thousands of genes under various biological experimental conditions. The micro-array datasets present specific characteristics which is the high density of data. These datasets are in the form of \textsc{(}N x M\textsc{)} matrix with N represents the rows \textsc{(}the conditions or the experiments\textsc{)} and M represents the columns \textsc{(}the genes\textsc{)}. In this regard, the key task in the interpretation of biological knowledge is to identify the differentially expressed genes.  In this respect, we are based on rare correlated patterns in order to identify up and down regulated genes.

Several related works \cite{top-interval,mcrMiner,pasquier2009,goodrule,Article3} were focused on the extraction of frequent patterns and the generation of frequent association rules in order to analyze micro-array data. Our motivation behind the choice of biological data is based on the review of the existing literature that confirms that there is no previous work that addresses the issue of analysis of gene expressions from rare correlated patterns.
Our proposed association-rules based process can be classified as an expression-based interpretation approach for biological associations.
In fact, we are based on gene expression profiles varying under hundreds of biological conditions. 

In what follows, we provide the description of the used micro-array dataset.
\subsection{Description of the Micro-array gene expression data} 
For the application of our approach, we used the breast cancer 2 \textsc{GSE1379} dataset 
$^\textsc{(}$\footnote{The breast cancer dataset is publicly available and downloaded from 
	http://www.ncbi.nlm.nih.gov/geo/query/acc.cgi?acc=GSE1379 .This dataset is submitted on May 2004 and updated on March 2012.}$^{\textsc{)}}$.
The original data is composed by 60 samples and 22, 575 genes.
We present in Table \ref{sample1} a sample of the \textsc{GSE1379} dataset containing only 5 genes on columns and 5 samples on rows.
\begin{table}[h]
	\begin{center}
		\footnotesize{
			\begin{tabular}{|c|c|c|c|c|c|c|c|}
				\hline & \texttt{id-G1}  & \texttt{id-G2}  & \texttt{id-G3}  & \texttt{id-G4} & \texttt{id-G5}& ...  \\
				\hline\hline
				GSM22449 & -1.3361553&	0.3867403&	-2.0288643&	-1.9541923	&-2.0088713      & ...   \\
				\hline
				GSM22450 &    -1.3361553 &	0.3867403& 	-2.0288643& 	-1.9541923 &	-2.0088713	& ... \\
				\hline
				GSM22451 & -1.3333233	 & -2.0482593 &	-2.0577023	& -1.6493243 &	-2.0727303   & ...\\
				\hline
				GSM22452 &    -1.6211983 &	-1.3905463 &	-1.2612803	& -1.4602183	 & -1.4401533  & ...\\
				\hline
				GSM22453 & -0.0878543 &	-0.1720993 & 	-0.2629703	& -0.0816163  & -1.6493243 &  ... \\
				\hline
				...& ...  & ... &...    &    ...     & ...& \\
				\hline	
		\end{tabular}}
	\end{center}
	\caption{An example of Micro-array data.}\label{sample1}
\end{table}
Table \ref{GenName} illustrates examples of some relevant genes of the \textsc{GSE1379} dataset enriched with their description.

\begin{sidewaystable}
	\begin{center}
		\footnotesize{
			\begin{tabular}{|l|l|l|}
				\hline 
				\textbf{Gene-Id} &  \textbf{Gene-name}  & \textbf{Description} \\
				\hline
				
				4048	& HOXB13	 & Homeobox B13: Sequence-specific transcription factor which is part of a developmental regulatory system.  \\
				\hline
				4753 &	CHDH&	Choline dehydrogenase.  \\
				\hline
				13983 &	ESR1 &	Estrogen receptor 1: Nuclear hormone receptor. It is involved in the regulation of eukaryotic gene expression. \\
				\hline
				14944 &	CKAP4&	Cytoskeleton-associated protein 4. \\
				\hline
				16227&	ABCC11	&ATP-binding cassette, sub-family C \textsc{(}CFTR/MRP\textsc{)}, member 11.\\  	
				\hline
				19980 &	IL17BR	& Interleukin 17B: Receptor for the pro-inflammatory cytokines IL17B and IL17E.  \\
				\hline
				20975 & 	ZNF197	& Zinc finger protein 197. \\
				\hline
		\end{tabular}}
	\end{center}
	\caption{Description of a sample of genes of the \textsc{GSE1379} dataset.}\label{GenName}
\end{sidewaystable}

\subsection{The Discretization process}
The discretization aimed to transform the continuous data into discrete data. We performed a discretization process
based on the R.Basic package of the R statistical framework $^\textsc{(}$\footnote{The R Project for Statistical Computing is downloaded from https://www.r-project.org.}$^{\textsc{)}}$. First of all, we apply the Z-Normalization 
\cite{zNormalization} over the whole dataset in order to 
transform the initial data distribution to a normal distribution.
The second step consists in determining the over-expressed cutoff $O_{c}$ and the under-expressed cutoff $U_{c}$.
In fact, according to the Z-Normal distribution table, when considering a confidence level $1-\alpha$ equal to 95$\%$, we have: 

$\bullet$ The over-expressed cutoff $O_{c}$ = Z($\alpha$/2) = 1.96 

$\bullet$ The under-expressed cutoff $U_{c}$ = -Z($\alpha$/2) = -1.96

Thus, we have for each sample $i$ and for each gene $j$, $V_{ij}$ corresponds to the value of the gene expression 
$j$ within the sample $i$. The $V_{ij}$ expression is evaluated as follows:
\begin{itemize}
	\item if $V_{ij}$ $\leq$  $U_{c}$  then $V_{ij}$ is under-expressed $\downarrow$
	\item  if $V_{ij}$  $\geq$  $O_{c}$ then $V_{ij}$ is over-expressed  $\uparrow$
	\item if $U_{c}$ < $V_{ij}$ > $O_{c}$ then $V_{ij}$ is unexpressed
\end{itemize}

We present in Table \ref{DistData}, a sample of the discretized data, where the over-expressed genes are 
referenced by the value of 1 whereas the under-expressed genes are referenced by the value of 0. The `$-$` symbol represent unexpressed gene expressions which are not relevant for our analysis.
\begin{table}[h]
	\begin{center}
		\footnotesize{
			\begin{tabular}{|c|c|c|c|c|c|c|c|}
				\hline & \texttt{id-G1}  & \texttt{id-G2}  & \texttt{id-G3}  & \texttt{id-G4} & \texttt{id-G5}& ...  \\
				\hline\hline
				GSM22449 & 1&	$-$ &	1&	0	&0      & ...   \\
				\hline
				GSM22450 &    0 &	1& 	1 & 	0 &	1	& ... \\
				\hline
				GSM22451 & 0	 & $-$ &	0	& 1 &	1  & ...\\
				\hline
				GSM22452 &    1 &	$-$ &	1	& 0	 & 1  & ...\\
				\hline
				GSM22453 & $-$ &	0 & 	0	& 0  & 0 &  ... \\
				\hline
				...& ...  & ... &...    &    ...     & ...& \\
				\hline	
		\end{tabular}}
	\end{center}
	\caption{An example of the discretized Micro-array data.}\label{DistData}
\end{table}

After the discretization process, we apply a substitution function in order to transform the discretized gene expression values in the adequate input format for the mining process.
Consequently, we apply our substitution function $\theta$ as follows:
\begin{itemize}
	\item if $V_{ij}$ is over-expressed $\uparrow$  then  $V_{ij}$ $\leftarrow$ $Id_{j}$, with 
	$Id_{j}$ corresponds to the unique identifier of gene $j$
	\item if $V_{ij}$ is under-expressed  $\downarrow$   then  $V_{ij}$ $\leftarrow$ `$Id_{j}$ + $\arrowvert$M$\arrowvert$`, with $Id_{j}$ is the unique identifier of gene $j$ and $\arrowvert$M$\arrowvert$ corresponds to the number of the distinct genes, $\arrowvert$M$\arrowvert$ = 22, 575 in our tested dataset.
\end{itemize}

We present in Table \ref{SubsData}, a sample of the final substituted data. This sample is in the 
adequate input format of our mining algorithm \textsc{Opt-Gmjp}.    

\begin{table}[h]
	\begin{center}
		\footnotesize{
			\begin{tabular}{|c|c|c|c|c|c|c|c|}
				\hline & \texttt{G1}  & \texttt{G2}  & \texttt{G3}  & \texttt{G4} & \texttt{G5}& ...  \\
				\hline\hline
				S1 &        1&	 &	3&	22 579	&22 580      & ...   \\
				\hline
				S2 &    22 576 &	22 577& 	3 & 	22 579 &	5	& ... \\
				\hline
				S3 & 22 576	 &  &	22 578	& 4 &	5  & ...\\
				\hline
				S4 &    1 &	 &	3	& 	22 579  & 5 & ...\\
				\hline
				S5 &  &	22 577 & 	22 578	& 22 579  & 22 580 &  ... \\
				\hline
				...& ...  & ... &...    &    ...     & ... &\\
				\hline	
		\end{tabular}}
	\end{center}
	\caption{Discretized values substituted by their identifiers.}\label{SubsData}
\end{table}

\subsection{Experimental results}
We conducted several experiments over the \textsc{GSE1379} dataset in order to extract the most relevant exact and approximate association rules. The \textsc{GSE1379} dataset was preprocessed with the 
\textsc{GEO2R} tool  in order to identify genes that are differentially expressed across experimental conditions. The Results obtained by the \textsc{GEO2R} tool are presented as a table of genes ordered by significance. Thus, we maintain the 550 most relevant genes from 22, 575 initial genes. 
For these experiments, the \textsc{Opt-Gmjp} algorithm was applied to the \textsc{GSE1379} with $|\mathcal{T}|$ = 60 and with $|\mathcal{I}|$ = 1, 100 distinct items values.
\begin{table}
	\hspace{-1.5cm}
	\begin{tabular}{|rrrr||rrrc|}
		\hline
		\noalign{\smallskip}
		\textbf{\textit{minsupp}} & \textbf{\textit{minbond}}&  $|\mathcal{MRCP}|$ &$|\mathcal{CRCP}|$ &   \textbf{\textit{minconf}}  & \textbf{$\#$ Exact} & \textbf{$\#$ Approximate} & \textbf{CPU} \\
		&  & &   & & \textbf{Rules} & \textbf{Rules} &  \textbf{Time \textsc{(}sec\textsc{)}} \\
		\noalign{\smallskip}\hline\noalign{\smallskip}
		{33$\%$}& {0.30}  & {120} & {146} & {0.70}  &{19} & {3} & {0.0405} \\
		&          &    &          & {0.50}  &{19} & {17} &{0.0405}\\
		&          &    &          & {0.30}   &{19} & {20} & {0.0405}\\\hline

		{50$\%$}  & {0.30}  & {157} & {244} & {0.70} &{26} & {77} & {0.0754}\\
		&         &       &       & {0.50}      &{26} & {128} &{0.0754}\\
		&          &       &      & {0.30}   &{26} & {134} &{0.0754}\\\hline 
		{50$\%$}  & {0.50}  & {79} & {72} & {0.30} &{7} & {6} & {0.0595}\\\hline   
		{50$\%$}  & {0.70}  & {59} & {56} & {0.30} &{3} & {0} & {0.0463}\\  
		\noalign{\smallskip}\hline
	\end{tabular}
	\caption{Execution Times and number of extracted association rules.}
	\label{geneAssRule1}
\end{table}

These experiments were conducted in order to assess the scalability of our \textsc{Opt-Gmjp} algorithm when applied to very dense 
biological dataset and to evaluate the impact of varying the \textit{minsupp}, the \textit{minbond} and the \textit{minconf} thresholds on the number of the extracted association rules. We report in Table \ref{geneAssRule1}  the execution times as well as the number of the approximate and exact extracted association rules. We can draw theses conclusions:
\begin{itemize}
	\item The sizes of the $\mathcal{MRCP}$ set of minimal correlated rare patterns as well as that of the $\mathcal{CRCP}$ set of closed rare correlated patterns  depends only on the variation of \textit{minsupp} and  \textit{minbond} thresholds. We deduce that, $|\mathcal{MRCP}|$ and $|\mathcal{CRCP}|$  decrease when increasing \textit{minbond} from {0.30} to {0.70}.
	\item The execution times are not affected by the variation of \textsc{minconf} threshold. In fact, the reported execution times corresponds to the CPU-time needed for extracting the $\mathcal{RCPR}$ representation. The CPU-time needed for the derivation of the association rules is negligible in all the performed experiments.
	\item  The number of the extracted association rules decreases while increasing the \textit{minconf} threshold. For example, for  \textit{minsupp} =  {50$\%$} , \textit{minbond} = {0.30} and \textit{minconf} = 0.30, we have  $|Approximate-Rules|$ = 134, while for \textit{minconf} = 0.70, $|Approximate-Rules|$ = 77. It's obviously that the number of the exact rules is insensitive to the variation of the \textit{minconf} threshold since the confidence of exact rules is equal to {100$\%$}.
	\item The increase of the \textit{minbond} threshold value from 0.30 to 0.70, induce a decrease in the size of the $\mathcal{MRCP}$ and $\mathcal{CRCP}$ sets.
	This reflects that the used dataset do not present important correlation degree among the items. The items are dispersed in the universe due to the low-level of co-expression of the mined genes.    
\end{itemize} 
\subsection{Biological significance of Extracted  Association rules}
Table \ref{ExpAR} shows different examples of association rules extracted by a dedicated procedure previously described in sub-section \ref{subse-descrip}. In Table \ref{ExpAR}, supports are expressed in number of transactions while confidence are given in percentages. 
The association rules show groups of genes that are over-expressed or under-expressed in a set of conditions.

To determine the functional relationship among the obtained gene sets, we used the \textsc{STRING 10} 
$^\textsc{(}$\footnote{\textsc{STRING} stands for the \textbf{S}earch \textbf{T}ool for the \textbf{R}etrieval of \textbf{IN}teracting \textbf{G}enes/Proteins and is publicly available at http://string-db.org.}$^{\textsc{)}}$
resource \cite{string}
which is a database of known and predicted protein-protein interaction. 

In this regard, the gene sets obtained within the association rules were uploaded into \textsc{STRING} and the following prediction methods were employed: co-expression, co-occurrence with a medium confidence score equal to 40$\%$.
This analysis shows the interactions among the gene sets as shown in Figures  \ref{fig-string1} and \ref{fig-string2}.
This finding support the hypothesis that the returned gene sets thank to our rare correlated association rules, show an important degree of biological interrelatedness.

In figure \ref{fig-string2}, we highlight just the most relevant genes reported in the biological literature and related to the analysis of breast cancer \cite{breast1}. These genes are: HOXB13,  ABCC11, CHDH, ESR1 and IL17BR \cite{Article3}.

\begin{center}
	\begin{figure}[htbp]
		\centering
		\includegraphics[scale =0.35]{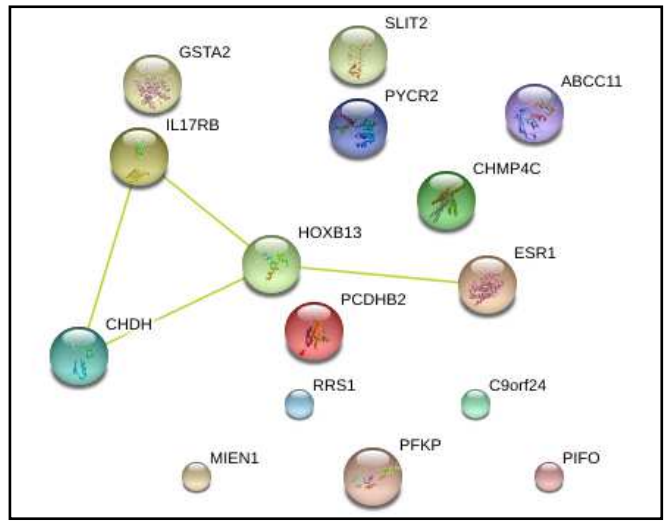}
		\caption{The \textsc{STRING} compact network view.} \label{fig-string2}
	\end{figure}
\end{center}

\begin{table}[h]
	\hspace{-0.5cm}
	\footnotesize{
		\begin{tabular}{|l|l|l|l|l|}
			\hline 
			\textbf{Rule} &  \textbf{Antecedent}  & \textbf{Conclusion} & \textbf{Support} & \textbf{Confidence} \\
			\hline
			0 &ESR1 $\uparrow$ &  GSTA2 $\uparrow$	& 1 &	100$\%$ \\
			\hline
			1& 	RRS1 $\downarrow$ and    ABCC11 $\uparrow$ & CRISPLD2 $\uparrow$ and    CHDH $\uparrow$ &	24	& 96$\%$ \\
			\hline
			2	& RRS1 $\downarrow$ and   HOXB13 $\downarrow$ & CRISPLD2 $\uparrow$  &	18	& 94$\%$  \\
			\hline
			3&	BLOC1S6 $\downarrow$ and  HOXB13 $\downarrow$ &  CRISPLD2 $\uparrow$   and  CHDH $\uparrow$ &	24	 &92$\%$ \\
			\hline
			4&	HOXB13 $\downarrow$  and ABCC11 $\uparrow$ &  BLOC1S6 $\downarrow$ and  CHDH $\uparrow$ &	17	& 77$\%$	\\
			\hline
			5 & INSIG1 $\downarrow$ &  CRISPLD2 $\uparrow$,  IRAK3 $\uparrow$  and  ABCC11 $\uparrow$ & 19&	 70$\%$ \\
			\hline	
			6 & IL17BR $\uparrow$ &  PFKP $\downarrow$ &	9	& 52$\%$ 	\\
			\hline
			7 &	HOXB13 $\uparrow$ &    NDUFAF2 $\downarrow$	&1	 & 50$\%$\\
			\hline
			8&	IL17BR $\uparrow$ &  C9orf24 $\uparrow$	& 8	&  47$\%$	\\  	
			\hline
	\end{tabular}}
	\caption{Association rules: Expression levels $\Rightarrow$ Expression levels.}\label{ExpAR}
\end{table}

According to Table \ref{ExpAR}, Rule 0 reflects that the estrogen receptor 1 which is a Nuclear hormone receptor and is expressed by ESR1 when it is over-expressed in this experiment induces an over-expression of the glutathione S-transferase alpha 2 traduced by gene GSTA2.
Rule 3 highlights that if the HOXB13 and the BLOC1S6 genes are down-expressed then the CHDH and the CRISPLD2 genes are over-expressed.
In fact, the HOXB13 gene refers to Sequence-specific transcription factor which is part of a developmental regulatory system that provides cells with specific positional identities on the anterior-posterior axis. While, the CRISPLD2 gene is cysteine-rich secretory protein LCCL domain and the CHDH gene expresses the choline dehydrogenase. Rule 6 present an interesting relation between the 
IL17BR gene, reflecting the interleukin 17 receptor B and playing a role in controlling the growth and  differentiation of hematopoietic cells, and the PFKP  gene corresponding to phospho-fructokinase, platelet. The PFKP  gene catalyzes the phosphorylation of fructose-6-phosphate (F6P) by ATP to generate fructose-1, 6-bisphosphate (FBP) and ADP.

Almost of the obtained association rules highlights important relationship of the HOXB13 and the IL17BR genes. In fact, the analysis of these two genes 
expression may be useful for identifying patients appropriate for alternative therapeutic regimens in early-stage breast cancer \cite{Article3}.
In summary, we conclude that the diverse obtained rare correlated association-rules reveals a variety of relationship between up and down gene-expression which proves that breast cancer is an interesting biologically heterogeneous research field. Thus to deduce that rare correlated patterns present good results  when applied to the context of biological data since they are able to reveal hidden and surprising relations among genes properties. 

\section{Conclusion} \label{ConcChap7}
This chapter was dedicated to the description of the associative classification process based on the correlated patterns. For this purpose, we started by presenting the framework of association rules extraction, we clarify the properties of the generic bases of association rules. We continued with the detailed description and presentation of the application of both frequent correlated and rare correlated patterns within the classification of some UCI benchmark datasets. 
We equally present the application of rare correlated patterns in the classification of intrusion detection data from the \textsc{KDD 99} dataset. The effectiveness of the proposed classification method has been experimentally proved.
The chapter was concluded with the application of the $\mathcal{RCPR}$ representation on the extraction of biologically relevant associations among Micro-array gene expression data. A better classification accuracy may be achieved while 
thinking about missing-values treatment \cite{OthmanY06}.




\part{Conclusion}
\chapter{Conclusion and Perspectives}\label{ch_conc}
\markboth{Conclusion and Perspectives}{Conclusion and Perspectives}
\section{Conclusion}

In this thesis, we were mainly interested to two complementary classes of patterns namely rare correlated patterns and frequent correlated patterns according to the \textit{bond} correlation measure.
In fact, the $\mathcal{FCP}$ set of frequent correlated patterns result from the intersection of the set of frequent patterns and the set of correlated patterns. The $\mathcal{FCP}$ set is then the result of the conjunction  
of two anti-monotonic constraints of frequency and of correlation. Consequently this $\mathcal{FCP}$ set induces an order ideal on the itemset lattice. Nevertheless, the $\mathcal{RCP}$ set of rare correlated patterns result from the conjunction of two constraints of distinct types namely the monotonic constraint of rarity and the anti-monotonic constraint of correlation. Thus, the localization of the  $\mathcal{RCP}$ set is more difficult
and the extraction process is more costly. This characteristic constitute one of the challenges to deal with through this thesis.

This thesis report was partitioned into four different parts. The first part was dedicated to the review of correlated patterns mining. In this regard, we started the first chapter of this part by introducing the basic notions related to the itemset search space, to itemset extraction. We defined the two distinct categories of constraints. We introduced equally the environment of Formal Concept Analysis FCA which offer the basis for the proposition of our approaches, specifically the notions of Closure Operator, Minimal Generator, Closed Pattern, Equivalence class and Condensed representation of a set of patterns. Thereafter, we studied in the second chapter of this first part,  the state of the art approaches dealing with correlated patterns mining. Our study covers the frequent correlated patterns mining, the rare correlated patterns as well as the approaches focusing on condensed representations of correlated itemsets. 

The second part was dedicated to the presentation of our approaches. The first chapter of this part was devoted to the characterization of both frequent correlated and rare correlated patterns and the introduction of their associated condensed representations. We deeply defined the properties of the $f_{bond}$
closure operator associated to the \textit{bond} measure and we describe the structural specificities of the induces equivalence classes. In fact, the condensed representations associated to the $\mathcal{RCP}$ set of rare correlated are composed by the union of the closed correlated rare patterns and their associated minimal generators. Nevertheless, for the case of frequent correlated patterns, the  closed correlated frequent patterns constitute a condensed concise representation of the $\mathcal{FCP}$ set.  In the second chapter, we focused on the presentation of our \textsc{Gmjp} extraction approach. In fact, 
\textsc{Gmjp} is the first approach to mine \textit{bond} correlated patterns in a generic way \textsc{(}i.e., with two types of constraints: anti-monotonic constraint of frequency and monotonic constraint of rarity\textsc{)}. Our mining approach was based on the key notion of bitsets codification that supports efficient correlated patterns computation thanks to an adequate condensed representation of patterns.
The deeply description of the whole steps of \textsc{Gmjp} as well as the theoretical complexity approximation and a running example were equally detailed. This fifth chapter was concluded by the algorithms of interrogation and of regeneration of the condensed representation associated to rare correlated patterns.

The third part of this report was dedicated to the experimental validation of our \textsc{Gmjp} 
as well as the presentation and evaluation of the associative-classification process.
In the first chapter of this third part we focused on the experimental evaluation of \textsc{Gmjp}.
The evaluation process was based on two main axes, the first is related to the compactness rates of the condensed representations  while the second axe concerns the running time. We equally proposed an optimized version of \textsc{Gmjp} which present much better performance than \textsc{Gmjp} over different benchmark datasets. The two main keys which constitute the thrust of the improved version: (\textit{i}) only one scan of the database is performed to build the new transformed dataset; (\textit{ii}) it offers a resolution of the problem of handling both monotonic and anti-monotonic constraints within a unique mining process. 
In fact, opposite constraint mining is classified as an NP-Hard problem \cite{boley2009}. But, our goal was optimally achieved without relying on the border's extraction. This constitute a strong added-value to \textsc{Gmjp}, since many approaches are based on border's identification in order to extract such difficult set of patterns.

In the second chapter of this third part, we presented the classification process based on correlated patterns.
Since the classification process that we proposed was based on associative rules, thus we started the chapter by presenting the framework of association rules extraction, we clarified the properties of the generic bases of association rules. We continued with the detailed presentation of the application of both frequent correlated and rare correlated patterns within the classification of some UCI benchmark datasets. In addition, we reported in this chapter  the application of rare correlated patterns in the classification of intrusion detection data from the \textsc{KDD 99} dataset. The obtained results showed the usefulness of our proposed classification method over four different intrusion classes.  We concluded the chapter with the application of rare correlated associative rules on Micro-array gene expression data. The obtained rules helped to identify potential relations among up and down regulated gene expressions related to Breast Cancer.

The fourth and final part concluded the thesis report.

\section{Perspectives}
The obtained results in this thesis opens many perspectives from which we quote:

\bigskip

\checkmark The extraction of generalized association rules starting from rare correlated patterns also from frequent correlated patterns.  In addition, we plan to extend our approach to other correlation measures \cite{Kimpkdd2011,borgelt,surana2010,Omie03} through classifying them into classes of measures sharing the same properties. An important direction is to propose a generic way allowing the extraction of the sets of frequent correlated patterns and rare correlated patterns as well as their associated concise representations. Pieces of new knowledge in the form of exact or approximate correlated generalized association rules can then be derived.

\checkmark The extension of the extraction of correlated patterns to the extraction of both frequent and rare sequential correlated patterns. A promoting area for applying sequential patterns is: opinion mining. In fact, Opinion Mining is an important research area \cite{Ohana2011} which is based on the extraction of opinions and the sentiment analysis from text data \textsc{(}Text Mining\textsc{)}. Opinion Mining is a fruitful field since it is concerned with many real life application fields such as: Financial analysis, market estimation,  customer behavior detection.
In fact, the evaluation of new products and services nearby customers is based on the comments and advices of web visitors. Consequently, the derivation of association rules and their application to opinion mining \cite{Jindal2010} is a potentially interesting research axe.

\checkmark Another Fruitful perspective consists in addressing the issue of correlated patterns mining from big datasets. In fact, big data mining is a new challenging task since computational requirements are difficult to 
provide. An interesting solution is to exploit parallel frameworks, such as MapReduce \cite{mapred-2012} that offer the opportunity to make powerful computing and storage. Consequently, mining condensed representations of correlated patterns from big real life datasets thank to the MapReduce environment is an up to date challenging mining task.

\section{Publication List}

$\bullet$ \textbf{\large{Journals}} \\

[1] \underline{Bouasker Souad}, Hamrouni Tarek, Ben Yahia Sadok. Motifs Corrélés rares: Caractérisation et nouvelles  représentations concises exactes . Appeared in the Revue of New Information Technologies RNTI: Quality of Data and Knowledge: Measure and Evaluate the Quality of Data and Knowledge, \textsc{(}MQDC 2012\textsc{)}, pages 89-116.\textbf{[indexed DBLP]}.

\rule{\linewidth}{.5pt}\\

[2] \underline{Bouasker Souad}, Hamrouni Tarek, Ben Yahia Sadok. Efficient Mining of New Concise representations of Rare Correlated Patterns. Appeared in the IDA journal  'Intelligent Data Analysis' 2015, Volume 19, pages 359-390. \textbf{\textsc{(}Impact factor in 2014 = 0.50\textsc{)}}.
\bigskip

$\bullet$  \textbf{\large{International Conferences}} \\

[3] \underline{Bouasker Souad}, Hamrouni Tarek, Ben Yahia Sadok. Algorithmes d'extraction et d'interrogation d'une repr\'esentation concise exacte des motifs corr\'el\'es rares. In proceedings of the 12th international francophone
conference on extraction and managing knowledges \textsc{(}EGC 2012\textsc{)}, Bordeaux, France, pages 225-230. \textbf{[indexed DBLP], rank C}.

\rule{\linewidth}{.5pt}\\

[4] \underline{Bouasker Souad}, Hamrouni Tarek, Ben Yahia Sadok. New Exact Concise Representation of Rare Correlated Patterns: Application to Intrusion Detection. In proceedings of the 16th  Pacific Asia conference \textsc{(}PAKDD 2012\textsc{)}, Kuala Lumpur, Malaysia, pages  61-72. \textbf{[indexed IEEE, DBLP], rank A}.

\rule{\linewidth}{.5pt}\\

[5] \underline{Bouasker Souad}, Ben Yahia Sadok. Inferring New Knowledge from Concise Representations of both Frequent and Rare Jaccard Itemsets. In proceedings of the 24th  International conference of Database and Expert Systems Applications \textsc{(}DEXA 2013\textsc{)}, Prague, Check Republic, pages 109-123. \textbf{[indexed IEEE, DBLP], rank B}.

\rule{\linewidth}{.5pt}\\

[6] \underline{Bouasker Souad}, Ben Yahia Sadok. Key Correlation Mining by Simultaneous Monotone and Anti-monotone Constraints Checking. In proceedings of the 30th ACM Symposium on Applied Computing \textsc{(}SAC 2015\textsc{)},  Salamanca, Spain. \textbf{[indexed ACM, IEEE, DBLP], rank B}.

\newpage
\addcontentsline{toc}{chapter}{Bibliography}
\bibliographystyle{apalike}


\end{document}